\documentclass{article}

\usepackage[utf8]{inputenc} 
\usepackage[T1]{fontenc} 
\usepackage[margin=1in]{geometry}
\usepackage{float}
\usepackage[colorlinks=true,linkcolor=black,citecolor=black]{hyperref}       
\usepackage{url}            
\usepackage{booktabs}       
\usepackage{amsfonts}       
\usepackage{nicefrac}       
\usepackage[expansion=false]{microtype}      
\usepackage{lipsum}
\usepackage{graphicx}
\graphicspath{ {./images/} }

\usepackage{natbib}
\usepackage{amsmath,amssymb,amsthm}

\usepackage{xcolor}
\usepackage{enumitem}
\usepackage{setspace}

\allowdisplaybreaks

\usepackage{xspace}

\newcommand{\Pb}{\mathbb{P}}

\newcommand{\E}{\mathbb{E}}

\newtheorem{theorem}{Theorem}[section]

\newtheorem{proposition}[theorem]{Proposition}

\numberwithin{equation}{section}

\title{Gaussian-efficient testing by betting on the mean of bounded data}

\author{
 Diego Martinez-Taboada$^{1}$ and Aaditya Ramdas$^{2}$ \\
  $^1$Department of Statistics \& Data Science, Carnegie Mellon University\\
$^2$Department of Statistics, Stanford University\\
  \texttt{diegomar@andrew.cmu.edu, aramdas@stanford.edu} } 
  
\begin{document}
\maketitle
\begin{abstract}
Given $[0,1]$-valued random variables $X_1,\dots,X_n$ such that $\E[X_i | X_1,\dots,X_{i-1}]= \mu$ for all $i$, we propose a new nonasymptotic confidence interval for $\mu$ that is obtained by inverting terminal e-values generated by a novel betting strategy. When the data are iid,  its limiting width matches that of the central limit theorem (``Gaussian-efficient''), finally surpassing the inefficient limits of  previous betting intervals.   
Our main conceptual advance involves designing betting fractions that track the conditional rejection probability of the
most powerful terminal test in a limiting Gaussian experiment.  When one
predictable variance estimator is shared across candidate means, the
deterministic inversion is an interval for every data sequence and its two
endpoints can be found easily.  
The width can be improved further with external randomization. In simulations, our method yields the tightest intervals to
date; for every distribution tested and all sufficiently large $n$, our
deterministic version beats STaR-Bets and is competitive with Gaffke, while
the randomized improvement beats both.  It thus combines finite-sample validity under
martingale dependence, easy endpoint computation, Gaussian-efficient
inference for iid data, and excellent empirical performance.
We further combine GE-betting with a mixture of constant-fraction bets to
attain both Gaussian and large-deviation efficiency under iid sampling
while retaining martingale validity.
We also extend the construction and its efficiency theory to sampling
without replacement, where it again achieves state-of-the-art empirical performance.
 
\end{abstract}

\section{Introduction}
\label{sec:introduction}

Constructing a confidence interval for the mean of bounded data is one of the
most basic problems in statistical inference.  Boundedness permits honest
finite-sample inference without a parametric model, but a useful interval
should also adapt to the observed variance, be sharp at realistic sample
sizes, and approach the classical central limit theorem interval when the sample is large (``Gaussian-efficient'').

Classical concentration inequalities provide distribution-free control
\citep{bennett1962probability,hoeffding1963probability,azuma1967weighted,
freedman1975tail,delapena1999general,bercu2008exponential,
bentkus2003inequality,pinelis2006binomial,bentkus2006domination,
kuchibhotla2024missing}. Sharper approaches include empirical-Bernstein,
CDF-envelope, small-sample, and finite-sample Gaussian constructions
\citep{maurer2009empirical,anderson1969confidence,fishman1991confidence,
rosenblum2009confidence,austern2022efficient}.  

These developments make
clear that validity and sharpness are distinct design goals.  Gaffke's
Dirichlet-randomization interval was the previous empirical state of the art
among finite-sample confidence intervals for bounded data \citep{gaffke2005three,
learnedmiller2020new,vlassis2026exact,ming2026gaffke}. However, Gaffke's interval is not valid under martingale dependence, and STaR-Bets~\citep{voracek2025star} is the empirical state-of-the-art in this setting, despite not being proven to strictly improve the width of other empirical-Bernstein constructions~\citep{waudby2024estimating}. We will extensively compare to these two methods in particular.

\paragraph{Three desiderata.} At a known horizon \(n\), we seek one two-sided \(1-\delta\) confidence
interval \(\mathcal I_n\) satisfying three requirements.  The first is
\textit{martingale validity}: whenever \(X_i\in[0,1]\) and
\(\E(X_i\mid\mathcal F_{i-1})=\mu\), the interval has the finite-sample
guarantee \(\Pb\{\mu\in\mathcal I_n\}\geq1-\delta\).  Independence is not
required, and conditional variances and higher moments may change with the
history.  The second is \textit{iid efficiency}: if the observations are
independent and identically distributed (iid) with mean \(\mu\in(0,1)\) and variance \(\sigma^2>0\), the width
attains the Gaussian benchmark
\begin{equation}\label{eq:intro-gaussian-efficiency}
  \operatorname{len}(\mathcal I_n)
  =\frac{2\sigma z_{1-\delta/2}}{\sqrt n}+o(n^{-1/2})
  \qquad\text{almost surely}.
\end{equation}
The third is \textit{efficient computation}: for every observed data sequence,
the accepted candidate means form one interval.  The
interval can therefore be computed by finding its two endpoints rather than
searching for an arbitrary subset of \([0,1]\).  These requirements concern
different properties: the conditional-mean model supplies martingale validity,
the iid submodel supplies the sharp local benchmark, and interval geometry
is a deterministic property of the update for each fixed data sequence.

Testing by betting gives a well-known route to the first requirement
\citep{waudby2024estimating}.
For a candidate mean \(m\), predictable betting fractions generate a
nonnegative test supermartingale under the corresponding null; its terminal
value is an e-value.  Markov's or Ville's inequality then controls rejection,
and inversion over \(m\) produces a confidence set;  this outline is now well known.
Write \(\Phi\), \(\phi\), and \(z_q=\Phi^{-1}(q)\) for the standard normal
cdf, density, and quantile.  Our new
ingredient is the use of an idealized hypothetical Gaussian experiment in order to derive the optimal betting strategy (which we refer to as \textit{Gaussian-efficient betting}, or \textit{GE-betting}), and then mimic that in finite samples. Define 
\[
 \psi(p)=\frac{\phi\{\Phi^{-1}(p)\}}p,\qquad 0<p<1.
\]
By convention, $\Phi^{-1}(0)=-\infty,\Phi^{-1}(1)=\infty$, $\psi(0)=\infty$, $\psi(1)=0$ (see Figure~\ref{fig:efficient-betting-function}).
If \(p=(\delta/2)K\) is current wealth divided by the rejection threshold,
\(t\in[0,1)\)
is the fraction of the horizon elapsed, and \(v\) is the variance per
observation in an idealized Gaussian design calculation, the optimal fraction of current
wealth bet per unit change in the
centered \(n^{-1/2}\)-scaled partial sum is
\begin{equation}\label{eq:intro-efficient-betting-rule}
 L(p,t,v)=\frac{\psi(p)}{\sqrt{(1-t)v}}.
\end{equation}
Theorem~\ref{thm:efficient-betting-efficiency} proves that the discrete
implementation below is martingale-valid, its inversion is an interval on
every sample path, and its width satisfies
\eqref{eq:intro-gaussian-efficiency} under iid sampling.  The same update also
permits optional external randomization of
the terminal threshold, which we include explicitly below.
Section~\ref{sec:optimal-gaussian-evalue} derives the following rule from the most
powerful fixed-horizon Gaussian test.

\paragraph{The proposed confidence interval.}
Discretizing the preceding strategy at the known horizon \(n\) gives the
following GE-betting implementable update.  Compute the single mean and variance
estimates, for \(t=0,\ldots,n-1\),
\begin{equation}\label{eq:regularized-variance-estimate}
 \widehat\mu_t=\frac{1/2+\sum_{j=1}^tX_j}{t+1},
 \qquad
 \widehat v_t=\frac{1/4+\sum_{j=1}^t
      (X_j-\widehat\mu_{j-1})^2}{t+1}.
\end{equation}
These estimates are available before \(X_{t+1}\), and the same
\(\widehat v_t\) is used for every candidate and both wealth sequences.  For
terminal calibration, draw
\(U_+,U_- \sim \operatorname{Unif}(0,1)\), independently
of the data, and hold this one pair fixed over every candidate mean (they could be drawn iid, or coupled, for example set to be equal or to sum to one).  If
external randomization is not desirable, simply set \(U_+=U_-=1\).  For each
\(m\in[0,1]\), initialize \(K_0^+(m)=K_0^-(m)=1\).  The upper-tail
process uses the increment \(X_i-m\), and the lower-tail process uses
\(m-X_i\).  Fix \(c\in(0,1]\); our default is \(c=1\).  At round
\(i=1,\ldots,n\), for each process
whose wealth lies strictly between zero and the rejection threshold
\(2/\delta\), define
\begin{align}
 \ell_{i,n}^+(m)
 &=\min\!\left\{
   \frac{\psi\{(\delta/2)K_{i-1}^+(m)\}}
        {\sqrt{(n-i+1)\widehat v_{i-1}}},\frac{c}m
 \right\}, \quad
 \ell_{i,n}^-(m)
 =\min\!\left\{
   \frac{\psi\{(\delta/2)K_{i-1}^-(m)\}}
        {\sqrt{(n-i+1)\widehat v_{i-1}}},\frac{c}{1-m}
 \right\},
 \label{eq:efficient-one-sided-fractions}
\end{align}
and update
\begin{align}
 K_i^+(m)
 &=\min\!\left\{\frac2\delta,
   K_{i-1}^+(m)\{1+\ell_{i,n}^+(m)(X_i-m)\}\right\},\quad
 K_i^-(m)
 =\min\!\left\{\frac2\delta,
   K_{i-1}^-(m)\{1+\ell_{i,n}^-(m)(m-X_i)\}\right\}.
 \label{eq:efficient-one-sided-updates}
\end{align}
At \(m=0\), define \(c/m=+\infty\), so the upper-tail cap is absent; at
\(m=1\), define \(c/(1-m)=+\infty\), so the lower-tail cap is absent.
Once a wealth reaches zero or \(2/\delta\), leave it unchanged.  The proposed
confidence interval for the chosen calibration pair is
\begin{equation}\label{eq:efficient-betting-interval}
 \mathcal I_n(U_+,U_-)
 =\left\{m\in[0,1]:K_n^+(m)<\frac{2U_+}{\delta},\
                    K_n^-(m)<\frac{2U_-}{\delta}\right\}.
\end{equation}
We write \(\mathcal I_n=\mathcal I_n(1,1)\) for the deterministic interval
analyzed in Theorem~\ref{thm:efficient-betting-efficiency}.  Drawing the two
uniforms gives uniformly randomized Markov calibration; the uniforms change
only the two terminal thresholds, not the betting updates.  Thus either
interval can be implemented directly from
\eqref{eq:regularized-variance-estimate}--%
\eqref{eq:efficient-betting-interval}, with no Gaussian approximation or
numerical dynamic program.

Empirically, the uniformly randomized version of this procedure produces the
tightest bounded-mean confidence intervals to date across our experiments,
generally improving on Gaffke's leading independence-based finite-sample
interval. Figure~\ref{fig:intro-deterministic-comparison} gives an early view of this
comparison for three representative distributions.  

\begin{figure}[t]
    \centering
    \includegraphics[width=0.96\textwidth]
      {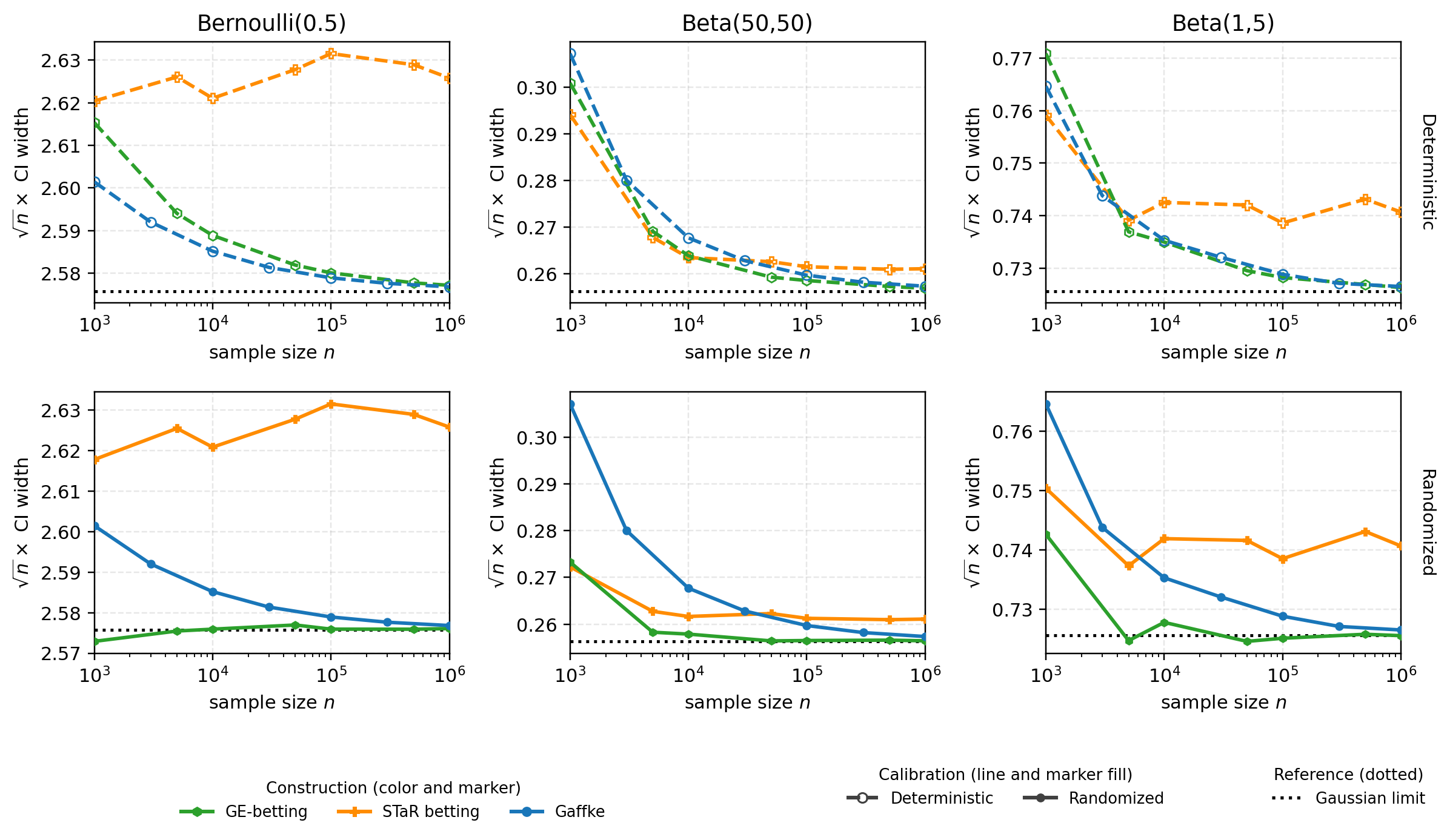}
    \caption{Mean \(\sqrt n\)-scaled widths of 99\% confidence
    intervals for three representative iid distributions.  Our GE-betting width rapidly approaches the
    Gaussian limit (dotted line). First row: all methods are deterministic, and
    GE-betting improves on STaR betting, which clearly converges to an
    inefficient limit, while being competitive with Gaffke, which is also Gaussian-efficient but not martingale-valid. Second row: all methods use external randomization,  and  GE-betting clearly dominates  Gaffke and STaR betting.   The full
    nine-distribution comparison    appears in Figure~\ref{fig:original-versus-star}.}
    \label{fig:intro-deterministic-comparison}
\end{figure}

\paragraph{Sampling without replacement.}
The same design principle applies when \(X_1,\ldots,X_n\) are drawn without
replacement from a fixed population \(x_{1:N}\in[0,1]^N\).  Unlike under iid
sampling, the centered partial sum is constrained to return to zero when the
population is exhausted:
\(S_N-N\mu_N=0\).  Its proportional-sampling Gaussian limit is therefore a
Brownian bridge, which is pinned at zero at time one, rather than Brownian
motion.  This pinning captures the negative dependence between draws and the
decline in uncertainty as more of the population is revealed.  For a
candidate population mean \(m\), let
\(S_i=\sum_{j=1}^iX_j\) and define
\[
 m_i(m)=\frac{Nm-S_{i-1}}{N-i+1},
 \qquad Y_i(m)=X_i-m_i(m).
\]
At the true population mean, \(m_i(m)\) is the conditional mean of the next
draw, so predictable nonnegative betting fractions applied to \(Y_i(m)\)
again generate valid wealth processes.  For a fixed horizon \(n<N\), the
Gaussian-bridge design calculation gives the raw fractions
\begin{equation}
 \widetilde\ell_{i,n}^{\mathrm{br},\pm}(m)
 =\psi\{(\delta/2)K_{i-1}^{\pm}(m)\}
  \sqrt{\frac{N-n}
  {(N-i)(n-i+1)\widehat v_{i-1}}}.
 \label{eq:intro-wor-fractions}
\end{equation}
Here \(\widehat v_{i-1}\) is the shared estimator in
\eqref{eq:regularized-variance-estimate}.  As above, we set \(c=1\); the
upper and lower fractions are capped at \(c/m_i(m)\) and
\(c/\{1-m_i(m)\}\), respectively.  The deterministic interval inverts the
two terminal wealths over the feasible range
\(\mathcal M_n=[S_n/N,\{S_n+N-n\}/N]\):
\begin{equation}
 \mathcal I_{N,n}^{\mathrm{br}}
 =\left\{m\in\mathcal M_n:
 K_n^+(m)<\frac2\delta,\ K_n^-(m)<\frac2\delta\right\}.
 \label{eq:intro-wor-interval}
\end{equation}

Section~\ref{sec:wor-bridge-betting} derives the bridge correction in
\eqref{eq:intro-wor-fractions}, gives the complete bounded-data update, and
establishes finite-population validity and Gaussian efficiency. Across our experiments, the proposed betting strategy yields the tightest bounded-mean confidence intervals to date. Figure~\ref{fig:intro-bridge-comparison} highlights this performance across three representative distributions. 

\begin{figure}[h!]
    \centering
    \includegraphics[width=0.96\textwidth]
      {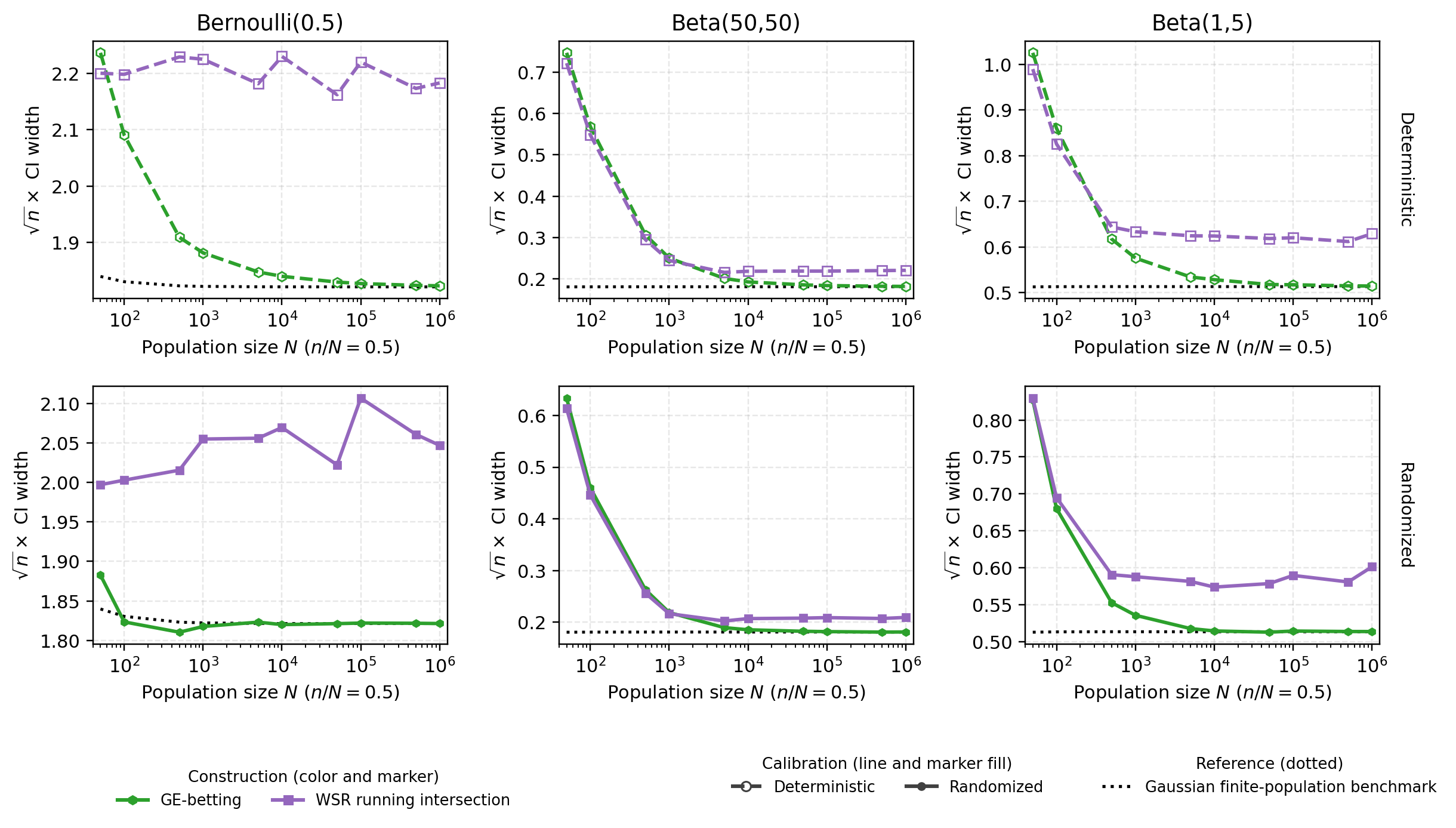}
    \caption{Mean \(\sqrt n\)-scaled widths of 99\% confidence intervals under
  sampling without replacement for three representative finite-population
  families, with sampling fraction \(n/N=.5\) and varying population size \(N\).
  Our GE-betting interval rapidly approaches the Gaussian
  finite-population benchmark (dotted line). GE-betting is compared to the confidence intervals from \cite{waudby2024estimating} (WSR running intersection), which serve as the tightest empirical baseline. First row: both methods use
  deterministic Markov calibration, and GE-betting is substantially
  narrower than WSR running intersection, for moderate and large \(N\).
  Second row: both methods use uniformly randomized Markov calibration, under
  which GE-betting again clearly improves upon WSR.  A full
  nine-population comparison, alongside further alternative confidence intervals such as those from~\citet{bardenet-maillard-2015}, appears in
  Figure~\ref{fig:wor-fixed-fraction}.}
    \label{fig:intro-bridge-comparison}
\end{figure}

\paragraph{Our contributions.}
First, we give a common derivation of product and STaR betting through
conditional expectations of terminal e-values in an auxiliary Gaussian
experiment.  At time \(t\), the
\emph{continuation value} is the conditional expectation of the terminal
e-value given the current Gaussian partial sum.  Its derivative with respect
to that partial sum gives the predictable betting amount.  This common
calculation reveals why a valid strategy can remain inefficient when it
tracks the wrong terminal event.

Second, applying the same calculation to the most powerful fixed-horizon Gaussian
rejection event yields the betting fraction above  and hence the
GE-betting update.  We prove its finite-sample martingale validity,
that its inversion is an interval for every data sequence when the estimator
is shared, and its iid Gaussian efficiency.  Our empirical study shows
that this theoretical efficiency translates into particularly sharp
finite-sample intervals, including in challenging low-variance settings.
We also show how to combine GE-betting with a mixture of constant-fraction
bets to attain both Gaussian and large-deviation efficiency under iid
sampling while retaining martingale validity.

Third, we extend the same design principle to sampling without replacement.
The relevant Gaussian experiment is then a Brownian bridge rather than
Brownian motion, and its continuation value yields a GE-betting
strategy that accounts for depletion of the population.  We prove
finite-population validity, interval-valued inversion, and Gaussian efficiency
under proportional sampling, and find that these guarantees are accompanied
by sharp empirical performance across the same range of bounded populations.

\paragraph{Outline.}
Having stated the proposed interval and its principal guarantee, we devote the
remainder of the paper to its derivation and analysis.
Section~\ref{sec:related-work} reviews related work, while
Section~\ref{sec:background} recaps testing by betting, product betting, and
STaR betting.  Section~\ref{sec:gaussian-continuation-design} then expresses
those existing methods through conditional Gaussian e-values, revealing which
terminal event each strategy tracks.  Using that comparison,
Section~\ref{sec:optimal-star-testing} derives the efficient fixed-horizon
strategy: Subsection~\ref{sec:optimal-gaussian-evalue} obtains
\(\phi\{\Phi^{-1}(p)\}\) from the optimal Gaussian event,
Subsection~\ref{sec:bounded-efficient-betting} explains its transfer to the
bounded-data update already displayed above, and
Subsection~\ref{sec:efficient-betting-efficiency} states the formal validity
and efficiency result.  Within
Section~\ref{sec:betting-to-confidence-intervals},
Subsection~\ref{sec:confidence-set-geometry} proves interval geometry,
Subsection~\ref{sec:uniformly-randomized-markov} discusses fixed-horizon
calibration, and Subsection~\ref{sec:experiments} reports the experiments.
Subsection~\ref{sec:gaussian-ld-efficiency} shows how to attain both Gaussian
and large-deviation efficiency, and Subsection~\ref{sec:ge-pvalues}
constructs \(p\)-values from GE-betting by backward inversion.
Section~\ref{sec:wor-bridge-betting} extends the construction to sampling
without replacement: Subsection~\ref{sec:wor-gaussian-design} gives the
Gaussian-bridge design, Subsection~\ref{sec:wor-bounded-update} gives the
bounded-data update, Subsection~\ref{sec:wor-guarantees} states its guarantees,
and Subsection~\ref{sec:wor-experiments} reports the finite-population
experiments.
Section~\ref{sec:conclusion} concludes.

\section{Related work}
\label{sec:related-work}

\paragraph{Testing by betting.}
Nonnegative wealth processes unify gambling, likelihood ratios, e-values, and
safe tests \citep{vovk2021values,grunwald2024safe,ramdas2023game,
ramdas2025hypothesis}.  For bounded means, \citet{waudby2024estimating}
develop variance-adaptive betting intervals, \citet{shekhar2023near} study
their near-optimality, and \citet{voracek2025star} introduce target
recalculation.  Recent work also learns horizon-dependent strategies
\citep{taga2026learning}.  Related finite-horizon approaches formulate the
choice of betting fractions as a dynamic-programming or stochastic-control
problem
\citep{clerico2026time,baas2026adaptive}.  Our construction instead starts
from a chosen terminal Gaussian payoff, computes its conditional expectation,
and differentiates that expectation with respect to the current Gaussian
partial sum to obtain the next betting amount; it requires no numerical dynamic
program.  The
uniformly randomized Markov inequality \citep{ramdas2026randomized} is an
optional terminal calibration applicable to every fixed-horizon e-value
considered here.  Our focus is instead the construction of a fixed-horizon
e-value whose inversion is Gaussian efficient and empirically sharp.  The
conditional-e-value calculation also identifies the terminal event tracked by
each earlier strategy and thereby explains its width penalty.

\paragraph{Finite-sample intervals for bounded means.}
Fixed-sample intervals can also be obtained by inverting distributional or
concentration bounds.  Examples include Hoeffding and empirical-Bernstein
intervals \citep{hoeffding1963probability,maurer2009empirical}, CDF-envelope
methods \citep{anderson1969confidence}, and small-sample bounded-mean
constructions \citep{fishman1991confidence,rosenblum2009confidence}.  Sharper
recent benchmarks include finite-sample Gaussian approximations
\citep{austern2022efficient} and Gaffke's  interval,
the previous empirical state of the art among finite-sample confidence
intervals for bounded data
\citep{gaffke2005three,learnedmiller2020new,vlassis2026exact,ming2026gaffke}.
\citet{ming2026gaffke} further propose a randomized product-orthant p-value,
whose inversion need not improve Gaffke's endpoints despite its pointwise
p-value improvement.  These methods exploit independence or distributional
envelopes directly.  Our experiments compare their sharpness with that of a
procedure valid under martingale dependence. GE-betting is generally
tighter;  independence is imposed only for  efficiency analysis. 

\paragraph{Sampling without replacement.}
Classical concentration results for sampling without replacement begin with
Hoeffding's comparison inequality and Serfling's finite-population correction
\citep{hoeffding1963probability,serfling-1974}.
\citet{bardenet-maillard-2015} sharpen these bounds and derive an empirical
Bernstein--Serfling interval.  In the betting literature,
\citet{waudby2024estimating} center each draw at the candidate mean of the
remaining population to obtain variance-adaptive confidence intervals and
confidence sequences.  More recently, \citet{shekhar-ramdas-2026} derive
empirical-rate intervals for finite alphabets and an eventually valid
almost-sure interval for general bounded populations.  Exact hypergeometric
inversion provides an additional benchmark for binary populations, but does
not extend directly to arbitrary values in \([0,1]\).  Our construction uses
the same remaining-mean identity as earlier betting methods, but chooses its
fixed-horizon fractions by tracking the optimal Gaussian-bridge rejection
event.

\paragraph{Tracking terminal decisions.}
The event--martingale viewpoint behind our construction goes back to
\citet{ville1939etude}.  Suppose \(A\) is an event (perhaps a terminal rejection event for a null $\Pb_0$) with 
\(\Pb_0(A)=\alpha\).  At time \(t\),
\(\Pb_0(A\mid\mathcal F_t)\) is simply the probability, given the data so far,
that $A$ will eventually occur.  These probabilities form a martingale (called the Doob martingale). Ville's theorem states that for any event $A$ of probability $\alpha$, there exists a nonnegative martingale which reaches $1/\alpha$ if $A$ occurs: this martingale is simply \(\Pb_0(A\mid\mathcal F_t)/\Pb_0(A)\) if $\alpha  > 0$; Ville's deep contribution was handling the $\alpha=0$ case. 
 The
same conditional-rejection
probability appeared also in stochastic curtailment
\citep{lan1982stochastically}.  
\citet{ramdas2022admissible} show that an
admissible point-null sequential test is obtained by thresholding the
normalized Doob martingale obtained from any other existing sequential test (of which fixed-sample tests are of course a special case). The same Doob martingale construction was later explored by  
\citet{koning2026anytime} and \citet{holmes2026predicting} from more practical perspectives.  Our objective is
different: we track the Neyman--Pearson event in an auxiliary Gaussian
experiment to derive bounded-data betting fractions whose fixed-horizon
confidence-set inversion is Gaussian efficient.

\section{Background}
\label{sec:background}

Let \((\mathcal F_i)_{i=0}^n\) be a filtration, with
\(X_i\in[0,1]\) measurable with respect to \(\mathcal F_i\).  We assume
\begin{equation}\label{eq:conditional-mean-model}
    \E[X_i\mid\mathcal F_{i-1}]=\mu,
    \qquad i=1,\ldots,n,
\end{equation}
for an unknown \(\mu\in[0,1]\); no independence is required.  A nonnegative
random variable \(E\) is an \emph{e-variable} for a null hypothesis if its
null expectation is at most one; its realized value is an \emph{e-value}.
As is customary, we also use ``e-value'' for the random variable when no
confusion can arise.  An adapted nonnegative process \((K_t)\), with
\(K_0=1\), is a test supermartingale if
\(\E(K_t\mid\mathcal F_{t-1})\leq K_{t-1}\) under the null.  Every \(K_t\)
is then an e-variable, and Ville's inequality gives
\[
 \Pb\!\left\{\max_{t\leq n}K_t\geq1/\alpha\right\}\leq\alpha.
\]
Markov's inequality gives the corresponding fixed-time guarantee,
\(\Pb\{K_n\geq1/\alpha\}\leq\alpha\).  This is the source of our
finite-sample guarantees under martingale dependence.
\subsection{Testing by betting}
\label{section:testing_by_betting}

To test the point null \(H_m:\mu=m\), we construct nonnegative wealth
processes from predictable bets.  Fix \(m\).  Nonnegative fractions
\(\lambda_i^+(m)\) and \(\lambda_i^-(m)\) generate the upper- and lower-tail
wealth processes
\begin{align}
 K_t^+(m)&=\prod_{i=1}^t\{1+\lambda_i^+(m)(X_i-m)\},\nonumber\\
 K_t^-(m)&=\prod_{i=1}^t\{1+\lambda_i^-(m)(m-X_i)\},
 \qquad K_0^+(m)=K_0^-(m)=1.
 \label{eq:product-wealth-background}
\end{align}
The fraction used at round \(i\) may depend on the past but not on \(X_i\).
We impose the nonnegativity constraints
\begin{equation}\label{eq:product-solvency}
 0\leq\lambda_i^+(m)\leq\frac1m,
 \qquad
 0\leq\lambda_i^-(m)\leq\frac1{1-m}
\end{equation}
which keep the factors nonnegative for every \(X_i\in[0,1]\).  Under
\(\mu\leq m\), \(K^+(m)\) is a nonnegative supermartingale; under
\(\mu\geq m\), the same is true of \(K^-(m)\).  At the point null
\(H_m:\mu=m\), both are test martingales.  Hence their average and their
Bonferroni-scaled maximum,
\[
 E_t^{\rm av}(m)=\tfrac12\{K_t^+(m)+K_t^-(m)\},
 \qquad
 E_t^{\max}(m)=\tfrac12\max\{K_t^+(m),K_t^-(m)\}
\]
are both e-values for \(H_m\); indeed,
\(E_t^{\max}(m)\leq E_t^{\rm av}(m)\).  Their Markov inversions are
\begin{equation}\label{eq:fixed-confidence-set}
\begin{aligned}
 \mathcal C_n^{\rm av}
 &=\{m:E_n^{\rm av}(m)<1/\delta\},\\
 \mathcal C_n^{\max}
 &=\{m:E_n^{\max}(m)<1/\delta\}\\
 &=\{m:K_n^+(m)<2/\delta,\ K_n^-(m)<2/\delta\}.
\end{aligned}
\end{equation}
Each set contains \(\mu\) with probability at least \(1-\delta\).
Here \(\delta\) is the two-sided error probability, and each one-sided test is
allocated error probability \(\delta/2\).  Predictability and the
conditional-mean restriction control expected wealth, while
\eqref{eq:product-solvency} keeps wealth nonnegative.  These two facts
establish validity; the procedures below differ only in their predictable
betting fractions.

Betting inversions are not automatically convex.  For example, averaging a
decreasing upper-tail e-value and an increasing lower-tail e-value can produce a
disconnected sublevel set.  Our main results instead use the maximum
inversion, which intersects the two one-sided acceptance sets.
Proposition~\ref{prop:shared-scale-betting-interval} shows that, when the two
processes share a predictable variance scale, this inversion is an interval on every
sample path.  The confidence intervals in the main results are therefore
convex and easy to compute by locating their two endpoints.

\subsection{The fixed-horizon product bet}

The fixed-horizon strategy of \citet{waudby2024estimating} sets its betting
fraction using a predictable variance estimate.  With the regularized
estimator \(\widehat v_t\) already defined in
\eqref{eq:regularized-variance-estimate}, its fraction before enforcing the
nonnegativity bound is
\[
 \widetilde\lambda_i
 =\sqrt{\frac{2\log(2/\delta)}{n\widehat v_{i-1}}}.
\]
The actual fractions are
\begin{equation}\label{eq:fixed-horizon-product-bet}
 \lambda_i^+(m)=\min\!\left\{\widetilde\lambda_i,\frac c m\right\},
 \qquad
 \lambda_i^-(m)=\min\!\left\{\widetilde\lambda_i,
                                  \frac c{1-m}\right\},
 \qquad 0<c\leq1.
\end{equation}
We set \(c=1\) in the experiments.  The variance estimate is adaptive, but
the target and horizon in the unclipped fraction are fixed in advance.

This strategy also illustrates the distinction in
\eqref{eq:fixed-confidence-set}.  Because the same predictable variance
estimate is used for every candidate \(m\), each one-sided wealth is monotone
in the appropriate direction.  The inversion of \(E_n^{\max}\) is therefore
an interval.  The inversion of \(E_n^{\rm av}\), which adds the two one-sided
wealths before thresholding, need not be an interval and may require
convexification.

\subsection{The STaR bet}
\label{section:star_betting}

STaR recalculates its plan after every observation
\citep{voracek2025star}.  Let \(T=2/\delta\) be the wealth threshold for
one-sided rejection, and let
\(\widehat v_t^\pm(m)\) estimate the next squared centered increment for the
two one-sided processes.  At time \(t<n\), define the remaining log-wealth gaps
\[
 g_t^\pm=\{\log(T/K_t^\pm)\}_+.
\]
STaR betting uses
\begin{align}
 \lambda_{t+1}^{+,\mathrm{STaR}}(m)
 &=\mathbf1\{K_t^+<T\}
   \min\!\left\{
     \sqrt{\frac{2g_t^+}{(n-t)\widehat v_t^+(m)}},\frac c m
   \right\},\nonumber\\
 \lambda_{t+1}^{-,\mathrm{STaR}}(m)
 &=\mathbf1\{K_t^-<T\}
   \min\!\left\{
     \sqrt{\frac{2g_t^-}{(n-t)\widehat v_t^-(m)}},\frac c{1-m}
   \right\}.
 \label{eq:original-star-background-bet}
\end{align}
The experiments use \(c=1\).  After each update, the strategy solves the same
remaining-horizon problem.  Equivalently, if
\(p_t^\pm=K_t^\pm/T=(\delta/2)K_t^\pm\), the wealth-dependent factor in its
betting fraction is
\(f_{\rm sqrt}(p)=\sqrt{2\log(1/p)}\).  Once a wealth reaches \(T\), it is
held fixed, so later data cannot reverse rejection.  The use of the
remaining wealth gap and remaining horizon changes power, but not the e-value
validity established in Subsection~\ref{section:testing_by_betting}.

\citet{voracek2025star} note that their candidate-dependent STaR inversion need
not be convex.
Their implementation estimates the centered variance separately for each
candidate null, so the update need not preserve an order in \(m\).  This does
not rule out convex inversion for STaR betting itself.
Section~\ref{sec:betting-to-confidence-intervals} shows that
using one predictable variance estimator shared across candidate means makes
the maximum-based inversion an interval pathwise.

\subsection{Betting under sampling without replacement}
\label{sec:wor-background}

Fix a population \(x_{1:N}\in[0,1]^N\), and reveal its values in a uniformly
random order.  Let
\begin{equation}
 S_t=\sum_{j=1}^tX_j,
 \qquad \mathcal F_t=\sigma(X_1,\ldots,X_t),
 \qquad 0\leq t\leq N.
 \label{eq:wor-partial-sums}
\end{equation}
The population mean and variance, with denominator \(N\), are
\begin{equation}
 \mu_N=\frac1N\sum_{j=1}^Nx_j,
 \qquad
 \sigma_N^2=\frac1N\sum_{j=1}^N(x_j-\mu_N)^2.
 \label{eq:wor-population-moments}
\end{equation}
For a simple random sample of size \(n\), the standard error of
\(\bar X_n=S_n/n\) is
\begin{equation}
 \tau_{N,n}=\sigma_N\sqrt{\frac{N-n}{n(N-1)}}.
 \label{eq:wor-standard-error}
\end{equation}
After \(t\) draws, the observed values can be completed to a population in
\([0,1]^N\) with mean \(m\) if and only if
\begin{equation}
 m\in\mathcal M_t
 :=\left[\frac{S_t}{N},\frac{S_t+N-t}{N}\right].
 \label{eq:wor-feasible-range}
\end{equation}

This completion identity also gives the centered observations used for
betting.  Immediately before draw \(i\), a population with candidate mean
\(m\) must have mean
\begin{equation}
 m_i(m)=\frac{Nm-S_{i-1}}{N-i+1}
 \label{eq:wor-remaining-mean}
\end{equation}
among its remaining values.  Thus
\(m_i(\mu_N)=\E(X_i\mid\mathcal F_{i-1})\).  If \(\mu_N\leq m\), the actual
conditional mean is at most \(m_i(m)\); if \(\mu_N\geq m\), it is at least
\(m_i(m)\).  Define
\begin{equation}
 Y_i(m)=X_i-m_i(m).
 \label{eq:wor-innovation}
\end{equation}
For predictable fractions satisfying
\[
 0\leq\lambda_i^+(m)\leq\frac1{m_i(m)},
 \qquad
 0\leq\lambda_i^-(m)\leq\frac1{1-m_i(m)},
\]
the one-sided updates are
\begin{align}
 K_i^+(m)&=K_{i-1}^+(m)
   \{1+\lambda_i^+(m)Y_i(m)\},\nonumber\\
 K_i^-(m)&=K_{i-1}^-(m)
   \{1+\lambda_i^-(m)(-Y_i(m))\}.
 \label{eq:wor-generic-updates}
\end{align}
The caps keep both multiplicative factors nonnegative.  The upper process is
a test supermartingale under \(\mu_N\leq m\), and the lower process is one
under \(\mu_N\geq m\).  Every admissible predictable choice is therefore
valid; the choice of fractions determines the width of the resulting
interval.

\section{Conditional Gaussian e-values as a design principle}
\label{sec:gaussian-continuation-design}

The Gaussian calculations in this section are design devices, not
distributional assumptions on the data.  The observations remain bounded and
thus non-Gaussian; under a candidate null, their centered partial sums are
sums of martingale differences.  We use an auxiliary Gaussian experiment (a
null experiment in which a variance-normalized Gaussian partial sum evolves
as Brownian motion) to derive predictable betting rules, and then apply those
rules to the bounded observations.  Their exact finite-sample validity follows
from predictability, the conditional-mean restriction, and clipping, not from
the Gaussian approximation.

Product betting and STaR betting initially look quite different: product
betting uses a fixed fraction, whereas STaR recomputes its fraction from
current wealth and the remaining horizon.  In the auxiliary experiment, each
process is obtained by conditioning a terminal e-value on the current
Gaussian partial sum.  Product betting uses an exponential terminal e-value;
STaR uses the indicator of a path-dependent barrier-hitting event.  In each case, the
derivative with respect to the current Gaussian partial sum determines the
next betting amount.  This common representation identifies the event rewarded by
each strategy and explains its fixed-horizon performance.  It also motivates
Section~\ref{sec:optimal-star-testing}, where we instead track the
most powerful terminal rejection event.

\subsection{Product betting from the exponential test function}
\label{sec:recovering-testing-by-betting}

We first apply the conditional-e-value construction in the auxiliary Gaussian
experiment.  Let \((W_t)_{0\leq t\leq1}\) be standard Brownian motion with
natural filtration \((\mathcal F_t)\), and let \(Z\sim\mathcal N(0,1)\) be
independent of this process.  Given \(g\geq0\) satisfying
\(0<\E g(Z)<\infty\), normalize the terminal payoff as
\(E_1=g(W_1)/\E g(Z)\).  Conditional on \(W_t=x\), the future increment has
variance \(v=1-t\).  Define
\[
 u(v,x)=\E g(x+\sqrt vZ).
\]
The function \(u\) satisfies the heat equation, and
\[
 E_t=\frac{u(1-t,W_t)}{u(1,0)}
     =\E_0(E_1\mid\mathcal F_t)
\]
is a nonnegative martingale starting from one, and hence an e-process; here
\(\E_0\) denotes expectation under the null Brownian law.  Itô's
formula gives
\[
 \mathrm dE_t
 =\frac{\partial_xu(1-t,W_t)}{u(1,0)}\,\mathrm dW_t.
\]
Thus \(\partial_xu/u(1,0)\) is the amount bet per unit increment in \(W_t\),
whereas \(\partial_xu/u\) is the corresponding fraction of current wealth.
For product betting, take \(g_\eta(x)=e^{\eta x}\).  Its conditional
expectation at \(W_t=x\), with
\(v\) units of variance remaining, is
\[
    u_\eta(v,x)=\E e^{\eta(x+\sqrt vZ)}
      =e^{\eta x+\eta^2v/2}.
\]
After normalization by the initial value \(u_\eta(1,0)\), the value process is
\[
 \frac{u_\eta(1-t,W_t)}{u_\eta(1,0)}
 =\exp\{\eta W_t-\eta^2t/2\},
\]
the familiar Gaussian exponential e-process.  The parameter \(\eta\) is the
betting fraction: a larger value makes wealth more responsive to favorable
increments but also incurs a larger variance penalty.  Since
\(\partial_xu_\eta(v,x)=\eta u_\eta(v,x)\), the amount bet is always the
fraction \(\eta\) of current wealth.  Indeed, if current wealth is
\(K=u_\eta(v,x)/u_\eta(1,0)\), then the amount bet per unit increment is
\[
 \frac{\partial_xu_\eta(v,x)}{u_\eta(1,0)}
 =\eta\frac{u_\eta(v,x)}{u_\eta(1,0)}
 =\eta K,
\]
so the amount bet is exactly \(\eta K\).  The conditional-e-value
construction therefore recovers ordinary product betting.
Thus, for a bounded-data increment \(y=X_i-m\), the corresponding update is
\begin{equation}\label{eq:product-one-step-update}
 K_i=K_{i-1}+\eta K_{i-1}y
     =K_{i-1}(1+\eta y).
\end{equation}
The amount bet before observing \(X_i\) is therefore predictable and equal to
\(\eta K_{i-1}\).  Predictability and \eqref{eq:product-solvency}, rather than
the Gaussian approximation, establish the e-value property.

The fixed-horizon scale is already visible from
\[
 \log K_n\approx
 \eta\sum_{i=1}^n(X_i-m)-\frac{\eta^2n\sigma^2}{2}.
\]
Under a local alternative, the centered partial sum is of order \(\sqrt n\),
so the linear gain and quadratic variance cost are both of constant order when
\(\eta\asymp(n\sigma^2)^{-1/2}\).  Product betting is therefore the
bounded-data counterpart of the exponential Gaussian e-process, with a
predictable variance estimate supplying the unknown scale.
Appendix~\ref{app:squared-hinge-methods} applies the same conditional-e-value
construction to a squared hinge that more closely approximates an indicator.

\subsection{STaR betting and its barrier-hitting event}
\label{sec:star-digital-hedging}
\label{sec:star-barrier-digital}

STaR fits the same conditional-e-value construction, although the event it
tracks is initially implicit in its wealth-dependent betting fraction.  In the
continuous Gaussian limit, relative wealth
\(P_t=K_t/T=(\delta/2)K_t\), held fixed after reaching one, is a bounded null
martingale starting from \(\delta/2\).  If this wealth converges at the deadline to either zero or
one, then
\(P_t=\Pb_0(A\mid\mathcal F_t)\) for some event \(A\).  We now identify that
event.

Let \(Y_t\) denote the null Brownian partial-sum process in Gaussian variance
time.  Relative STaR wealth satisfies
\[
 \mathrm dP_t
 =P_t\sqrt{\frac{2\log(1/P_t)}{1-t}}\,\mathrm dY_t.
\]
Define the remaining standardized distance to rejection by
\[
 R_t=\sqrt{2(1-t)\log(1/P_t)}.
\]
A direct change-of-variable calculation gives
\begin{equation}\label{eq:barrier-distance-sde}
 \mathrm dR_t=-\mathrm dY_t-\frac{\mathrm dt}{2R_t},
 \qquad \tau_0=\inf\{t:R_t=0\}.
\end{equation}
After \(R_t\) reaches zero, it remains zero because STaR freezes the
one-sided wealth process at the rejection target.  For \(Q_t=R_t^2\), the
same calculation yields
\[
 \mathrm dQ_t
 =-2\sqrt{Q_t}\,\mathrm dY_t
 =2\sqrt{Q_t}\,\mathrm dW_t,
 \qquad W_t=-Y_t.
\]
This is called a dimension-zero squared-Bessel process; \emph{dimension} is
the name of its drift parameter and has nothing to do with the dimension of
the observations.  Its probability of
hitting zero within \(v\) remaining units of variance is
\begin{equation}\label{eq:barrier-digital-price}
 H(v,r):=\Pb_r(\tau_0\leq v)
 =\exp\{-r^2/(2v)\}.
\end{equation}
Here \(\Pb_r\) denotes probability for the process initialized at \(R_0=r\).
In particular, \(P_t=H(1-t,R_t)\) is the conditional probability of the event
\(A_{\rm STaR}=\{\tau_0\leq1\}\).  STaR betting therefore spends its
one-sided error probability on a path-dependent hitting event.

Conversely, solving \(p=H(v,r)\) for the required distance gives
\begin{equation}\label{eq:barrier-price-matching-distance}
 r(p,v)=\sqrt{2v\log(1/p)}.
\end{equation}
Differentiating the conditional hitting probability with respect to the
current Gaussian partial sum gives the betting fraction per unit increment,
\begin{equation}\label{eq:barrier-digital-betting-fraction}
 \frac{R_t}{1-t}
 =\sqrt{\frac{2\log(1/P_t)}{1-t}}.
\end{equation}
Replacing remaining Gaussian variance by a predictable estimate yields
\begin{equation}\label{eq:star-receding-horizon-bet}
 \ell_{t+1}^{\rm STaR}(m)
 =\sqrt{\frac{2\log(T/K_t)}{(n-t)v_t(m)}},
\end{equation}
the STaR betting fraction.  Here \((n-t)v_t(m)\) estimates the conditional
variance of the centered sum still to be observed, while the logarithm
converts the current wealth shortfall into the standardized distance in
\eqref{eq:barrier-price-matching-distance}.  Thus the square-root rule is the
derivative of the conditional hitting probability divided by that
probability, recomputed after every observation.

Product betting and STaR betting are thus both conditional-e-value strategies:
the former reveals an exponential terminal e-variable, while the latter
reveals a barrier-hitting indicator.  This interpretation also explains why
STaR is not efficient at a fixed horizon.  The most powerful Gaussian test
rejects according to the terminal Gaussian observation, whereas STaR
allocates null probability to paths that hit the barrier early, including
paths that subsequently finish
below the terminal cutoff.  Finite-sample e-value validity does not remove
this mismatch between the path-dependent hitting event and the optimal
terminal event.

\section{Gaussian-efficient testing by betting}
\label{sec:optimal-star-testing}

The conditional-e-value calculation suggests starting from the fixed-horizon
test one ultimately wants to perform.  In the Gaussian shift experiment, this
is the most powerful terminal test.  GE-betting tracks its conditional
rejection probability and uses the resulting martingale representation to
define a finite-sample-valid e-value for bounded data.  The benchmark analyses in
Section~\ref{sec:gaussian-continuation-design} show why this choice matters: a
strategy can be perfectly valid yet inefficient when its wealth tracks a
different event.

The derivation proceeds in three steps.  First,
Subsection~\ref{sec:optimal-gaussian-evalue} obtains the betting rule in an
auxiliary Gaussian experiment, without imposing Gaussianity on the observed
data.  Next, Subsection~\ref{sec:bounded-efficient-betting} transfers that
rule to the bounded-data update already stated in
Section~\ref{sec:introduction}.  Finally,
Subsection~\ref{sec:efficient-betting-efficiency} establishes finite-sample
validity and iid efficiency.
Appendix~\ref{app:horizon-free-wr-cs} extends the construction to
planned-window confidence sequences.

\subsection{Gaussian design device: the optimal fixed-horizon e-value}
\label{sec:optimal-gaussian-evalue}

This subsection is a design calculation in a Gaussian experiment.  It is used
only to derive a betting fraction; the bounded observations to which the
resulting strategy is applied will not be Gaussian.

In the Gaussian shift experiment, let \((B_t)\) be standard Brownian motion and write
\(\mathrm dY_t=\mathrm dB_t+h\,\mathrm dt\), \(0\leq t\leq1\), so that
\(Y_1\sim\mathcal N(h,1)\).  For testing \(h=0\) against \(h>0\), the most powerful
level-\(\delta/2\) event and its indicator e-value are
\[
    A^*=\{Y_1\geq z_{1-\delta/2}\},
    \qquad
    E_1^*=\frac2\delta\mathbf1_{A^*}.
\]
Under \(h=0\), the future increment \(Y_1-Y_t\) is independent Gaussian noise
with variance \(1-t\).  The conditional rejection probability is therefore
\begin{equation}\label{eq:gaussian-digital-price}
    p_t=\Pb_0(A^*\mid Y_t)
       =\Phi\left\{\frac{Y_t-z_{1-\delta/2}}{\sqrt{1-t}}\right\}.
\end{equation}
The formula is a conditional normal probability: \(\sqrt{1-t}\) is the
standard deviation of the future increment, so the argument of \(\Phi\) is
the signed distance to rejection in units of remaining uncertainty.  In
particular, \(p_t\) is a bounded null martingale.  Under the identification
\(K_t:=(2/\delta)p_t\), we have \(p_t=(\delta/2)K_t\), so \(p_t\) is current
wealth as a fraction of the one-sided target.

To obtain the corresponding betting fraction, let
\(U(t,y)=\Phi\{(y-z_{1-\delta/2})/\sqrt{1-t}\}\).
Differentiating with respect to the current partial sum and applying the
Brownian martingale representation gives
\[
    \mathrm dp_t
      =\frac{\phi\{\Phi^{-1}(p_t)\}}{\sqrt{1-t}}\,\mathrm dY_t.
\]
Dividing the last display by \(p_t=(\delta/2)K_t\) gives
\[
 \frac{\mathrm dK_t}{K_t}
 =L_{\rm opt}(p_t,1-t)\,\mathrm dY_t.
\]
Thus
\begin{equation}\label{eq:efficient-betting-fraction}
    L_{\rm opt}(p,r)=\frac{\psi(p)}{\sqrt r},
    \qquad
    \psi(p)=\frac{\phi\{\Phi^{-1}(p)\}}{p}.
\end{equation}
Here \(p\in(0,1)\) is current wealth divided by its target, \(r\) is the
remaining Gaussian variance, and \(L_{\rm opt}(p,r)\) is the fraction of
current wealth bet per unit Gaussian increment.  The notation
\(L_{\rm opt}\) refers to the Neyman--Pearson fixed-horizon event, not to a
general optimality statement for finite-sample betting strategies.

Figure~\ref{fig:efficient-betting-function} separates the fraction of current
wealth bet from the amount bet as a fraction of the rejection target.
The left panel shows that \(\psi(p)\) is positive and strictly decreasing: it
diverges as \(p\downarrow0\) and tends to zero as \(p\uparrow1\).  A strategy
far below its target therefore bets a larger fraction of its current wealth,
whereas a strategy close to its target bets a smaller fraction.  The amount
bet as a fraction of the target is
\[
 p\psi(p)=\phi\{\Phi^{-1}(p)\}.
\]
As the right panel shows, this amount tends to zero at both endpoints, is
symmetric about \(p=1/2\), and is largest at \(p=1/2\).  Thus the large value
of \(\psi(p)\) near zero does not imply a large amount bet relative to the
target, because current wealth is itself small.  The strict concavity of
\(p\psi(p)\), shown by
\(\{p\psi(p)\}''=-1/\phi\{\Phi^{-1}(p)\}<0\), is also the property used in
Section~\ref{sec:confidence-set-geometry} to prove that inversion gives an
interval.

\begin{figure}[t]
    \centering
    \includegraphics[width=0.92\textwidth]
      {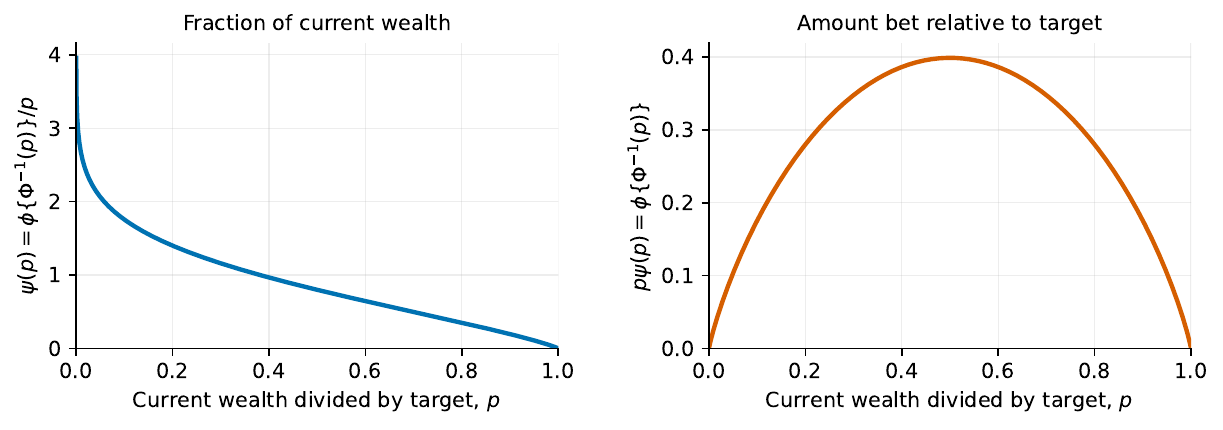}
    \caption{The GE-betting fraction and the corresponding amount bet
    as a fraction of the rejection target.  The fraction of current wealth
    is largest when wealth is small (left), but the amount relative to the
    target goes to zero both near zero wealth and near the rejection target
    (right).}
    \label{fig:efficient-betting-function}
\end{figure}

Under the null Gaussian experiment, the two betting strategies in the main
text have the common form
\(\mathrm dP_t=\sigma_f(P_t)\,\mathrm dY_t/\sqrt{1-t}\).  The denominator
is the standard deviation of the future Gaussian increment.  The numerator
\(\sigma_f(p)\) is the amount bet, expressed as a fraction of the rejection
target, per unit increment in \(Y_t\) when relative wealth is \(p\).  STaR betting uses
\(p\sqrt{2\log(1/p)}\), whereas GE-betting uses
\(\phi\{\Phi^{-1}(p)\}\).  They use the same remaining-variance scale but
correspond to different fixed-horizon events.  Only
\eqref{eq:gaussian-digital-price} tracks the Neyman--Pearson event.  We now
leave the Gaussian design experiment: the next subsection explains how this
Gaussian betting fraction leads to the bounded-data update displayed in the
introduction.

\subsection{From the Gaussian design to the bounded-data update}
\label{sec:bounded-efficient-betting}

The implementable strategy is already given in
\eqref{eq:regularized-variance-estimate}--%
\eqref{eq:efficient-betting-interval}.  Here we explain the short passage from
the Gaussian design calculation to that finite-sample-valid bounded-data
update.

The continuous-time betting fraction diverges at the deadline, but in discrete
time the last round retains one observation of variance.  For a positive
predictable estimate \(\widehat v_{i-1}\), replacing the remaining Gaussian
variance \(1-t\) by \((n-i+1)\widehat v_{i-1}\) gives the untruncated fraction
\begin{equation}\label{eq:efficient-betting-fraction-discrete}
    \ell_{i,n}^{\rm eff}(m)
    =\frac{\psi\{(\delta/2) K_{i-1}(m)\}}
      {\sqrt{\{n-i+1\}\widehat v_{i-1}}},
\end{equation}
where \((\delta/2)K_{i-1}\) is current wealth divided by the rejection
threshold \(2/\delta\).  Applying
this fraction to the upper and reflected lower increments gives the two
fractions in \eqref{eq:efficient-one-sided-fractions}.  Their caps, \(c/m\) and
\(c/(1-m)\), with \(c=1\), keep the next wealth nonnegative for every
\(X_i\in[0,1]\).
Holding wealth fixed after it reaches \(2/\delta\) makes rejection
irreversible.

Predictability controls the conditional expectation of each update, while
the caps keep it nonnegative.  Each one-sided wealth process in
\eqref{eq:efficient-one-sided-updates} is therefore a nonnegative
supermartingale under its one-sided null.  Running the two reflected processes and
inverting their terminal decisions gives
\eqref{eq:efficient-betting-interval}.  Sharing the estimator across the two processes and
candidate means is not needed for pointwise e-value validity, but it makes the
inversion an interval.  Section~\ref{sec:uniformly-randomized-markov} justifies
the two choices of \((U_+,U_-)\) included in the introductory definition;
they change neither the betting strategy nor this derivation.

\subsection{Gaussian efficiency of GE-betting}
\label{sec:efficient-betting-efficiency}

Estimated variance, clipping, and the rule that freezes wealth at zero or the
rejection target separate the bounded-data
strategy from its Gaussian idealization.  The next result gathers the three
main guarantees: finite-sample validity under the conditional-mean model,
interval-valued inversion on every sample path, and optimal first-order
Gaussian width under iid sampling.  We state it for the regularized estimator
\eqref{eq:regularized-variance-estimate}, which is the choice used throughout
our comparisons.

\begin{theorem}[Validity and Gaussian efficiency of GE-betting]
\label{thm:efficient-betting-efficiency}
Fix \(\delta\in(0,1)\) and \(c\in(0,1]\).  In
\eqref{eq:efficient-one-sided-fractions}--%
\eqref{eq:efficient-one-sided-updates}, use the estimator \(\widehat v_{i-1}\)
from \eqref{eq:regularized-variance-estimate}, shared by both wealth sequences
and every candidate mean, and let
\(\mathcal I_n=\mathcal I_n(1,1)\) use deterministic terminal calibration.
For any adapted observations \(X_i\in[0,1]\) satisfying
\eqref{eq:conditional-mean-model}, \(\mathcal I_n\) is an interval on every
sample path and
\[
 \Pb\{\mu\in\mathcal I_n\}\geq1-\delta.
\]
The same validity and interval conclusions hold for
\(\mathcal I_n(U_+,U_-)\), where \(U_+,U_-\) are the independent uniforms used
for randomized terminal calibration.
If, in addition, the observations are iid with
\(\mu\in(0,1)\) and variance \(\sigma^2>0\), then
\[
\frac{\sqrt n}{2\sigma}\operatorname{len}(\mathcal I_n)
\xrightarrow{\mathrm{a.s.}}z_{1-\delta/2}.
\]
\end{theorem}

Since \(\mathcal I_n(U_+,U_-)\subseteq\mathcal I_n(1,1)\) pathwise, the iid
conclusion also gives
\(\limsup_n\sqrt n\,\operatorname{len}\{\mathcal I_n(U_+,U_-)\}
\leq2\sigma z_{1-\delta/2}\) almost surely. 

\textit{Proof sketch.}
The validity and interval claims require no asymptotics.  At \(m=\mu\), the
predictable fractions and nonnegativity caps make both one-sided wealth
processes test supermartingales, so deterministic or uniformly randomized
Markov calibration and a union bound give coverage \(1-\delta\).  Moreover,
\(p\psi(p)=\phi\{\Phi^{-1}(p)\}\) is concave, and the variance estimator is
shared across both processes and all candidate means.  Proposition~%
\ref{prop:shared-scale-betting-interval} therefore shows that the inverted
confidence set is an interval, for every fixed pair of
terminal thresholds.

The efficiency argument is driven by one exact cancellation.  In the Gaussian
design~\eqref{eq:gaussian-digital-price}, relative wealth is \(\Phi\) applied
to a shifted, rescaled partial sum.
Applying \(\Phi^{-1}\) therefore recovers that statistic, suggesting the
same transformation \(\Phi^{-1}\) for bounded data to expose the empirical sum underlying
the Gaussian test.  Hence, let
\(p_i=(\delta/2)K_i\) be upper-tail wealth divided by its target, and let
\(q_i=\Phi^{-1}(p_i)\) be its normal-quantile transform (when \(0<K_i<2/\delta\)).   We shall prove that, for some $j_n$ smaller but close to $n$,  $q_{j_n} / \sqrt{n}$ is roughly the aforementioned shifted rescaled partial sum;  the remaining \((n - j_n)\)-rounds  polarize this preterminal separation
  into rejection or acceptance.

In particular, denote
\(r_i=n-i+1\) and \(Y_i=X_i-m\). While the cap is
inactive and wealth has not reached a boundary, the update is
\[
 p_i-p_{i-1}
 =\phi(q_{i-1})\frac{Y_i}{\sqrt{r_i\widehat v_{i-1}}}.
\]
Since \((\Phi^{-1})'(p)=1/\phi\{\Phi^{-1}(p)\}\), and calling $q = \Phi^{-1}(p)$, we have the key identity
\begin{equation}\label{eq:main-proof-key-cancellation}
 (\Phi^{-1})'(p)\,\phi(q)=1.
\end{equation}
This identity allows us to infer that the first-order change in \(q_i\) is the standardized
centered observation \(Y_i/\sqrt{r_i\widehat v_{i-1}}\).  Indeed, a
first-order Taylor expansion of \(q_i=\Phi^{-1}(p_i)\) gives
\[
 q_i-q_{i-1}
 \approx(\Phi^{-1})'(p_{i-1})(p_i-p_{i-1})
 =\frac{Y_i}{\sqrt{r_i\widehat v_{i-1}}}.
\]
The above discrete derivative expresses a change in relative wealth on the normal-quantile
scale: its first-order increment does not depend on current wealth.
To pass from this one-step relation to the cumulative expansion, we rescale
\(q_i\) by \(\sqrt{(n-i)/n}\), the ratio of the remaining to the total
standard deviation.
The second-order Taylor term cancels, to first order, the average change
caused by this decreasing factor:
\[
 \sqrt{\frac{r_i-1}{n}}q_i-\sqrt{\frac{r_i}{n}}q_{i-1}
 =\frac{Y_i}{\sqrt{n\widehat v_{i-1}}}
 +\frac{q_{i-1}}{2\sqrt{nr_i}}
  \left\{\frac{Y_i^2}{\widehat v_{i-1}}-1\right\}
 +\frac{e_{i,n}}{\sqrt n}.
\]
For local candidates,
\(\E(Y_i^2\mid\mathcal F_{i-1})\approx\sigma^2\approx\widehat v_{i-1}\),
so the expression in braces has approximately zero conditional mean.
The proof shows that its accumulated contribution and the remainder are
negligible, leaving the linear term
\(Y_i/\sqrt{n\widehat v_{i-1}}\approx Y_i/(\sigma\sqrt n)\), whose sum gives
the standardized empirical sum below.
The displayed identity is \eqref{eq:appendix-scaled-quantile-centered}
divided by \(\sqrt n\), with the remainder \(e_{i,n}\) bounded in
\eqref{eq:appendix-scaled-quantile-error-bound}.
At \(i=j_n=n-k_n\), the rescaling factor is \(\sqrt{k_n/n}\).  For
\(m=\mu-h\sigma/\sqrt n\) and \(j_n=n-k_n\), the resulting expansion after
\(j_n\) observations is\footnote{This is the step that fails for the earlier betting strategies.  If
\(\sigma_f(p)\) denotes the amount bet as a fraction of the target, its
first-order multiplier on the normal-quantile scale is
\(\sigma_f(p)/\phi\{\Phi^{-1}(p)\}\).  For STaR this multiplier has numerator
\(p\sqrt{2\log(1/p)}\), and for product betting it has a numerator
proportional to \(p\); neither ratio is constant.  Their recursions for
\(q_i=\Phi^{-1}(p_i)\) therefore retain a wealth-dependent weight instead of
reducing to the ordinary standardized sample mean.} 
\[
 \sqrt{\frac{k_n}{n}}q_{j_n}
 =\frac{1}{\sigma\sqrt n}\sum_{i=1}^{j_n}(X_i-m)
   -z_{1-\delta/2}+o(1)
 =\underbrace{\frac{\sqrt n(\bar X_n-\mu)}{\sigma}}_{Z_n}
   +h-z_{1-\delta/2}+o(1)
   \qquad\text{almost surely}.
\]

The quantity \(Z_n+h=\sqrt n(\bar X_n-m)/\sigma\) is precisely the usual
  standardized statistic for the candidate mean \(m\). The displayed expansion shows
  that \(q_{j_n}\) becomes large and positive when the test statistic
  \(Z_n+h\) stays above the cutoff \(z_{1-\delta/2}\) by a fixed amount,
  and large and negative when it stays below the cutoff by a fixed amount. For the remaining
  \(k_n\)-rounds, probabilistic  argument then shows that this preterminal separation polarizes
  into rejection or acceptance, so that the resulting acceptance--rejection
  boundary agrees with that of the Gaussian \(z\)-test.



The proof is given in Appendix~\ref{proof:efficient-betting-efficiency}.

\section{From betting strategies to confidence intervals}
\label{sec:betting-to-confidence-intervals}

Two issues remain after constructing the e-value \(K_n(m)\): whether its
inversion is connected, and how it is calibrated into a test.  We address
them in turn.

\subsection{When does inversion give an interval?}
\label{sec:confidence-set-geometry}

Pointwise validity does not imply connected inversion.  Candidate-dependent
variance estimates can make the statistic at time \(n\) oscillate with \(m\).
A predictable scale shared across
candidate means restores an order property for the betting strategies
considered here.
Appendix~\ref{sec:empirical-common-estimator-cost} shows that sharing the
variance estimator need not increase width: in the paired experiments, the
shared-estimator intervals are comparable to, and sometimes narrower than,
their candidate-dependent counterparts.

This distinction also appears in earlier betting intervals.  Hedged-CI uses
fractions independent of \(m\), apart from clipping that keeps wealth
nonnegative, and has
interval-valued one-sided inversion; candidate-dependent aGRAPA need not
\citep{waudby2024estimating}.  STaR-Bets estimates \(\E[(X-m)^2]\) separately
for each \(m\) and returns the interval enclosing nonrejected grid cells
\citep{voracek2025star}.

Write \(p\in[0,1]\) for wealth relative to a one-sided target,
\(f(p)\) for the betting fraction before division by the predictable scale,
and \(\sigma_f(p)=pf(p)\) for the amount bet relative to the target.  With
a common predictable scale \(s>0\), observation \(x\), candidate \(m\), and
clipping constant \(c\in(0,1]\), the upper-tail update, with wealth held at
one after reaching the normalized rejection target, is
\begin{equation}\label{eq:shared-scale-betting-update}
 T_{x,m,s}(p)
 =\min\!\left[1,
 p\left\{1+\min\!\left(\frac{f(p)}s,\frac c m\right)(x-m)\right\}
 \right].
\end{equation}
At \(m=0\), define \(c/m=+\infty\).  For the reflected lower-tail update at
\(m=1\), define \(c/(1-m)=+\infty\).  These definitions remove the
nonnegativity cap on the side where the centered observation cannot make
wealth negative.

\begin{proposition}[Pathwise interval inversion under a shared scale]
\label{prop:shared-scale-betting-interval}
Suppose \(s_i>0\) is predictable and common to both one-sided processes and all \(m\).
Assume \(\sigma_f\) is nonnegative and concave on \([0,1]\), continuously
differentiable on \((0,1)\), and \(\sigma_f(0)=0\).  For any fixed one-sided
thresholds \(u_+,u_-\in(0,1]\), the set accepted by the two reflected updates
\eqref{eq:shared-scale-betting-update} is an interval, possibly empty.
\end{proposition}

The proof is given in Appendix~\ref{proof:shared-scale-betting-interval}. For STaR betting,
\[
 \sigma_{\rm sqrt}(p)=p\sqrt{2\log(1/p)},\qquad
 \sigma_{\rm sqrt}''(p)
 =-\frac{1+\{2\log(1/p)\}^{-1}}
         {p\sqrt{2\log(1/p)}}<0,
\]
whereas for GE-betting, with \(q=\Phi^{-1}(p)\),
\[
 \sigma_{\rm eff}(p)=\phi(q),\qquad
 \sigma_{\rm eff}'(p)=-q,\qquad
 \sigma_{\rm eff}''(p)=-\frac1{\phi(q)}<0.
\]
Both second derivatives are negative.  The fixed-horizon product rule is
covered by the analogous Hedged-CI argument of
\citet{waudby2024estimating}.  Hence the inversions of both shared-scale
betting strategies in the main text are intervals pathwise.

\subsection{Fixed-horizon calibration by uniformly randomized Markov}
\label{sec:uniformly-randomized-markov}

E-value construction and fixed-horizon calibration are separate.  For method
\(A\) and candidate \(m\), let
\(K_{n,A}^{+}(m)\) and \(K_{n,A}^{-}(m)\) denote its terminal
one-sided e-values.  Thus \(\E K_{n,A}^{\pm}(m)\leq1\) under the corresponding
one-sided null.
Deterministic Markov calibration rejects when
\(K_{n,A}^{\pm}(m)\geq2/\delta\).  The uniformly randomized Markov
inequality \citep{ramdas2026randomized} draws
\(U_+,U_-\stackrel{\rm iid}{\sim}\operatorname{Unif}(0,1)\), independently of
the data, and rejects when
\begin{equation}\label{eq:randomized-markov-rejection}
 K_{n,A}^{\pm}(m)\geq \frac{2U_\pm}{\delta}.
\end{equation}
Equivalently, the randomized e-value corresponding to this decision is
\begin{equation}\label{eq:randomized-markov-terminal}
 E_{n,A}^{\pm}(m)
 =\frac{2}{\delta}
  \mathbf1\!\left\{U_\pm\leq
  \min\!\left((\delta/2)K_{n,A}^{\pm}(m),1\right)\right\}.
\end{equation}
The same pair \((U_+,U_-)\) is held fixed over all \(m\) during inversion.
The randomizers do not alter the betting fractions, variance estimator,
clipping, or rule that freezes wealth at a boundary.  They change only the
terminal threshold, allowing rejection below \(2/\delta\) when the independent
uniform randomizer is sufficiently small.  Randomized Markov therefore changes
only the conversion of the terminal e-value into a binary test decision; it is
not a new betting strategy.

\begin{proposition}[Validity of fixed-horizon calibration]
\label{prop:terminal-calibration-validity}
Deterministic and uniformly randomized Markov calibration each give a
level-\(\delta/2\) test from any one-sided e-value \(K_{n,A}^{\pm}(m)\).
Inverting the two reflected tests gives coverage at least \(1-\delta\).
\end{proposition}

The proof is given in Appendix~\ref{proof:terminal-calibration-validity}.
Setting \(U_+=U_-=1\) recovers deterministic Markov.  Since \(U_\pm\leq1\),
the randomized confidence set is pathwise contained in its deterministic
counterpart.  For the shared-scale strategies of
Proposition~\ref{prop:shared-scale-betting-interval}, fixed random thresholds
also preserve interval geometry.
Uniformly randomized Markov is generic.  The contribution of GE-betting is
its fraction in
\eqref{eq:efficient-betting-fraction-discrete};
Section~\ref{sec:experiments}
compares both calibrations for the main-text procedures.

\subsection{Experiments}
\label{sec:experiments}

We set \(\delta=.01\) and \(c=1\) for all betting methods, and study nine
bounded distributions, including three low-variance examples.
Appendix~\ref{app:solvency-sensitivity} examines the sensitivity of the
results to the choice of \(c\).
Figure~\ref{fig:original-versus-star} extends the
introductory comparison to all methods, distributions, and calibrations in one
display.  Its first row uses
\(10\leq n\leq10^6\), its second row uses
\(10^2\leq n\leq10^6\), and its third row uses
\(10^3\leq n\leq10^6\).  The figure shows deterministic and uniformly
randomized versions of product betting, STaR betting, GE-betting, and
Gaffke.  For Gaffke, the randomized version is its product-orthant refinement in~\cite{ming2026gaffke}.
Gaffke's interval was the previous empirical
state of the art for finite-sample bounded-data confidence intervals and is
therefore our main independence-based benchmark. The code can be found at \href{https://github.com/DMartinezT/betting}{https://github.com/DMartinezT/betting}. 

Each betting method is shown under deterministic and uniformly randomized
Markov calibration.  Within a betting experiment, one uniform pair is shared
across methods and held fixed over \(m\); the randomized Gaffke refinement
likewise holds its one-sided randomizers fixed throughout inversion.
The STaR curve uses the shared predictable variance estimator in
\eqref{eq:regularized-variance-estimate} and
wealth absorbed at the rejection target, so its inversion is an interval on
every sample path;
the product and GE-betting curves likewise use shared estimators.
Proposition~\ref{prop:shared-scale-betting-interval} gives the common
pathwise argument.  Appendix~\ref{sec:experimental-design} gives the
distributions, replication counts, Gaffke implementation, and numerical
details.

\begin{figure}[tp]
    \centering
    \includegraphics[width=0.96\textwidth]
      {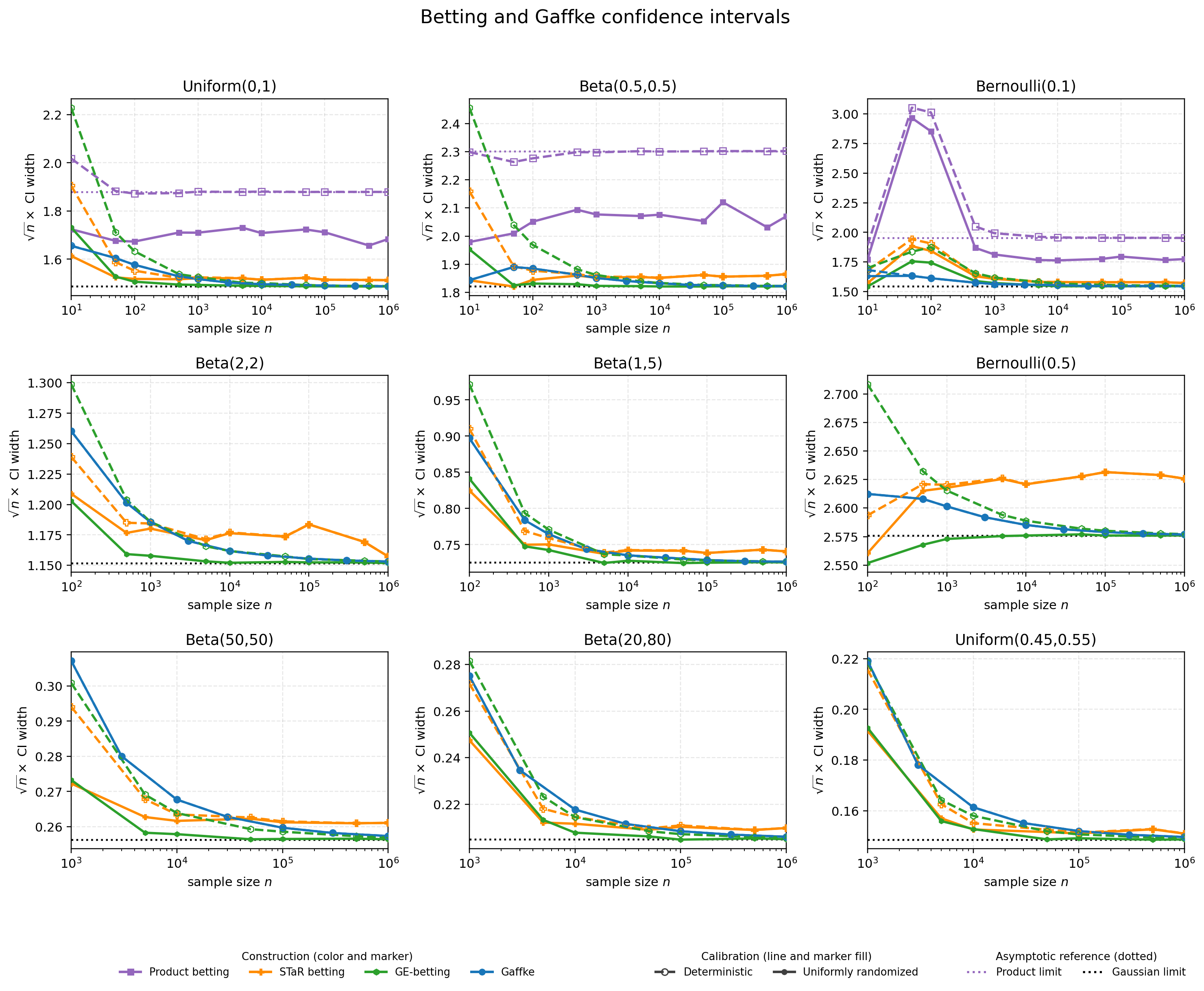}
    \caption{Mean \(\sqrt n\)-scaled confidence-interval widths.  The first
    row covers \(10\leq n\leq10^6\), the second
    \(10^2\leq n\leq10^6\), and the low-variance third row
    \(10^3\leq n\leq10^6\).  Color and marker identify the construction:
    product, STaR, GE-betting, or Gaffke.  Within each construction,
    dashed curves with hollow markers are deterministic and solid curves with
    filled markers are uniformly randomized.  For betting methods,
    randomization means uniformly randomized Markov; for Gaffke, it means the
    product-orthant refinement.  Product betting is shown only in the first
    row, and dotted lines show the Gaussian and product limits.}
    \label{fig:original-versus-star}
\end{figure}

Figure~\ref{fig:original-versus-star} makes the theoretical separation
visible across the full range of distributions and horizons.  Product betting
approaches its wider exponential-test-function limit, STaR betting retains its
efficiency gap, and GE-betting rapidly tracks the Gaussian benchmark.
At \(n=10^6\), uniformly randomized GE-betting is within \(0.11\%\) of
the benchmark in every panel, and deterministic GE-betting is within
\(0.81\%\).

The comparison with Gaffke depends on the GE-betting calibration.  Uniformly
randomized GE-betting has smaller mean width than both Gaffke versions at most
displayed horizons, with exceptions concentrated at shorter horizons, most
notably for Bernoulli\((.1)\); its advantage is largest in the low-variance
panels.  Deterministic GE-betting gives a more mixed comparison: for the first
six distributions it is often wider than Gaffke at shorter horizons and
becomes nearly indistinguishable as \(n\) grows, whereas in the low-variance
panels it is generally narrower.  Thus these experiments suggest that the
martingale-valid randomized GE-betting interval can improve on the independence-based
benchmark, while deterministic GE-betting is better described as competitive with
it, particularly in low-variance settings.

Uniformly randomized Markov calibration provides its largest improvement at
short horizons and becomes nearly indistinguishable from deterministic
calibration as \(n\) grows.  The two Gaffke curves likewise coincide whenever
the ordinary endpoints lie within the sample range; the randomized
product-orthant refinement can shorten them at the smallest horizons.  The
comparisons are descriptive rather than a finite-sample dominance theorem,
but their consistency across shapes, variances, and sample sizes reinforces
the main conclusion: martingale-valid inference need not pay an empirical
width penalty.

\clearpage
\subsection{On Gaussian and large deviation efficiency}
\label{sec:gaussian-ld-efficiency}

Gaussian efficiency fixes the error probability while increasing the sample
size. A complementary benchmark lets the error probability decrease
exponentially: if \(X_i\) are iid from a law \(P\) on \([0,1]\) and
\(n^{-1}\log(1/\delta_n)\to\rho\in(0,\infty)\), the information limit is
\begin{equation}\label{eq:ge-ld-information-limit}
\begin{aligned}
 d(P,m)&=\inf_{Q:\,\E_QX=m}\operatorname{KL}(P\|Q),\\
 J_\rho(P)&=\{m\in[0,1]:d(P,m)\leq\rho\}
 =[l_\rho(P),u_\rho(P)].
\end{aligned}
\end{equation}
Here \(Q\) ranges over all laws on \([0,1]\). The large deviation optimality of
KL-inf intervals is studied by \citet{deep2025asymptotic}. This benchmark
differs from the Gaussian benchmark: convergence of an unscaled width to
zero does not identify its leading constant on the \(n^{-1/2}\) scale.
Both benchmarks can be attained by averaging the underlying wealths.

For \(m\in(0,1)\), define the two constant-fraction mixture wealth processes
\begin{equation}\label{eq:ge-ld-accounts}
\begin{aligned}
 M_t^+(m)&=\int_0^1\prod_{i=1}^t
       \left(1-u+u\frac{X_i}{m}\right)du,\\
 M_t^-(m)&=\int_0^1\prod_{i=1}^t
       \left(1-u+u\frac{1-X_i}{1-m}\right)du.
\end{aligned}
\end{equation}
Use their one-sided limits at \(m=0,1\). For \(h=\log(2/\delta)\), set
\begin{equation}\label{eq:ge-ld-hedged-wealth}
\begin{aligned}
 w_{\rm G}&=\frac{n}{n+h},\qquad w_{\rm K}=\frac{h}{n+h},\\
 H_{t,n}^\pm(m)&=w_{\rm G}K_t^\pm(m;\delta w_{\rm G})
                  +w_{\rm K}M_t^\pm(m).
\end{aligned}
\end{equation}
Here \(K_t^\pm(m;\delta w_{\rm G})\) denotes the GE-betting wealth in
\eqref{eq:efficient-one-sided-fractions}--%
\eqref{eq:efficient-one-sided-updates}, computed with \(\delta w_{\rm G}\)
in place of \(\delta\). The horizon remains \(n\), and the shared variance
estimate is \eqref{eq:regularized-variance-estimate}.
As in \eqref{eq:efficient-betting-interval}, accept a candidate mean
\(m\in[0,1]\) when \(H_{n,n}^+(m)<2/\delta\) and
\(H_{n,n}^-(m)<2/\delta\). Let \(\mathcal I_n^{\rm H}(\delta)\) be the
closed interval whose endpoints are the infimum and supremum of the
accepted means; if no mean is accepted, set
\(\mathcal I_n^{\rm H}(\delta)=\{\bar X_n\}\).
Each \(H_{t,n}^\pm(\mu)\) is a nonnegative test
supermartingale starting from one, and its one-sided dependence on \(m\)
is monotone. Thus inversion remains an interval with coverage at least
\(1-\delta\) under \eqref{eq:conditional-mean-model}.

\begin{proposition}[Gaussian and large deviation efficiency]
\label{prop:ge-gaussian-ld-efficiency}
Let \(X_1,X_2,\ldots\) be iid from \(P\), with mean \(\mu\in(0,1)\)
and variance \(\sigma^2\). If \(\sigma^2>0\), then for each fixed
\(\delta\in(0,1)\),
\begin{align*}
 \inf\mathcal I_n^{\rm H}(\delta)
 &=\bar X_n-\frac{\sigma z_{1-\delta/2}}{\sqrt n}+o(n^{-1/2}),\\
 \sup\mathcal I_n^{\rm H}(\delta)
 &=\bar X_n+\frac{\sigma z_{1-\delta/2}}{\sqrt n}+o(n^{-1/2})
 \qquad\text{almost surely}.
\end{align*}
If instead \(n^{-1}\log(1/\delta_n)\to\rho\in(0,\infty)\), then, including
when \(\sigma^2=0\),
\[
 \inf\mathcal I_n^{\rm H}(\delta_n)
 \xrightarrow{\mathrm{a.s.}}l_\rho(P),\qquad
 \sup\mathcal I_n^{\rm H}(\delta_n)
 \xrightarrow{\mathrm{a.s.}}u_\rho(P).
\]
Moreover, any interval procedure \(C_n\) with coverage at least
\(1-\delta_n\) for every iid law on \([0,1]\) satisfies, under sampling from
\(P\),
\[
 d(P,m)<\rho\quad\Longrightarrow\quad
 \Pb\{m\in C_n\text{ for all sufficiently large }n\}=1.
\]
\end{proposition}

The proof is in Appendix~\ref{proof:ge-gaussian-ld-efficiency}. At fixed
\(\delta\), \(w_{\rm G}\to1\) and the additional mixture wealth is negligible
on the Gaussian scale. In the large deviation regime,
\(w_{\rm K}\to\rho/(1+\rho)>0\), and the constant allocation changes the
logarithmic rejection threshold by only \(o(n)\). Finite-sample validity
continues to allow conditional-mean dependence; independence is used only
for these efficiency conclusions.

\subsection{Constructing p-values from GE-betting}
\label{sec:ge-pvalues}

GE-betting is tuned to a significance level through its wealth-dependent
fractions.  Nevertheless, a \(p\)-value can be computed without rerunning the
strategy over a grid of levels.  Normalizing wealth by its rejection target
removes the level from every update, leaving it only in the initial value.
A backward recursion then identifies the smallest initial value that leads
to rejection.

Fix \(n\), \(m\in[0,1]\), and \(c\in(0,1]\), and use the same variance estimates
\eqref{eq:regularized-variance-estimate} at every level.  Write
\(s_i=\sqrt{(n-i+1)\widehat v_{i-1}}\) and
\(\gamma(p)=\phi\{\Phi^{-1}(p)\}\), with
\(\gamma(0)=\gamma(1)=0\).  The normalized updates are
\begin{equation}\label{eq:ge-pvalue-maps}
\begin{aligned}
 F_i^+(p;m)&=\min\!\left\{1,p+(X_i-m)
       \min\!\left(\frac{\gamma(p)}{s_i},\frac{cp}{m}\right)\right\},\\
 F_i^-(p;m)&=\min\!\left\{1,p+(m-X_i)
       \min\!\left(\frac{\gamma(p)}{s_i},\frac{cp}{1-m}\right)\right\}.
\end{aligned}
\end{equation}
At \(m=0\), omit the upper-tail cap, and at \(m=1\), omit the lower-tail cap;
in all cases set \(F_i^\pm(0;m)=0\) and \(F_i^\pm(1;m)=1\).
For a one-sided level \(a\in(0,1)\), let \(V_0^{a,\pm}=a\) and
\(V_i^{a,\pm}=F_i^\pm(V_{i-1}^{a,\pm};m)\).  This is the GE wealth process
calibrated to target \(1/a\), multiplied by \(a\); its deterministic rejection
event is \(\{V_n^{a,\pm}=1\}\).

\begin{proposition}[Backward inversion of GE-betting]
\label{prop:ge-pvalues}
Initialize \(q_n^\pm(m)=1\) and compute, for \(i=n,\ldots,1\),
\begin{equation}\label{eq:ge-pvalue-backward}
 q_{i-1}^\pm(m)=\inf\{p\in[0,1]:F_i^\pm(p;m)\geq q_i^\pm(m)\}.
\end{equation}
Then \(\pi_n^\pm(m)=q_0^\pm(m)\) satisfies
\(\{\pi_n^\pm(m)\leq a\}=\{V_n^{a,\pm}(m)=1\}\) for every \(a\in(0,1)\).
Under the conditional-mean null at \(m\), these are valid one-sided \(p\)-values,
and
\begin{equation}\label{eq:ge-pvalue-two-sided}
 p_n^{\rm GE}(m)=\min\{1,2\pi_n^+(m),2\pi_n^-(m)\}
\end{equation}
is a valid two-sided \(p\)-value.
\end{proposition}

The proof is given in Appendix~\ref{proof:ge-pvalues}.

Writing \(\mathcal I_n^{\rm GE}(\delta)\) for the deterministic GE interval at level
\(\delta\), the construction gives the exact identity
\begin{equation}\label{eq:ge-pvalue-inversion}
 \{m\in[0,1]:p_n^{\rm GE}(m)>\delta\}=\mathcal I_n^{\rm GE}(\delta),
 \qquad \delta\in(0,1).
\end{equation}
Consequently it retains both interval geometry and the fixed-level Gaussian
efficiency of Theorem~\ref{thm:efficient-betting-efficiency}.
For each candidate \(m\), two backward passes of \(n\) scalar inversions suffice.
Bisection must locate the leftmost crossing, including when the target is
one; \(B\) iterations per inversion require \(O(nB)\) map evaluations,
independent of the number of levels subsequently queried.  This is a
fixed-horizon computation using the stored observations and scales.  It
encodes the family of level-specific wealth processes rather than the
reciprocal of one terminal e-value.

\section{GE-betting under sampling without replacement}
\label{sec:wor-bridge-betting}

We now apply the Gaussian design principle to the finite-population model in
Subsection~\ref{sec:wor-background}.  Under iid sampling, the centered partial
sum has a Brownian-motion limit.  Without replacement, however, revealing the
whole population forces \(S_N-N\mu_N=0\), so the centered process is pinned at
zero at time one and has a Brownian-bridge limit.  Designing the bet in this bridge experiment
therefore accounts for the uncertainty removed as the population is revealed
and produces the square-root correction in
\eqref{eq:intro-wor-fractions}; a Brownian-motion calculation misses this
first-order effect when \(n/N\) does not vanish.

We first derive this update in the Gaussian-bridge experiment, then transfer
it to bounded observations using the exact Doob martingale of the terminal
sample mean.  We next state the finite-sample and asymptotic guarantees and
compare the resulting interval with existing without-replacement methods.
As in Section~\ref{sec:optimal-star-testing}, the Gaussian experiment is only
a device for choosing the betting fractions; the observed population need
not be Gaussian, and finite-sample validity follows from the
remaining-population conditional mean and nonnegative predictable betting.
Appendix~\ref{app:horizon-free-cs} extends the GE-betting construction
to a horizon-free confidence sequence under sampling without replacement.

\subsection{Gaussian-bridge design}
\label{sec:wor-gaussian-design}

Let \((B_t)_{0\leq t\leq1}\) be a standard Brownian bridge and fix a terminal
sampling fraction \(\rho\in(0,1)\).  The standardized terminal statistic is
\begin{equation}
 Z_\rho=\frac{B_\rho}{\sqrt{\rho(1-\rho)}}.
 \label{eq:wor-terminal-statistic}
\end{equation}
If \(\mathcal G_t=\sigma(B_s:s\leq t)\), then
\begin{equation}
 M_t=\E(Z_\rho\mid\mathcal G_t)
 =\frac{1-\rho}{1-t}
  \frac{B_t}{\sqrt{\rho(1-\rho)}},
 \qquad 0\leq t\leq\rho.
 \label{eq:wor-bridge-doob}
\end{equation}
The variance clock of this martingale is
\begin{equation}
 u(t)=\operatorname{Var}(M_t)
 =\frac{(1-\rho)t}{\rho(1-t)},
 \qquad u(\rho)=1.
 \label{eq:wor-bridge-clock}
\end{equation}
Indeed, \(\operatorname{Cov}(M_s,M_t)=u(s)\) for \(s\leq t\), so \(M_t\)
has the same law as Brownian motion evaluated at time \(u(t)\).

Consider the level-\(\delta/2\) upper-tail event
\[
 A^\star=\{Z_\rho\geq z_{1-\delta/2}\}
\]
and its indicator e-value \(2\mathbf 1\{A^\star\}/\delta\).  Conditional on
\(\mathcal G_t\),
\[
 Z_\rho=M_t+\sqrt{1-u(t)}\,\xi_t,
 \qquad
 \xi_t\sim\mathcal N(0,1),\quad \xi_t\perp\mathcal G_t.
\]
Hence its continuation probability is
\begin{equation}
 p_t=\Pb(A^\star\mid\mathcal G_t)
 =\Phi\left\{
   \frac{M_t-z_{1-\delta/2}}{\sqrt{1-u(t)}}\right\}.
 \label{eq:wor-bridge-continuation}
\end{equation}
If \(K_t=2p_t/\delta\), then \(p_t=(\delta/2)K_t\) is current wealth
relative to the rejection threshold.  Differentiating the continuation
probability with respect to the revealed martingale gives
\begin{equation}
 \mathrm d p_t
 =\frac{\phi\{\Phi^{-1}(p_t)\}}{\sqrt{1-u(t)}}\,
   \mathrm dM_t.
 \label{eq:wor-continuation-differential}
\end{equation}
Thus the betting fraction per unit increment in \(M_t\) is
\begin{equation}
 \frac{\psi(p_t)}{\sqrt{1-u(t)}}.
 \label{eq:wor-optimal-fraction}
\end{equation}
The denominator is the standard deviation of the part of the terminal
statistic that has not yet been observed.  Equivalently, if \(H_t\) is an
unstandardized terminal-statistic martingale and
\(V_t=\operatorname{Var}(H_\rho\mid\mathcal G_t)\), then
\begin{equation}
 \mathrm d p_t
 =\frac{\phi\{\Phi^{-1}(p_t)\}}{\sqrt{V_t}}\,
   \mathrm dH_t.
 \label{eq:wor-unstandardized-continuation}
\end{equation}

\subsection{The bounded-data update}
\label{sec:wor-bounded-update}

Fix a horizon \(1\leq n<N\), an error level \(\delta\in(0,1)\), and a
clipping constant \(c\in(0,1]\).  For a candidate population mean \(m\), the
terminal statistic before standardization is
\begin{equation}
 A_n(m)=S_n-nm.
 \label{eq:wor-terminal-contrast}
\end{equation}
The process \(S_t-tm\) is not a martingale under sampling without
replacement.  The appropriate process is instead the conditional
expectation of \(A_n(m)\).  Under the point null \(\mu_N=m\), this Doob
martingale is
\begin{equation}
 H_t^{(n)}(m)
 =\E_m\{A_n(m)\mid\mathcal F_t\}
 =\frac{N-n}{N-t}(S_t-tm),
 \qquad 0\leq t\leq n.
 \label{eq:wor-doob-martingale}
\end{equation}
It terminates at \(H_n^{(n)}(m)=A_n(m)\), and its increments satisfy
\begin{equation}
 H_i^{(n)}(m)-H_{i-1}^{(n)}(m)
 =\gamma_iY_i(m),
 \qquad
 \gamma_i=\frac{N-n}{N-i}.
 \label{eq:wor-doob-increment}
\end{equation}
If the \(N-i+1\) values remaining before draw \(i\) have variance
\[
 v_{i-1}^{\mathrm{rem}}
 =\frac1{N-i+1}\sum_{x\ \mathrm{remaining}}
   \{x-m_i(m)\}^2,
\]
then the conditional variance remaining at the horizon and the resulting
standardization are
\begin{align}
 V_{i-1}
 :=\operatorname{Var}_m\{A_n(m)\mid\mathcal F_{i-1}\}
 &=(n-i+1)\frac{N-n}{N-i}v_{i-1}^{\mathrm{rem}},\nonumber\\
 \frac{\gamma_i}{\sqrt{V_{i-1}}}
 &=\sqrt{\frac{N-n}
 {(N-i)(n-i+1)v_{i-1}^{\mathrm{rem}}}}.
 \label{eq:wor-exact-bridge-scale}
\end{align}

Thus \(\gamma_i\) measures the contribution of the next centered draw to the
terminal statistic, while \(V_{i-1}\) measures the uncertainty still
remaining at the horizon.  Substituting the shared predictable estimator
\(\widehat v_{i-1}\) from
\eqref{eq:regularized-variance-estimate} for
\(v_{i-1}^{\mathrm{rem}}\) gives
\begin{equation}
 b_{i,n}
 =\sqrt{\frac{N-n}
 {(N-i)(n-i+1)\widehat v_{i-1}}}.
 \label{eq:wor-predictable-scale}
\end{equation}

The GE-betting procedure now follows directly from the betting rule in
Section~\ref{sec:introduction}.  On the feasible range \(\mathcal M_n\), use
the raw fractions in \eqref{eq:intro-wor-fractions}, with scale \(b_{i,n}\),
and apply the clipping, absorbing, and wealth-update rules in
\eqref{eq:efficient-one-sided-fractions}--%
\eqref{eq:efficient-one-sided-updates}, replacing \(X_i-m\) and \(m-X_i\)
by \(Y_i(m)\) and \(-Y_i(m)\).  Inverting the terminal wealths gives the
deterministic interval \(\mathcal I_{N,n}^{\mathrm{br}}\) in
\eqref{eq:intro-wor-interval}.  Replacing its two thresholds by
\(2U_+/\delta\) and \(2U_-/\delta\), as in Section~\ref{sec:introduction},
defines the randomized interval
\(\mathcal I_{N,n}^{\mathrm{br}}(U_+,U_-)\); setting \(U_+=U_-=1\) gives the
deterministic version.  Every quantity used in an update is predictable, and
the variance estimator is shared across candidate means.  Hence validity
comes from nonnegative betting rather than the Gaussian approximation, while
the endpoints can be found by bisection over \(\mathcal M_n\).

\subsection{Validity and Gaussian efficiency}
\label{sec:wor-guarantees}

The next result collects the same three guarantees as
Theorem~\ref{thm:efficient-betting-efficiency}: finite-sample validity,
interval-valued inversion, and optimal first-order Gaussian width.  We use
the regularized estimator \eqref{eq:regularized-variance-estimate}, shared by
both wealth processes and all candidate means.

\begin{theorem}[Validity and Gaussian efficiency under sampling without replacement]
\label{thm:wor-main}
Fix \(\delta\in(0,1)\) and \(c\in(0,1]\), and construct the intervals in
Subsection~\ref{sec:wor-bounded-update}.  For every fixed population
\(x_{1:N}\in[0,1]^N\) and every \(n<N\),
\(\mathcal I_{N,n}^{\mathrm{br}}\) is an interval on every sample path, and
\begin{equation}
 \Pb_{x_{1:N}}\{\mu_N\in\mathcal I_{N,n}^{\mathrm{br}}\}
 \geq1-\delta.
 \label{eq:wor-validity}
\end{equation}
The same validity and interval conclusions hold for
\(\mathcal I_{N,n}^{\mathrm{br}}(U_+,U_-)\).  

Now let \(x_{N,1:N}\) be a deterministic triangular array with population
means \(\mu_N\) and variances \(\sigma_N^2\to\sigma^2>0\), and let
\(n=n_N\) satisfy \(n_N/N\to\rho\in(0,1)\).  Couple the uniformly random
permutations of the rows on any common probability space; independence
across rows is not required.  Then
\begin{equation}
 \frac{\operatorname{len}(\mathcal I_{N,n_N}^{\mathrm{br}})}
 {2z_{1-\delta/2}\tau_{N,n_N}}
 \xrightarrow{\mathrm{a.s.}}1.
 \label{eq:wor-efficiency}
\end{equation}
\end{theorem}

\textit{Proof sketch.}
The validity and interval claims follow as in
Theorem~\ref{thm:efficient-betting-efficiency}. 
For efficiency, the Gaussian formula
\eqref{eq:wor-bridge-continuation} in
Subsection~\ref{sec:wor-gaussian-design} again motivates applying
\(\Phi^{-1}\): it recovers the predicted terminal statistic minus its
rejection cutoff, divided by the remaining standard deviation.
For a candidate mean \(m\), let \(K_i=K_i^+(m)\) be upper-tail wealth,
\(p_i=(\delta/2)K_i\) its value divided by the rejection target, and
\(q_i=\Phi^{-1}(p_i)\) while \(0<p_i<1\).

The quantities from Subsection~\ref{sec:wor-bounded-update} account for the
dependence between draws.  Write \(Y_i=X_i-m_i(m)\), where
\(m_i(m)=(Nm-S_{i-1})/(N-i+1)\) is the remaining-population mean implied by
\(m\), and put
\[
 \gamma_i=\frac{N-n}{N-i},\qquad
 r_{i,N}=\frac{(n-i+1)(N-n)}{N-i},\qquad r_{n+1,N}=0.
\]
By \eqref{eq:wor-doob-martingale}--\eqref{eq:wor-doob-increment},
\(\gamma_iY_i\) is the change in the conditional prediction
\(H_i^{(n)}(m)\) of the final centered sum \(S_n-nm\).
Meanwhile, \eqref{eq:wor-exact-bridge-scale} shows that \(r_{i,N}\) times
the remaining-population variance is the conditional variance of that sum
before draw \(i\).  Replacing the remaining-population variance by the
shared estimator \(\widehat v_{i-1}\) gives the scale
\(b_{i,n}=\gamma_i/\sqrt{r_{i,N}\widehat v_{i-1}}\) used in the update.

While the cap is inactive and wealth has not reached a boundary,
\[
 p_i-p_{i-1}
 =\phi(q_{i-1})\frac{\gamma_iY_i}
 {\sqrt{r_{i,N}\widehat v_{i-1}}}.
\]
To use the full-sample scale, we multiply \(q_i\)
by \(\sqrt{r_{i+1,N}/r_{1,N}}\), the counterpart of
\(\sqrt{(n-i)/n}\) in the iid sketch.

The same identity \((\Phi^{-1})'(p)\phi\{\Phi^{-1}(p)\}=1\) as in the iid
sketch, together with a second-order Taylor expansion, therefore gives
\[
\begin{aligned}
 &\sqrt{\frac{r_{i+1,N}}{r_{1,N}}}q_i
  -\sqrt{\frac{r_{i,N}}{r_{1,N}}}q_{i-1}=
 \frac{\gamma_iY_i}{\sqrt{r_{1,N}\widehat v_{i-1}}}
 +\frac{q_{i-1}}{2\sqrt{r_{1,N}r_{i,N}}}
 \left\{
  \frac{\gamma_i^2Y_i^2}{\widehat v_{i-1}}
  -(r_{i,N}-r_{i+1,N})
 \right\}+R_{i,N}.
\end{aligned}
\]
Here \(R_{i,N}\) collects the remaining approximation errors, bounded in
\eqref{eq:wor-proof-expansion-error} after division by \(\sqrt{r_{1,N}}\).
The first term is the change in the predicted final sum, divided by the
estimated standard deviation of the full sample sum.
The key calculation is that
\(r_{i,N}\) decreases by approximately \(\gamma_i^2\) per draw:
\[
 r_{i,N}-r_{i+1,N}=\gamma_i^2\{1+O(N^{-1})\}.
\]
For candidate means close to \(\mu_N\),
\(\E(Y_i^2\mid\mathcal F_{i-1})/\widehat v_{i-1}\approx1\), so the two
terms in braces in the rescaled update have conditional average
approximately \(\gamma_i^2\).
Their difference therefore has approximately zero conditional mean, which
is the cancellation needed to recover the standardized sum.

The proof chooses \(k_N\to\infty\) with \(k_N/n\to0\) and stops the
expansion at \(j_N=n-k_N\), before a single late draw can make the Taylor
approximation inaccurate.  Recall that
\(\tau_{N,n}=\sigma_N\sqrt{(N-n)/\{n(N-1)\}}\) is the standard error of
\(\bar X_n\), so \(n\tau_{N,n}=\sigma_N\sqrt{r_{1,N}}\) is the standard
deviation of \(S_n\).  For local candidates \(m=\mu_N-h\tau_{N,n}\), where
\(h\) measures distance from \(\mu_N\) in standard-error units, summing the
rescaled increments and using \(q_0=\Phi^{-1}(\delta/2)=-z_{1-\delta/2}\)
gives
\[
\begin{aligned}
 \sqrt{\frac{r_{j_N+1,N}}{r_{1,N}}}q_{j_N}
 &=\frac{H_{j_N}^{(n)}(m)}{n\tau_{N,n}}
   -z_{1-\delta/2}+o(1)=\frac{\bar X_n-m}{\tau_{N,n}}
   -z_{1-\delta/2}+o(1)
   \qquad\text{almost surely}.
\end{aligned}
\]
Here the omitted final draws and variance-estimation errors contribute only
\(o(1)\).

The right-hand side is the usual finite-population Gaussian statistic minus
its rejection cutoff.  Thus \(q_{j_N}\) becomes large and positive when
that statistic stays above the cutoff by a fixed amount, and large and
negative when it stays below by a fixed amount.  Direct bounds for the last
\(k_N\) draws turn these two cases into rejection and acceptance,
respectively.  The proof makes the error probabilities summable over \(N\)
on a grid of candidate means, so these decisions hold eventually almost
surely.  The reflected argument handles the lower-tail wealth, and
monotonicity then places the endpoints at
\(\bar X_n\pm z_{1-\delta/2}\tau_{N,n}+o(\tau_{N,n})\) almost surely.
Subtracting them gives \eqref{eq:wor-efficiency}.
Appendix~\ref{proof:wor-main} gives the complete argument.

\subsection{Experiments}
\label{sec:wor-experiments}

We set \(\delta=.01\) and \(c=1\) throughout.  For the fixed-horizon
  experiment, we set \(N=4000\) and use 200 uniformly random reveal orders at
  each \(\rho=n/N\in\{.1,.3,.5,.7,.8,.9\}\).  For the fixed-fraction
  experiment, we set \(\rho=.5\) and vary
  \[
  N\in\{50,100,500,10^3,5\times10^3,10^4,
         5\times10^4,10^5,5\times10^5,10^6\},
  \]
  using 50 reveal orders through \(N=10^4\), 30 at
  \(N\in\{5\times10^4,10^5\}\), and 20 at
  \(N\in\{5\times10^5,10^6\}\).  The nine finite-population families correspond
  to the same distributions as Figure~\ref{fig:original-versus-star}.  For 
  GE-betting and the WSR running intersection, we report both 
  deterministic and uniformly randomized Markov calibration.  At each reveal 
  order and design point, one independent pair \((U_+,U_-)\) is shared by the
  two methods and held fixed over candidate population means.  We leave the  
  remaining bounds in their published form: the empirical Bernstein--Serfling
  interval of \citet{bardenet-maillard-2015} is assembled from several
  concentration events rather than presented as the inversion of one terminal
  e-value, and \citet{shekhar-ramdas-2026} do not discuss a randomized
  calibration.  For Bernoulli populations, exact hypergeometric inversion is a 
  particularly informative baseline because it uses the exact finite-population
  sampling law; those panels also include the finite-alphabet rate interval of  
  \citet{shekhar-ramdas-2026}.  Appendix~\ref{sec:wor-experimental-design} gives
  the population construction and numerical details. Figure~\ref{fig:wor-fixed-horizon} reports mean widths normalized by the
  finite-population Gaussian width \(2z_{.995}\tau_{N,n}\), while
  Figure~\ref{fig:wor-fixed-fraction} reports mean \(\sqrt n\)-scaled widths as
  \(N\) varies at \(\rho=.5\).
 
  Across the six standard populations in the fixed-horizon comparison,  
  GE-betting stays close to the Gaussian finite-population
  benchmark and, under each calibration, is uniformly narrower than the 
  corresponding WSR running intersection.  The fixed-fraction comparison
  reinforces this pattern: as \(N\) increases, the \(\sqrt n\)-scaled
  GE-betting widths approach the Gaussian benchmark and are generally
  narrower than WSR under the matched calibration.  In both figures, the
  randomized curves lie below or coincide with their deterministic counterparts
  because the randomized interval is pathwise contained in the deterministic
  one.  In the low-variance panels, the regularization in
  \eqref{eq:regularized-variance-estimate} is visible at the smallest sampling
  fractions or population sizes, but GE-betting remains much
  sharper than the empirical Bernstein--Serfling and AS intervals and is
  generally narrower than WSR under the matched calibration.  In the Bernoulli
  panels it also closely tracks exact hypergeometric inversion, the natural   
  exact baseline for binary populations.  The finite-alphabet rate interval is
  finite-sample valid for those populations, whereas the continuous-alphabet   
  AS-CI has an asymptotic almost-sure guarantee.  We report these procedures as
  width benchmarks under their respective guarantees.

\begin{figure}[tp]
 \centering
 \includegraphics[width=0.96\textwidth]{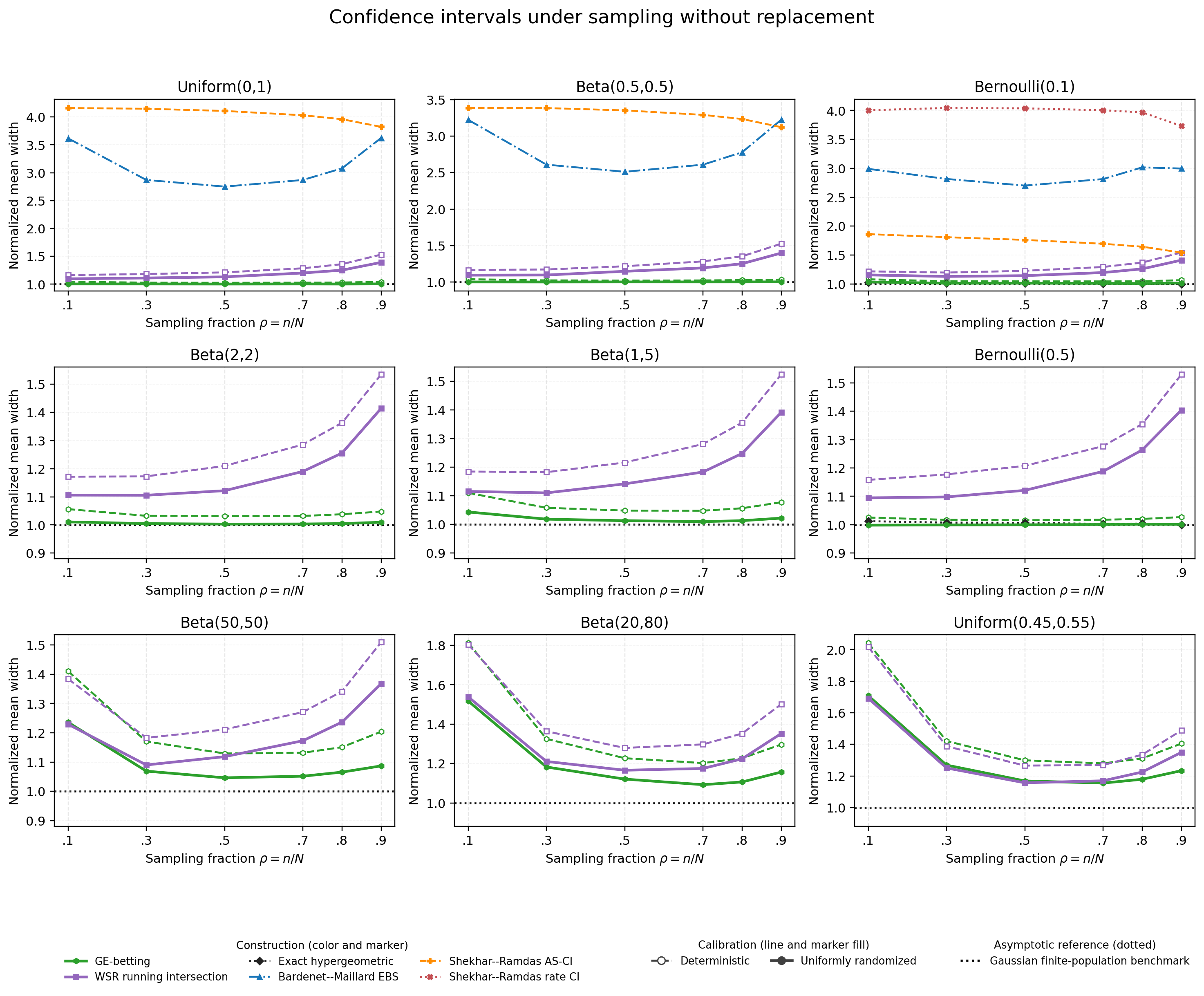}
 \caption{(Fixed $N$, varying $\rho$.)  Normalized mean fixed-horizon width over 200 random reveal orders
 for nine fixed populations.  Color and marker identify the construction.
 For GE-betting and WSR, dashed curves with hollow markers use
 deterministic Markov calibration, while solid curves with filled markers use
 uniformly randomized Markov calibration.  The remaining intervals are shown
 in their published form, and the dotted line is the Gaussian
 finite-population benchmark.  Exact hypergeometric and finite-alphabet rate
 intervals appear only for Bernoulli populations.  Bardenet--Maillard and
 Shekhar--Ramdas curves are omitted from rows two and three because their much
 wider intervals obscure the comparison with WSR.}
 \label{fig:wor-fixed-horizon}
\end{figure}

\begin{figure}[tp]
 \centering
 \includegraphics[width=0.96\textwidth]{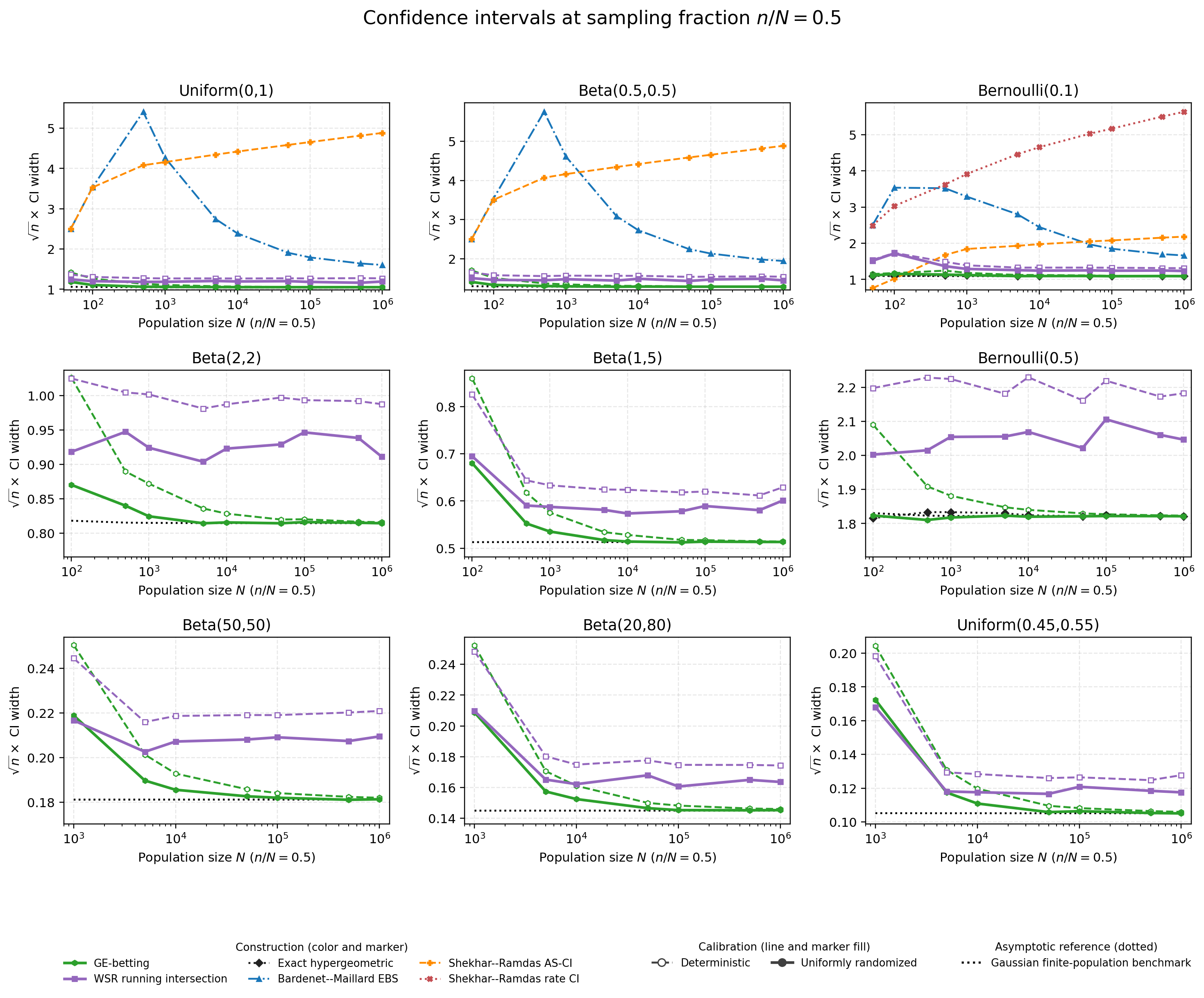}
 \caption{(Fixed $\rho$, varying $N$.) \(\sqrt{n}\)-scaled mean 99\% confidence-interval width at the fixed
  sampling fraction \(n/N=.5\), as the population size \(N\) varies, for nine
  finite-population families.  Means are computed over 50 random reveal orders
  for \(N\leq 10^4\), 30 for \(N\in\{5\times10^4,10^5\}\), and 20 for
  \(N\in\{5\times10^5,10^6\}\).  Color and marker identify the construction.
  For GE-betting and WSR, dashed curves with hollow markers use
  deterministic Markov calibration, while solid curves with filled markers use
  uniformly randomized Markov calibration.  The remaining intervals are shown
  in their published form, and the dotted line is the \(\sqrt{n}\)-scaled
  Gaussian finite-population benchmark.  Exact hypergeometric and
  finite-alphabet rate intervals appear only for Bernoulli populations.
  Bardenet--Maillard and Shekhar--Ramdas curves are omitted from rows two and
  three because their much wider intervals obscure the comparison with WSR.
  For readability, the top, middle, and bottom rows begin at \(N=50\),
  \(N=10^2\), and \(N=10^3\), respectively.}
 \label{fig:wor-fixed-fraction}
\end{figure}

\section{Conclusion}
\label{sec:conclusion}

Testing by betting can deliver more than finite-sample validity: it can produce
fixed-horizon inference that is theoretically efficient and empirically
state-of-the-art.  GE-betting achieves this combination by tracking the
conditional probability of the optimal terminal Gaussian rejection event.
Its clipped predictable updates remain finite-sample valid under the
constant-conditional-mean model, so the coverage guarantee allows martingale
dependence.  Under iid sampling, a predictable variance estimator shared
across candidate means yields an interval pathwise and recovers the
first-order Gaussian width.
Under sampling without replacement, the corresponding Brownian-bridge design
retains these guarantees while incorporating the finite-population correction.

The simulations reinforce the theory.  Across our experiments, GE-betting
achieves the tightest confidence intervals to date for bounded data,
with its largest gains in the low-variance settings.  This is an empirical
comparison rather than a dominance theorem,
but it demonstrates the paper's main practical point: martingale-valid
inference need not be conservative relative to the best independence-based
finite-sample benchmark.

Two lines of research appear particularly promising.  First, this
conditional-e-value construction suggests a broader efficiency theory for
e-values.  For a given local statistical experiment, one would like
to characterize which rejection regions can be tracked without first-order
loss by nonnegative betting processes and obtain matching lower bounds.
Natural extensions beyond the one-dimensional mean include inference on
variance \citep{maurer2009empirical, catoni2012challenging,
martinez2025sharp}, vector-valued means in Euclidean or smooth Banach spaces
\citep{pinelis1994optimum, catoni2018dimensionfree,
martinez2026empirical}, self-normalized constructions in which a random scale
or covariance enters the state \citep{delapena2004selfnormalized,
clerico2025confidence, martinez2025vector}, and matrix-valued means
\citep{tropp2012userfriendly, wang2025sharp,
martinez2026intrinsic}.  These works provide
sharp concentration and nonasymptotic estimation tools for the corresponding
settings; a broader efficiency theory would ask when those tools can be
converted into finite-sample-valid e-values that retain the corresponding
local efficiency, including with nuisance parameters and richer forms of
martingale dependence.  Second, the construction could be extended beyond
bounded observations.  The goal is to replace the known range used in
Section~\ref{sec:bounded-efficient-betting} by conditional moment or tail
assumptions while retaining finite-sample-valid e-values and Gaussian
efficiency.  Possible tools include predictable truncation and mixture
e-values that adapt to unknown scale and tail behavior without changing the
target rejection event at first order.

\section*{Acknowledgements}
DMT thanks Arun Kuchibhotla and Ben Chugg for insightful conversations. AI was employed to assist with the preparation of this paper,  but the authors maintain responsibility for correctness of the claims.

\bibliographystyle{apalike}
\bibliography{bib}

\appendix
\section{Experimental details}
\label{app:experimental-details}

\subsection{Experimental design under sampling with replacement}
\label{sec:experimental-design}

All experiments set \(\delta=.01\) and retain chronological order.  The six
standard distributions are Beta\((2,2)\), Beta\((1,5)\),
Beta\((1/2,1/2)\), Uniform\((0,1)\), Bernoulli\((.5)\), and
Bernoulli\((.1)\); the low-variance distributions are Beta\((50,50)\),
Beta\((20,80)\), and Uniform\((.45,.55)\), with variances approximately
\(.00248\), \(.00158\), and \(.00083\).  Every betting method in the main
experiment is evaluated under deterministic and uniformly randomized Markov
calibration.  Randomized calibration uses one independent uniform for each
one-sided test, held fixed over candidate means and shared across methods
within a dataset.  Product, STaR, and GE-betting all use the solvency
constant \(c=1\).

The betting curves in Figure~\ref{fig:original-versus-star} and the
fixed-hinge comparison in Figure~\ref{fig:fixed-hinge-versus-product} use
\[
 n\in\{10,50,100,500,10^3,5000,10^4,5\cdot10^4,10^5,
 5\cdot10^5,10^6\}.
\]
The first row contains Uniform\((0,1)\), Beta\((1/2,1/2)\), and
Bernoulli\((.1)\), with horizons beginning at \(10\).  The second contains
Beta\((2,2)\), Beta\((1,5)\), and Bernoulli\((.5)\), with horizons
beginning at \(10^2\); the third contains the three low-variance
distributions, with horizons beginning at \(10^3\).  The product-betting
curve in the first row, along with all STaR and GE-betting curves, uses
the predictable variance estimator in
\eqref{eq:regularized-variance-estimate}, shared across candidate means.  In
particular, the plotted STaR implementation keeps the square-root dependence
on current wealth, uses this shared estimator, and holds wealth fixed after it
reaches its one-sided target.  Its inverted acceptance set is therefore an interval for
every dataset.  Figure~\ref{fig:intrinsic-interval-cost} separately
contrasts this implementation with the published candidate-centered STaR
implementation and gives the analogous comparison for GE-betting.

The Gaffke curves in Figure~\ref{fig:original-versus-star} come from a
dedicated experiment in which one simulated sequence supplies all nested
sample sizes.
For the six standard distributions, the
short-horizon points use \(n\in\{10,50,100,500\}\).  From \(n=10^3\)
onward, all nine distributions use
\(n\in\{10^3,3\cdot10^3,10^4,3\cdot10^4,10^5,
3\cdot10^5,10^6\}\).  Ordinary Gaffke is evaluated by normalized B-splines
through \(n=3000\) for continuous data and by a fourth-order
Cornish--Fisher approximation to the exact Dirichlet moments thereafter.
The experiment uses 120 paths per distribution and horizon through \(n=10^3\),
60 through \(n=10^4\), and 30 thereafter, for 6,120 paths per calibration.
The randomized product-orthant endpoints are computed on the same paths with
one saved uniform for each one-sided test, held fixed throughout inversion.
This refinement
can differ from ordinary Gaffke only when an ordinary endpoint lies beyond the
sample range, which occurs mainly at the shortest horizons.  The Gaffke curves
report means.  The saved experiments also retain the pointwise 10th--90th
percentiles of the widths, although Figure~\ref{fig:original-versus-star}
displays only their means.

The betting sample contains 7,560 paired datasets per method and calibration:
120 per distribution and horizon through \(n=10^3\), 60 through \(n=10^4\),
and 30 thereafter.  One-sided processes using a shared estimator
are inverted directly by bisection.  For the auxiliary fixed squared-hinge comparison, we
evaluate decisions on an adaptively refined grid over the full candidate
range.  If this grid finds more than one connected accepted piece, we report
the distance from the smallest to the largest accepted candidate.  The
large-sample comparison in Figure~\ref{fig:intrinsic-interval-cost} uses 30
paired sequences per distribution and horizon for
\(10^3\le n\le10^7\).  Candidate-centered endpoints are first bracketed and
then found by bisection.  For the candidate-dependent GE-betting rule, we also
use the adaptively refined full grid to check for disconnected accepted
pieces.  The two methods using a shared estimator are inverted by bisection
at their ordered one-sided boundaries.

\subsection{Empirical cost of enforcing an interval-valued inversion}
\label{sec:empirical-common-estimator-cost}

Figure~\ref{fig:intrinsic-interval-cost} compares candidate-centered
implementations of STaR and GE-betting with their counterparts using a
shared estimator.  Using a predictable variance estimator shared across
candidate means makes the latter two inversions intervals before
post-processing.  For the candidate-dependent GE-betting rule, we show both
the distance between its smallest and largest accepted grid points and the
length of its largest connected accepted piece.  The candidate-centered STaR
curve reports the connected piece containing the sample mean.

\begin{figure}[!htbp]
    \centering
    \includegraphics[width=0.92\textwidth]
      {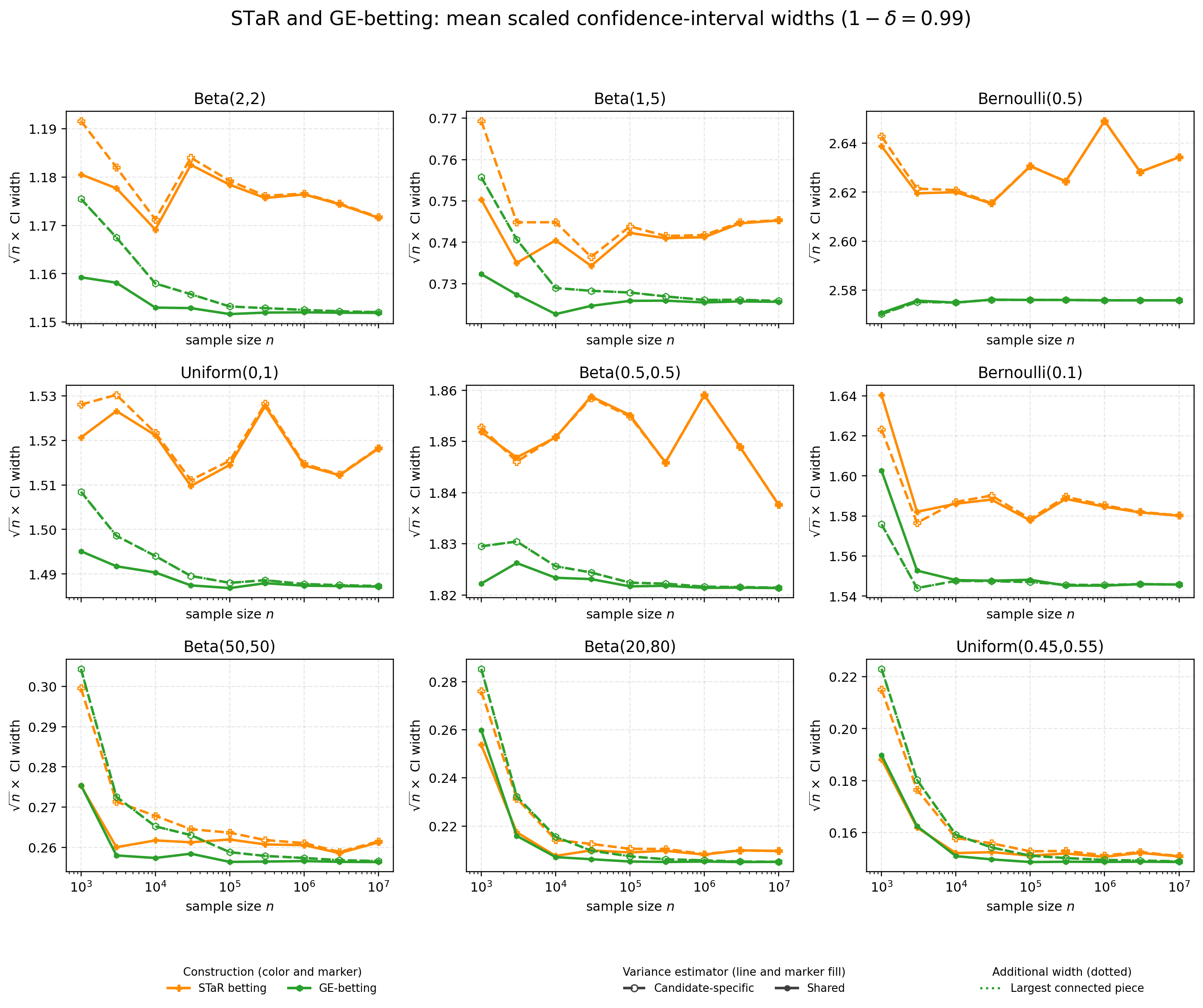}
    \caption{Direct paired comparison of mean \(\sqrt n\)-scaled widths for
    STaR and GE-betting over \(10^3\leq n\leq10^7\).  Color and marker
    identify the betting construction.  Dashed curves with hollow markers use
    candidate-specific variance estimators; solid curves with filled markers
    use an estimator shared across candidate means.  The dotted
    GE-betting curve gives the length of the largest connected accepted
    piece; its dashed counterpart gives the distance between the smallest and
    largest accepted grid points.  Candidate-centered STaR reports the
    connected piece containing the sample mean, while the curves using a
    shared estimator are interval widths for every dataset.  All four
    procedures use uniformly randomized Markov calibration, with the same
    one-sided uniforms held fixed and shared within each dataset.  Every curve
    averages the same 30 datasets per distribution and horizon.}
    \label{fig:intrinsic-interval-cost}
\end{figure}

There is no deterministic ordering between a candidate-dependent strategy and
its counterpart using a shared estimator.  Sharing the estimator across
candidate means changes the predictable variance estimate, not merely the
post-processing of a fixed confidence set, so it can move either endpoint.
Empirically, however, both procedures with shared estimators retain essentially
the same width as their candidate-centered counterparts while guaranteeing
interval geometry on every path.

For GE-betting, the interval using the shared estimator is narrower on
average in
seven of nine panels at \(n=10^3\), with its largest reduction for
Uniform\((.45,.55)\), and the two implementations become nearly
indistinguishable at the largest horizons.  The STaR comparison shows the same
practical message: the shared predictable variance estimator yields a convex
confidence set at little empirical cost.  These comparisons describe average
performance; they do not claim one method is narrower on every dataset.

\subsection{Sensitivity to the e-value-preserving nonnegativity cap}
\label{app:solvency-sensitivity}

The constant \(c\in(0,1]\) determines how close a predictable betting
fraction may come to its one-step nonnegativity boundary.  In the upper-tail
process the cap is \(c/m\), and in the reflected lower-tail process it is
\(c/(1-m)\).  The main experiments use \(c=1\).  To display the empirical
effect of this choice, Figure~\ref{fig:solvency-comparison} repeats the
fixed-sample experiment with \(c\in\{1/2,3/4,1\}\).

The comparison is paired: each value of \(c\) uses the same simulated
sequences, horizons, and replication counts.  Figure~\ref{fig:solvency-comparison}
shows only deterministic product betting, STaR betting, and GE-betting.
Color identifies the betting construction, while line and marker identify
the solvency fraction.

\begin{figure}[!htbp]
    \centering
    \includegraphics[width=0.98\textwidth]
      {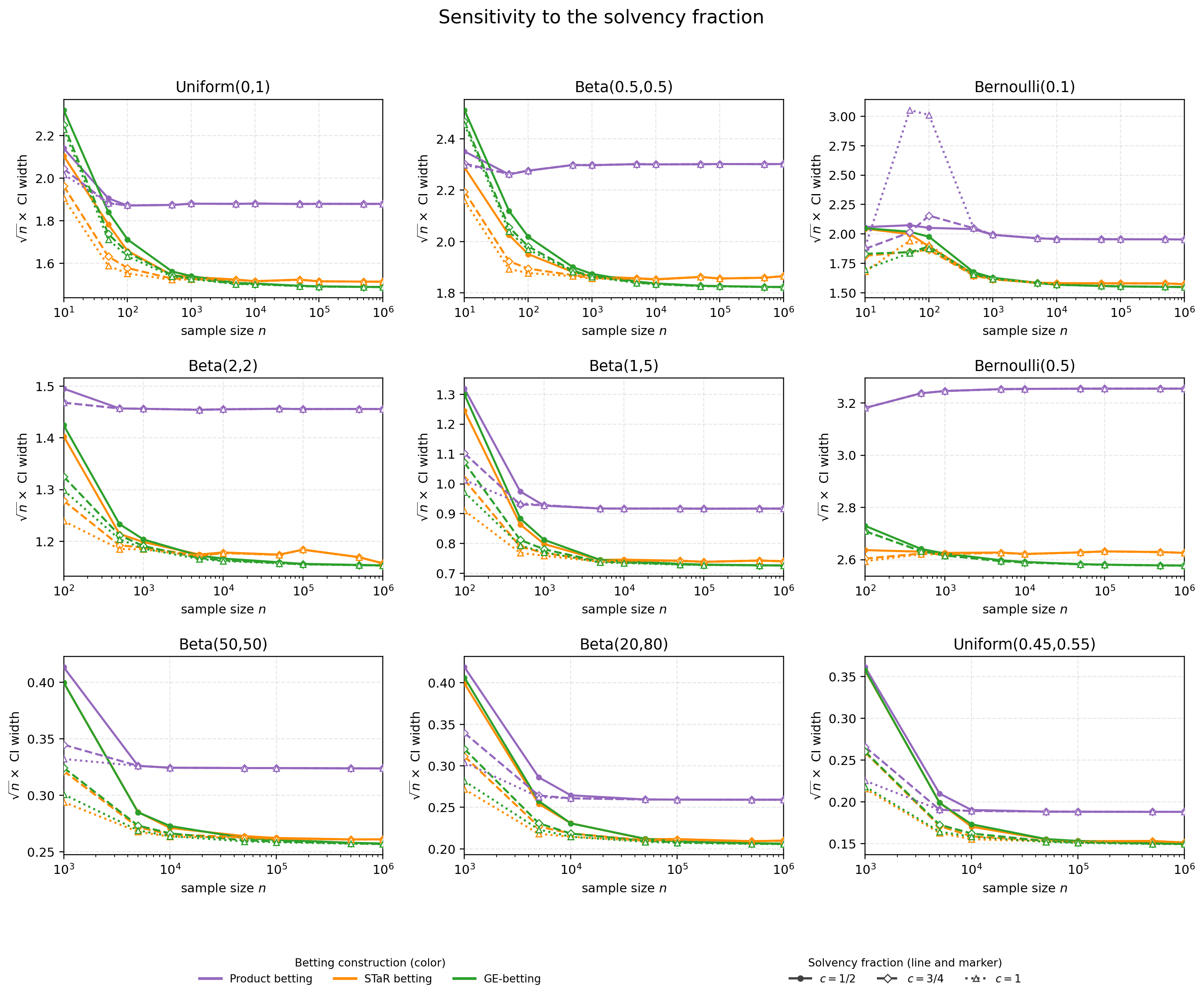}
    \caption{Mean \(\sqrt n\)-scaled widths of deterministic confidence
    intervals for product betting, STaR betting, and GE-betting at
    \(c\in\{1/2,3/4,1\}\).  Color identifies the betting construction; solid
    lines with circles use \(c=1/2\), dashed lines with diamonds use
    \(c=3/4\), and dotted lines with triangles use \(c=1\).}
    \label{fig:solvency-comparison}
\end{figure}

\subsection{Experimental design under sampling without replacement}
\label{sec:wor-experimental-design}

The experiment in Figure~\ref{fig:wor-fixed-horizon} uses populations of size
\(N=4000\) corresponding to the same nine distributions as the main
with-replacement experiment.  Figure~\ref{fig:wor-fixed-fraction} uses the
same deterministic construction separately at each population size.  For
each continuous distribution, the finite population consists of its
\((j-1/2)/N\) quantiles, \(j=1,\ldots,N\).  The Bernoulli populations contain
exactly the corresponding proportions of zeros and ones.  Thus both
experiments fix each population and average only over the random reveal
order; neither introduces a second layer of iid sampling from the named
distribution.

We set \(\delta=.01\) and \(c=1\), and use 200 uniformly random reveal orders
at each sampling fraction
\[
 \rho=n/N\in\{.1,.3,.5,.7,.8,.9\}.
\]
For Figure~\ref{fig:wor-fixed-fraction}, we instead fix \(\rho=.5\) and use
\[
 N\in\{50,100,500,10^3,5\times10^3,10^4,
        5\times10^4,10^5,5\times10^5,10^6\}.
\]
We use 50 reveal orders for \(N\leq10^4\), 30 for
\(N\in\{5\times10^4,10^5\}\), and 20 for
\(N\in\{5\times10^5,10^6\}\).  For GE-betting and the
running-intersection WSR product
strategy, we report both deterministic and uniformly randomized Markov
calibration.  For each population, reveal order, and design point, the
randomized versions use one independent pair
\(U_+,U_-\sim\operatorname{Unif}(0,1)\), shared across the two betting methods
and held fixed throughout inversion over the
candidate mean.  The permutation and randomizer streams are generated
separately from the master seed (20260813).  Betting-interval endpoints are
located by 30 bisection steps for Figure~\ref{fig:wor-fixed-horizon} and 24 for
Figure~\ref{fig:wor-fixed-fraction}.  We retain the published
calibrations for the empirical Bernstein--Serfling interval of
\citet{bardenet-maillard-2015} and the continuous-alphabet AS-CI of
\citet{shekhar-ramdas-2026}: the former combines several concentration events
rather than inverting one explicit terminal e-value, while the latter paper
does not discuss randomized calibration.  For the two Bernoulli populations,
we additionally report exact equal-tail hypergeometric inversion and the
finite-alphabet rate interval of \citet{shekhar-ramdas-2026}.  These last two
procedures use the binary population structure and are therefore omitted from
the other seven panels.  Every interval is intersected with the feasible range
\(\mathcal M_n\).  In Figure~\ref{fig:wor-fixed-horizon}, each width is divided
by \(2z_{.995}\sigma_N\sqrt{(N-n)/\{n(N-1)\}}\).  Figure~%
\ref{fig:wor-fixed-fraction} instead plots the Monte Carlo mean of
\(\sqrt n\) times the width; its dotted benchmark is
\(2z_{.995}\sigma_N\sqrt{(N-n)/(N-1)}\).  The saved data contain every
pathwise width as well as the displayed means and their Monte Carlo standard
errors.

For a binary population, the exact hypergeometric continuation probability
also gives a useful interpretation of GE-betting.  Suppose
\(M=Nm\) is an integer, let \(A\) be an upper-tail terminal event, and write
\(p_i=\Pb_m(A\mid\mathcal F_i)\).  If \(p_i^{(1)}\) and \(p_i^{(0)}\) are the
continuation probabilities after the next observation is one or zero, then
\begin{equation}
 p_{i-1}
 =m_i(m)p_i^{(1)}+\{1-m_i(m)\}p_i^{(0)}
 \label{eq:wor-binary-continuation}
\end{equation}
and
\begin{equation}
 p_i=p_{i-1}
 +\{p_i^{(1)}-p_i^{(0)}\}\{X_i-m_i(m)\}.
 \label{eq:wor-binary-continuation-update}
\end{equation}
The exact betting fraction is consequently
\begin{equation}
 \ell_i^{\mathrm{exact}}
 =\frac{p_i^{(1)}-p_i^{(0)}}{p_{i-1}}.
 \label{eq:wor-exact-binary-fraction}
\end{equation}
GE-betting replaces the hypergeometric continuation probability
by its Gaussian approximation and the finite difference in
\eqref{eq:wor-exact-binary-fraction} by the derivative in
\eqref{eq:wor-continuation-differential}.  This explains why the new interval
can approach exact hypergeometric inversion while remaining applicable to
arbitrary populations in \([0,1]^N\).

\section{Near-optimal concentration and squared-hinge betting}
\label{app:squared-hinge-methods}

Product betting may be understood as tracking the conditional expectation of
an exponential terminal e-value.  The exponential is a smooth approximation
to the rejection indicator, not the indicator itself.  After normalization at the rejection
boundary, it majorizes that indicator and has an explicit null expectation.
Conditioning on the current Gaussian partial sum and differentiating gives the
next predictable betting amount.  This convenience has a cost: the exponential is
positive far below the rejection boundary and continues to grow far above it,
so its null expectation assigns weight to outcomes that do not affect the
test decision.

The conditional-e-value calculation is not tied to the exponential.  We can
replace it by a test function closer to the indicator, take its conditional
expectation in the same Gaussian design experiment, and transfer the
derivative of that expectation to bounded data.  Powered hinges interpolate
between an indicator and an exponential; the squared hinge is the simplest smooth
choice that retains explicit Gaussian conditional expectations.  This appendix
develops that construction and compares its finite-sample performance with
product betting.

\paragraph{Related work on near-optimal concentration.}
Near-optimal concentration replaces exponential Chernoff majorants by sharper
comparison inequalities and positive-part moments
\citep{talagrand1995missing,bentkus2002remark,bentkus2003inequality,
bentkus2004hoeffding,pinelis2006binomial,bentkus2006domination,
pinelis2006normal,pinelis2014bennett,kuchibhotla2024missing}.
\citet{kuchibhotla2021near} obtain near-optimal confidence sequences from
Bentkus bounds, while \citet{martinez2026bentkus} use the same ideas to remove
an analogous constant-factor loss for asymptotic e-values.  Here the squared
hinge is used only to choose a predictable betting amount.  The
resulting e-value is validated directly under the conditional-mean model,
rather than through the comparison inequality.

\subsection{Powered hinges and their Gaussian benchmark}
\label{sec:near-optimal-background}
Near-optimal concentration can be understood through the test-function form
of Markov's inequality.  If \(f\) is nonnegative and nondecreasing, then
\begin{equation}\label{eq:test-function-majorant}
 \mathbf 1\{Y\geq x\}\leq\frac{f(Y)}{f(x)},
 \qquad
 \Pb\{Y\geq x\}\leq\frac{\E f(Y)}{f(x)}.
\end{equation}
The test function majorizes the rejection indicator and has a tractable
expectation.  Chernoff's choice \(f(y)=e^{\lambda y}\) is convenient but loose:
it is positive far below the boundary and continues to grow far above it.
A tractable majorant closer to the indicator can therefore produce sharper
inference.  The same principle applies when an e-value is calibrated by
Markov's inequality, or an e-process by Ville's inequality.

Bentkus--Pinelis inequalities use \emph{\(\alpha\)-powered hinges}, where
\(\alpha\) is a shape parameter rather than an error probability:
\[
 f_{\alpha,t}(y)=
 \left(1+\frac{y-t}{\alpha}\right)_+^\alpha,
 \qquad 0<\alpha<\infty.
\]
This is a rescaled positive-part power with truncation point
\(a=t-\alpha\).  Its two endpoints are
\[
 f_{0,t}(y)=\mathbf 1\{y\geq t\},
 \qquad
 f_{\infty,t}(y)=e^{y-t},
\]
where the second limit follows from
\((1+u/\alpha)^\alpha\to e^u\).  Finite \(\alpha\) interpolates between the
indicator and exponential: it vanishes below a truncation point but retains
smooth polynomial growth above it.  This family underlies near-optimal bounds
of Bentkus and Pinelis
\citep{bentkus2004hoeffding,pinelis2006binomial,bentkus2006domination}.

The corresponding Gaussian benchmark makes the improvement quantitative.
For \(Z\sim\mathcal N(0,1)\), let
\[
 I_\alpha(a)=\E[(Z-a)_+^\alpha],
 \qquad
 E_{\alpha,a}(Z)=\frac{(Z-a)_+^\alpha}{I_\alpha(a)}.
\]
Thresholding this normalized hinge at \(2/\delta\) rejects when
\begin{equation}\label{eq:powered-hinge-cutoff}
 Z\geq U_{\alpha,\delta/2}(a)
 :=a+\left\{\frac{2I_\alpha(a)}{\delta}\right\}^{1/\alpha},
 \qquad
 q_{\alpha,\delta/2}:=\inf_{a\in\mathbb R}
 U_{\alpha,\delta/2}(a).
\end{equation}
\citet{martinez2026bentkus} show that the optimized cutoff is near-optimal:
\begin{equation}\label{eq:powered-hinge-gaussian-benchmark}
 z_{1-\delta/2}
 \leq q_{\alpha,\delta/2}
 \leq z_{1-\delta/(2c_\alpha)},
 \qquad
 c_\alpha=e^\alpha\alpha^{-\alpha}\Gamma(\alpha+1),
\end{equation}
with the indicator endpoint \(\alpha=0\) attaining the exact Gaussian
quantile.  Consequently, inversion in a Gaussian mean experiment has
first-order two-sided width
\begin{equation}\label{eq:powered-hinge-gaussian-width}
 \frac{2\sigma q_{\alpha,\delta/2}}{\sqrt n}.
\end{equation}
Subsection~\ref{sec:main-construction} specializes this construction to
\(\alpha=2\).  Its bounded-data strategy has the same Gaussian cutoff, and
Figure~\ref{fig:fixed-hinge-versus-product} illustrates the resulting width.

\subsection{A squared-hinge terminal e-value and an exact bounded-data 
strategy}
\label{sec:main-construction}
\label{sec:heat-flow-construction}
\label{sec:heat-asymptotics}

The exponential test function is tractable, but Markov's (or Ville's)
inequality uses it only as a majorant of the rejection indicator.  It assigns
positive value even far below the rejection region and can therefore
spend null expectation where it does not help the test.  The powered-hinge
construction above suggests replacing it by the squared hinge
\(h_a(z)=(z-a)_+^2\), which is zero below its truncation point \(a\) and
increases only when the Gaussian statistic exceeds \(a\), the side on which the test can
reject.  This gives the fixed-horizon connection with near-optimal testing by
betting: choose a
fixed-horizon e-value closer to the indicator, then construct a predictable
bounded-data betting strategy from its conditional expectation.  Define
\[
 I_2(a)=\E h_a(Z),\qquad
 U_{2,\delta/2}(a)=a+\sqrt{2I_2(a)/\delta}.
\]
Then \(h_a(Z)/I_2(a)\) is a Gaussian e-value whose level-\(\delta/2\) cutoff is
\(Z\geq U_{2,\delta/2}(a)\).  Here \(I_2(a)\) is its normalizing constant,
while \(U_{2,\delta/2}(a)\) is the standardized Gaussian cutoff for rejection.
This is the \(\alpha=2\) instance of \eqref{eq:powered-hinge-cutoff}.  We
therefore choose
\begin{equation}\label{eq:heat-lambda-star}
 a^*_{\delta/2}\in\arg\min_{a\geq0}U_{2,\delta/2}(a).
\end{equation}
The optimization selects the truncation point that reaches the rejection
threshold with the smallest Gaussian cutoff.  We use the resulting Gaussian
e-value to choose a predictable betting amount, and then validate the bounded-data
e-value directly.

With current Gaussian partial sum \(x\) and remaining variance \(v\), its
conditional expectation and spatial derivative are
\begin{align}
    U_a(v,x)&=\E[(x+\sqrt vZ-a)_+^2],\nonumber\\
    D_a(v,x)&=\partial_xU_a(v,x)
      =2\{(x-a)\Phi(d)+\sqrt v\phi(d)\},
      \qquad d=(x-a)/\sqrt v.
    \label{eq:hinge-price-delta-short}
\end{align}
Although the fixed-horizon hinge is nonsmooth, averaging over the unobserved
Gaussian increment smooths its conditional expectation.  The standardized distance
\(d=(x-a)/\sqrt v\) records how far the current partial sum is from the truncation
point in units of remaining uncertainty.  Consequently, \(D_a(v,x)\) is small
when this point is remote and increases as the data make rejection plausible.
For standard Brownian motion \((W_t)_{0\leq t\leq1}\), the martingale representation gives
\[
 I_2(a)+\int_0^tD_a(1-s,W_s)\,\mathrm dW_s=U_a(1-t,W_t).
\]
Hence \(D_a\) is the predictable betting amount per unit Gaussian increment
associated with the squared-hinge e-variable.
The finite-sample construction replaces the Brownian partial sum and variance by
predictable bounded-data counterparts.  Let \(\gamma_i\) be the scale of the
next centered increment, \(v_{i-1}\) the variance still available before the
horizon, and \(S_{i-1}=\sum_{j<i}\gamma_j(X_j-m)\) the scaled partial sum already observed.
The chain rule then gives the unclipped upper-tail betting amount
\begin{align*}
    \widetilde b_i=\gamma_iD_a(v_{i-1},S_{i-1}).
\end{align*}
Starting from
\(K_0=I_2(a)\), and for a clipping constant \(0<c\leq1\), we use
\[
 b_i=\min\{\widetilde b_i,cK_{i-1}/m\},
 \qquad K_i=K_{i-1}+b_i(X_i-m),
\]
and reflect the construction for the lower tail.  The normalized terminal
wealth \(K_n/I_2(a)\) is the resulting one-sided e-value.  The Gaussian calculation
determines how much to bet, whereas clipping prevents a
worst-case bounded jump from making wealth negative.  Clipping changes only
the amount actually bet; it does not alter the scaled partial
sum used to compute later betting amounts.

The same Gaussian calculation suggests the following heuristic first-order
prediction: under consistent studentization, vanishing increments,
and an asymptotically inactive clip, one expects for large \(n\) that (the squared hinge limit is)
\begin{equation}\label{eq:fixed-hinge-width}
 \sqrt n\,\operatorname{len}(C_{n,\bar X}^{\rm hinge})
 \approx 2\sigma U_{2,\delta/2}(a),
 \qquad
 U_{2,\delta/2}(a)=a+\sqrt{\frac{2I_2(a)}{\delta}},
\end{equation}
where \(C_{n,\bar X}^{\rm hinge}\) is the accepted component
containing \(\bar X_n\).
At \(a=a^*_{\delta/2}\), the right-hand side is the predicted optimized
Bentkus--Pinelis constant.  This local prediction concerns the component
containing the empirical mean and does not exclude remote accepted components.

\begin{figure}[!htbp]
    \centering
    \includegraphics[width=0.92\textwidth]
      {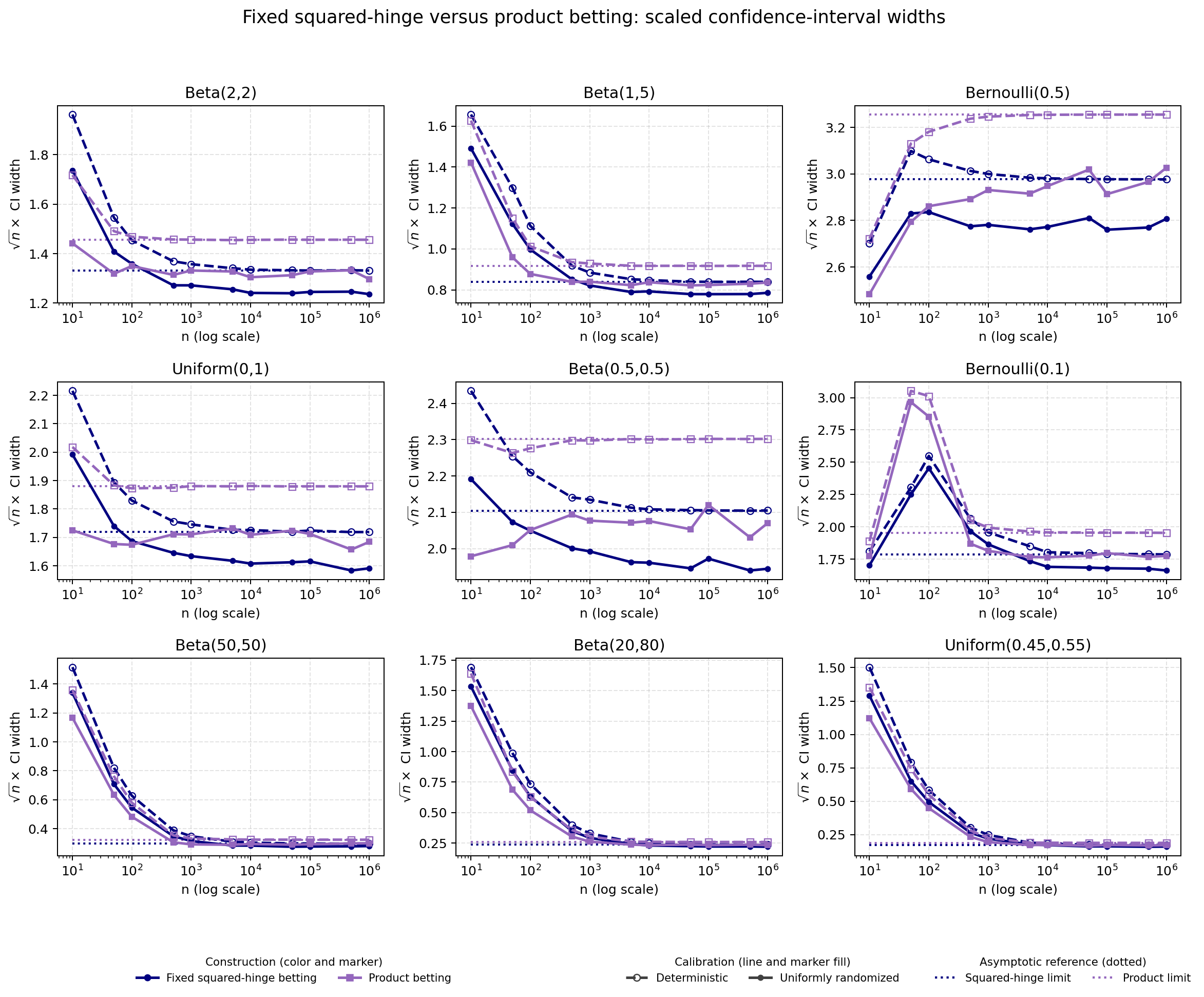}
    \caption{Mean \(\sqrt n\)-scaled confidence-interval widths for product
    betting and fixed squared-hinge betting over the nine distributions in
    the main experiment.  Color and marker identify the betting construction;
    dashed curves with hollow markers use deterministic Markov calibration,
    and solid curves with filled markers use uniformly randomized Markov
    calibration.  Dotted lines show the product Gaussian limit and the
    heuristic squared-hinge Gaussian prediction.}
    \label{fig:fixed-hinge-versus-product}
\end{figure}

Figure~\ref{fig:fixed-hinge-versus-product} compares the present
squared-hinge construction with product betting from
Section~\ref{sec:recovering-testing-by-betting}.  The squared-hinge strategy
empirically approaches the sharper prediction in
\eqref{eq:fixed-hinge-width}, whereas product betting approaches its
exponential-test-function limit.  The comparison is descriptive: both methods
retain their finite-sample validity under the conditional-mean model.

\section{Planned-window confidence sequences under sampling with replacement}
\label{app:horizon-free-wr-cs}

The fixed-horizon interval in Subsection~\ref{sec:bounded-efficient-betting} is
designed to be narrow at one chosen sample size.  Suppose instead that a
maximum sample size \(n\) is known, but the interval may be inspected or the
experiment may stop at any \(t\leq n\).  We want simultaneous coverage over
this entire window without concentrating the design at \(t=n\).  The
construction below places the fixed-horizon GE-betting rule at a geometric
collection of checkpoints and combines the resulting accounts into one
e-process.  In this sense it is horizon-free within the planned window; the
known upper bound makes the portfolio finite, but is not the target maturity
of the portfolio.  The construction applies to every finite \(n\), however
large.

\subsection{Construction}
\label{app:horizon-free-wr-construction}

Let \(X_1,\ldots,X_n\in[0,1]\) satisfy the constant conditional-mean model
\eqref{eq:conditional-mean-model}.  Choose checkpoints
\(h_0<\cdots<h_L=n\) approximately geometrically in ordinary time.  More
precisely, for \(\eta>1\), set
\begin{equation}
 h_{k+1}=\min\{n,\max\{h_k+1,\lceil\eta h_k\rceil\}\}.
 \label{eq:horizon-free-wr-grid}
\end{equation}
This schedule has \(O\{\log(n/h_0)\}\) accounts, so enlarging the planned
window increases the work per update only logarithmically.
Fix deterministic weights \(w_0,\ldots,w_L>0\) with
\(\sum_kw_k=1\), write \(T=2/\delta\), and, for every candidate mean
\(m\in[0,1]\), initialize
\[
 A_{0,k}^+(m)=A_{0,k}^-(m)=w_k.
\]

Account \(k\) follows the fixed-horizon update from Subsection~%
\ref{sec:bounded-efficient-betting}, with \(h_k\) in place of \(n\).  While
\(i\leq h_k\) and \(0<A_{i-1,k}^{\pm}(m)<T\), define
\begin{align}
 p_{i-1,k}^{\pm}(m)&=\frac{\delta}{2}A_{i-1,k}^{\pm}(m),
 &
 b_{i,k}&=\frac{1}{\sqrt{(h_k-i+1)\widehat v_{i-1}}},
 \label{eq:horizon-free-wr-scale}\\
 \ell_{i,k}^+(m)&=\min\left\{
 \psi\{p_{i-1,k}^+(m)\}b_{i,k},\frac{c}{m}\right\},
 &
 \ell_{i,k}^-(m)&=\min\left\{
 \psi\{p_{i-1,k}^-(m)\}b_{i,k},\frac{c}{1-m}\right\}.
 \label{eq:horizon-free-wr-fractions}
\end{align}
Here \(\widehat v_{i-1}\) is the shared predictable estimator in
\eqref{eq:regularized-variance-estimate}.  As in the main construction, the
upper cap is absent at \(m=0\), and the lower cap is absent at \(m=1\).  The
wealth updates are
\begin{align}
 A_{i,k}^+(m)&=\min\left[T,A_{i-1,k}^+(m)
 \{1+\ell_{i,k}^+(m)(X_i-m)\}\right],\nonumber\\
 A_{i,k}^-(m)&=\min\left[T,A_{i-1,k}^-(m)
 \{1+\ell_{i,k}^-(m)(m-X_i)\}\right].
 \label{eq:horizon-free-wr-updates}
\end{align}
An account is held fixed after time \(h_k\), or after its wealth reaches zero
or \(T\).  Sum the accounts within each one-sided test,
\[
 K_t^+(m)=\sum_{k=0}^L A_{t,k}^+(m),
 \qquad
 K_t^-(m)=\sum_{k=0}^L A_{t,k}^-(m),
\]
and invert them:
\begin{equation}
 J_t=\{m\in[0,1]:K_t^+(m)<T,\ K_t^-(m)<T\},
 \qquad
 C_t=\bigcap_{r=1}^tJ_r.
 \label{eq:horizon-free-wr-inversion}
\end{equation}
The raw intervals \(J_t\) already form a confidence sequence.  The running
intersections \(C_t\) are an optional nested version.

\subsection{Why the construction works}
\label{app:horizon-free-wr-guarantee}

Under iid sampling with variance \(\sigma^2\), the pointwise standard error at
time \(t\) is \(\sigma/\sqrt t\).  Thus ordinary time \(t\), rather than the
bridge clock \(t/(N-t)\) from Appendix~\ref{app:horizon-free-cs}, indexes the
amount of information.  A geometric grid in \(t\) consequently places the
accounts at geometric variance scales.  Each account is locally designed for
one such scale, while the portfolio covers every intermediate time.  The
largest checkpoint \(n\) only closes the finite schedule: unlike the
fixed-horizon interval, the construction does not give all initial wealth to
an account designed for \(n\).

\begin{proposition}[Anytime validity and interval inversion]
\label{prop:horizon-free-wr-validity}
For any adapted observations \(X_1,\ldots,X_n\in[0,1]\) satisfying
\eqref{eq:conditional-mean-model},
\begin{equation}
 \Pb\{\mu\in J_t\text{ for every }t=1,\ldots,n\}\geq1-\delta.
 \label{eq:horizon-free-wr-validity}
\end{equation}
The same guarantee holds for \((C_t)_{t=1}^n\).  Moreover, \(J_t\) and \(C_t\)
are intervals on every sample path.
\end{proposition}

\begin{proof}
At \(m=\mu\), the centered increments in
\eqref{eq:horizon-free-wr-updates} have conditional mean zero.  The solvency
caps make every account nonnegative, and stopping or truncating an account at
\(T\) preserves the supermartingale inequality.  Indeed, before stopping,
conditional Jensen's inequality gives
\[
 \E\!\left[\min\{T,A(1+\ell Y)\}\mid\mathcal F_{i-1}\right]
 \leq \min\{T,A\}=A
\]
whenever \(A<T\) and \(\E[Y\mid\mathcal F_{i-1}]=0\).  Each account is
therefore a nonnegative supermartingale starting from \(w_k\).  Consequently,
\(K_t^+(\mu)\) and \(K_t^-(\mu)\) are test supermartingales, each starting
from one.  Ville's inequality and a union bound give
\[
 \Pb\left\{\sup_{t\leq n}K_t^+(\mu)\geq T\right\}
 +\Pb\left\{\sup_{t\leq n}K_t^-(\mu)\geq T\right\}
 \leq\frac2T=\delta,
\]
which proves \eqref{eq:horizon-free-wr-validity}.  Taking running
intersections does not change the event that the true mean is ever excluded.

For the interval claim, the amount bet relative to the rejection target is
\(p\psi(p)=\phi\{\Phi^{-1}(p)\}\), a nonnegative concave function of \(p\).
The predictable scale \(b_{i,k}\) is shared across candidate means.  The
argument of Proposition~\ref{prop:shared-scale-betting-interval}, applied to
each account, shows that every upper-tail wealth is nonincreasing in \(m\)
and every lower-tail wealth is nondecreasing.  These orders are preserved by
summation, so \(J_t\) is an interval; the intersection \(C_t\) is also an
interval.
\end{proof}

\subsection{Experiment}
\label{app:horizon-free-wr-experiment}

We use \(n=10{,}000\), \(\delta=.05\), and 40 iid paths from each of the nine
laws in Figure~\ref{fig:original-versus-star}.  The proposed procedure takes
\(h_0=5\), \(\eta=1.3\), uniform weights over the resulting 29 checkpoints,
and \(c=1\).  We compare
it with three modern confidence sequences for bounded means.  Hedged-CS and
aGRAPA use the predictable plug-in strategies of
\citet{waudby2024estimating}, both with \(c=1/2\), prior mean \(1/2\), prior
variance \(1/4\), and one initial pseudo-observation.  For Hedged-CS, the
unclipped one-sided fraction is
\begin{equation}
 \ell_i^{\rm H}=\sqrt{\frac{2\log(2/\delta)}
 {\widehat v_{i-1}i\log(i+1)}}.
 \label{eq:horizon-free-wr-hedged-fraction}
\end{equation}
We also include PRECiSE-A-CO96, the constant-time universal-portfolio
approximation of \citet{orabona2024tight}, using the authors' Cover--Ordentlich
regret bound.  These competitors represent predictable plug-in and universal-
portfolio approaches that are among the sharpest current confidence sequences
for bounded means.

The figure reports the raw current-time intervals for every method; running
intersections can only shorten them.  For aGRAPA we report the convex hull of
the numerically inverted acceptance set, which preserves coverage if the set
is disconnected.  All coverage diagnostics monitor the true mean at every
time \(1,\ldots,n\), rather than only at the displayed times.  Width is divided
by the pointwise Gaussian reference
\(2z_{.975}\sigma/\sqrt t\).  This denominator is a scale reference, not a
time-uniform interval.

\begin{figure}[!htbp]
 \centering
 \includegraphics[width=0.98\textwidth]{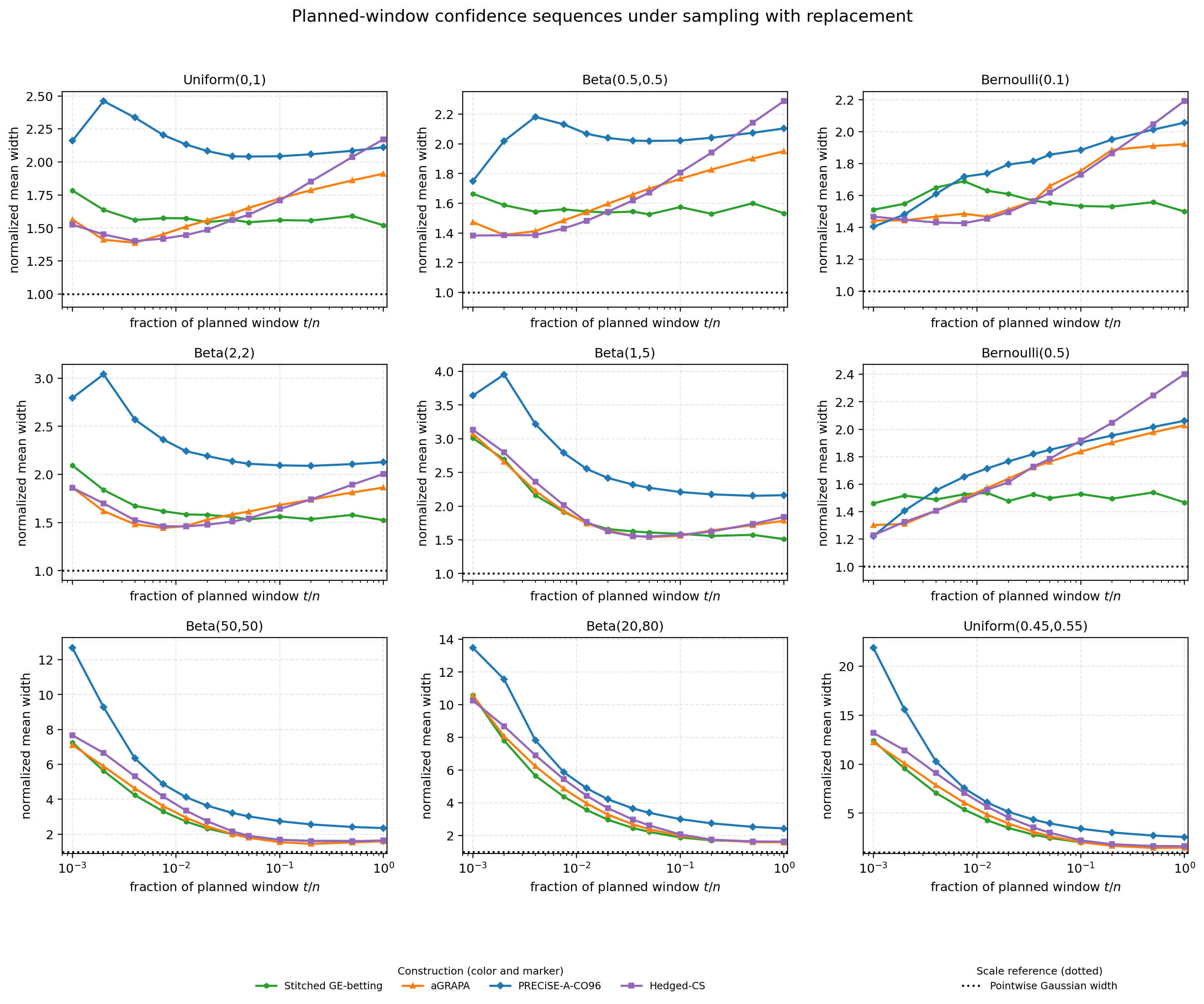}
 \caption{Normalized mean confidence-sequence width over 40 iid paths for the
 nine distributions in Figure~\ref{fig:original-versus-star}.  The
 horizontal axis is the fraction \(t/n\) of the known planning window.  Color
 and marker identify the construction, and the dotted line is the pointwise
 Gaussian width used only as a scale reference.  Stitched GE-betting uses
 \(c=1\); aGRAPA and Hedged-CS use \(c=1/2\).}
 \label{fig:horizon-free-wr-cs}
\end{figure}

Figure~\ref{fig:horizon-free-wr-cs} shows that distributing initial wealth
over checkpoint-specific GE-betting accounts is particularly effective after the
earliest part of the window.  aGRAPA or Hedged-CS is often narrower at the
first few displayed times, but stitched GE-betting overtakes them as
information accumulates.  For the six standard-variance laws, it has the
smallest terminal mean width by \(15\%\)--\(28\%\).  It is also slightly
narrower than aGRAPA for the two low-variance Beta laws.  For
Uniform\((.45,.55)\), aGRAPA is narrower at the terminal time; here the
regularization in \(\widehat v_t\) remains visible over a larger fraction of
the window.  Thus stitched GE-betting has the smallest terminal mean width in
eight of the nine panels, even though it was not designed to optimize
\(t=n\) alone.

As a numerical check, simultaneous coverage over all \(10{,}000\) times
occurred on respectively \(348\), \(341\), \(357\), and \(354\) of the 360
paths for stitched GE-betting, aGRAPA, PRECiSE-A-CO96, and Hedged-CS.  The
129-point aGRAPA topology scan found no disconnected acceptance set at any of
the 4,320 displayed inversions.  These finite simulation counts are
diagnostic rather than a coverage argument; Proposition~%
\ref{prop:horizon-free-wr-validity} supplies the finite-sample guarantee for
the proposed sequence.

\section{Horizon-free GE-betting confidence sequences}
\label{app:horizon-free-cs}

Section~\ref{sec:wor-bounded-update} constructs a confidence interval for a
fixed sampling horizon.  If the stopping time will instead be chosen from the
data, recomputing that interval at every time is not valid: changing the
horizon also changes the bets that would have been placed on earlier draws.
This appendix obtains a confidence sequence by fixing a portfolio of
horizon-specific accounts before sampling begins.  Each account uses the
GE-betting rule at one checkpoint, and their sum is a single e-process.

\subsection{Construction}
\label{app:horizon-free-construction}

We retain the notation of Subsection~\ref{sec:wor-background}.  In particular,
for a candidate population mean \(m\), \(m_i(m)\) and \(Y_i(m)\) are given by
\eqref{eq:wor-remaining-mean} and \eqref{eq:wor-innovation}, and candidates at
time \(t\) are restricted to \(\mathcal M_t\) in
\eqref{eq:wor-feasible-range}.  Let
\[
 u_t=\frac{t}{N-t},\qquad 1\leq t<N,
\]
and choose checkpoints \(h_0<\cdots<h_L=N-1\) approximately geometrically in
this clock.  More precisely, for \(\eta>1\), set
\begin{equation}
 h_{k+1}=\min\left\{N-1,\min\left\{h>h_k:
 \frac{h}{N-h}\geq\eta\frac{h_k}{N-h_k}\right\}\right\}.
 \label{eq:horizon-free-grid}
\end{equation}
Fix deterministic weights \(w_0,\ldots,w_L>0\) with \(\sum_kw_k=1\), and
write \(T=2/\delta\).  For each checkpoint and each of the two one-sided
tests, initialize
\[
 A_{0,k}^+(m)=A_{0,k}^-(m)=w_k.
\]

The update is the fixed-horizon rule from
Subsection~\ref{sec:wor-bounded-update}, applied separately to every
checkpoint.  While \(i\leq h_k\) and \(0<A_{i-1,k}^{\pm}(m)<T\), define
\begin{align}
 p_{i-1,k}^{\pm}(m)&=\frac{\delta}{2}A_{i-1,k}^{\pm}(m),
 &
 b_{i,k}&=\sqrt{\frac{N-h_k}
 {(N-i)(h_k-i+1)\widehat v_{i-1}}},
 \label{eq:horizon-free-scale}\\
 \lambda_{i,k}^+(m)&=\min\left\{
 \psi\{p_{i-1,k}^+(m)\}b_{i,k},\frac{c}{m_i(m)}\right\},
 &
 \lambda_{i,k}^-(m)&=\min\left\{
 \psi\{p_{i-1,k}^-(m)\}b_{i,k},\frac{c}{1-m_i(m)}\right\},
 \label{eq:horizon-free-fractions}
\end{align}
where \(\psi(p)=\phi\{\Phi^{-1}(p)\}/p\) and \(\widehat v_{i-1}\) is the
shared predictable estimator in \eqref{eq:regularized-variance-estimate}.
When \(m_i(m)=0\), the upper cap in
\eqref{eq:horizon-free-fractions} is absent; when \(m_i(m)=1\), the lower cap
is absent.  The wealth updates are
\begin{align}
 A_{i,k}^+(m)&=\min\left[T,A_{i-1,k}^+(m)
 \{1+\lambda_{i,k}^+(m)Y_i(m)\}\right],\nonumber\\
 A_{i,k}^-(m)&=\min\left[T,A_{i-1,k}^-(m)
 \{1-\lambda_{i,k}^-(m)Y_i(m)\}\right].
 \label{eq:horizon-free-updates}
\end{align}
An account is held fixed after time \(h_k\), or after its wealth reaches zero
or \(T\).  Sum the accounts within each one-sided test,
\[
 K_t^+(m)=\sum_{k=0}^L A_{t,k}^+(m),
 \qquad
 K_t^-(m)=\sum_{k=0}^L A_{t,k}^-(m),
\]
and invert them:
\begin{equation}
 J_t=\{m\in\mathcal M_t:K_t^+(m)<T,\ K_t^-(m)<T\},
 \qquad
 C_t=\bigcap_{r=1}^tJ_r.
 \label{eq:horizon-free-inversion}
\end{equation}
At \(t=N\), set \(J_N=\{S_N/N\}\) and
\(C_N=C_{N-1}\cap J_N\).  The raw sets \(J_t\) already form a confidence
sequence; the running intersections \(C_t\) are an optional nested version.

\subsection{Why the construction works}
\label{app:horizon-free-guarantee}

The bridge clock in \eqref{eq:horizon-free-grid} follows directly from the
finite-population standard error:
\begin{equation}
 \tau_{N,t}^2
 =\sigma_N^2\frac{N-t}{t(N-1)}
 =\frac{\sigma_N^2}{(N-1)u_t}.
 \label{eq:horizon-free-clock}
\end{equation}
Thus a geometric grid in \(u_t\) places checkpoints at geometric variance
scales.  At checkpoint \(h_k\), the factor \(b_{i,k}\) in
\eqref{eq:horizon-free-scale} is exactly the predictable bridge scale
\(b_{i,h_k}\) from \eqref{eq:wor-predictable-scale}.  Account \(k\) is
therefore the fixed-horizon GE-betting strategy at \(h_k\), initialized
with capital \(w_k\).  After division by \(w_k\), its rejection threshold is
\(2/(\delta w_k)\), corresponding to one-sided error \(\delta w_k/2\).
The weighted sum of these stopped e-processes is again an e-process.

\begin{proposition}[Anytime validity and interval inversion]
\label{prop:horizon-free-validity}
For every fixed population \(x_{1:N}\in[0,1]^N\),
\begin{equation}
 \mathbb P_{x_{1:N}}\{\mu_N\in J_t\text{ for every }t=1,\ldots,N\}
 \geq1-\delta.
 \label{eq:horizon-free-validity}
\end{equation}
The same guarantee holds for \((C_t)_{t=1}^N\).  Moreover, \(J_t\) and
\(C_t\) are intervals on every sample path.
\end{proposition}

\begin{proof}
At \(m=\mu_N\), \(Y_i(m)\) has conditional mean zero.  Each stopped account
in \eqref{eq:horizon-free-updates} is consequently a nonnegative test
supermartingale starting from \(w_k\).  Hence \(K_t^+(\mu_N)\) and
\(K_t^-(\mu_N)\) are nonnegative test supermartingales starting from one.
Ville's inequality and a union bound give
\[
 \mathbb P\left\{\sup_{t<N}K_t^+(\mu_N)\geq T\right\}
 +\mathbb P\left\{\sup_{t<N}K_t^-(\mu_N)\geq T\right\}
 \leq\frac{2}{T}=\delta,
\]
which proves \eqref{eq:horizon-free-validity}; taking running intersections
does not change the event that the true mean is ever excluded.

For interval geometry, the amount bet relative to the target is
\(p\psi(p)=\phi\{\Phi^{-1}(p)\}\), a nonnegative concave function of \(p\).
The scale \(b_{i,k}\) is shared across candidate means.  The argument of
Proposition~\ref{prop:shared-scale-betting-interval}, applied account by
account, shows that every upper wealth is nonincreasing in \(m\) and every
lower wealth is nondecreasing.  Summing preserves these orders, so \(J_t\) is
an interval; the intersection defining \(C_t\) is also an interval.
\end{proof}

\subsection{Experiment}
\label{app:horizon-free-experiment}

We use \(N=1000\), \(\delta=.05\), and 40 uniformly random reveal orders of
the nine deterministic populations corresponding to the distributions in
Figure~\ref{fig:original-versus-star}.  The stitched procedure uses
\(h_0=5\), \(\eta=1.3\), uniform checkpoint weights, and \(c=1\).  We compare
it with the Hedged-WoR predictable plug-in confidence sequence of
\citet{waudby2024estimating}, whose unclipped fraction is
\begin{equation}
 \lambda_i^{\rm H}
 =\sqrt{\frac{2\log(2/\delta)}
 {\widehat v_{i-1}i\log(i+1)}}.
 \label{eq:horizon-free-hedged-fraction}
\end{equation}
That benchmark uses \(c=1/2\), as specified for this implementation.  Both
methods use the same remaining-mean innovations and shared variance estimator.
The figure reports the raw intervals \(J_t\); using running intersections can
only shorten them.  Width is divided by the pointwise Gaussian reference
\(2z_{.975}\tau_{N,t}\), which is a scale reference rather than a
time-uniform interval.

\begin{figure}[!htbp]
 \centering
 \includegraphics[width=0.96\textwidth]{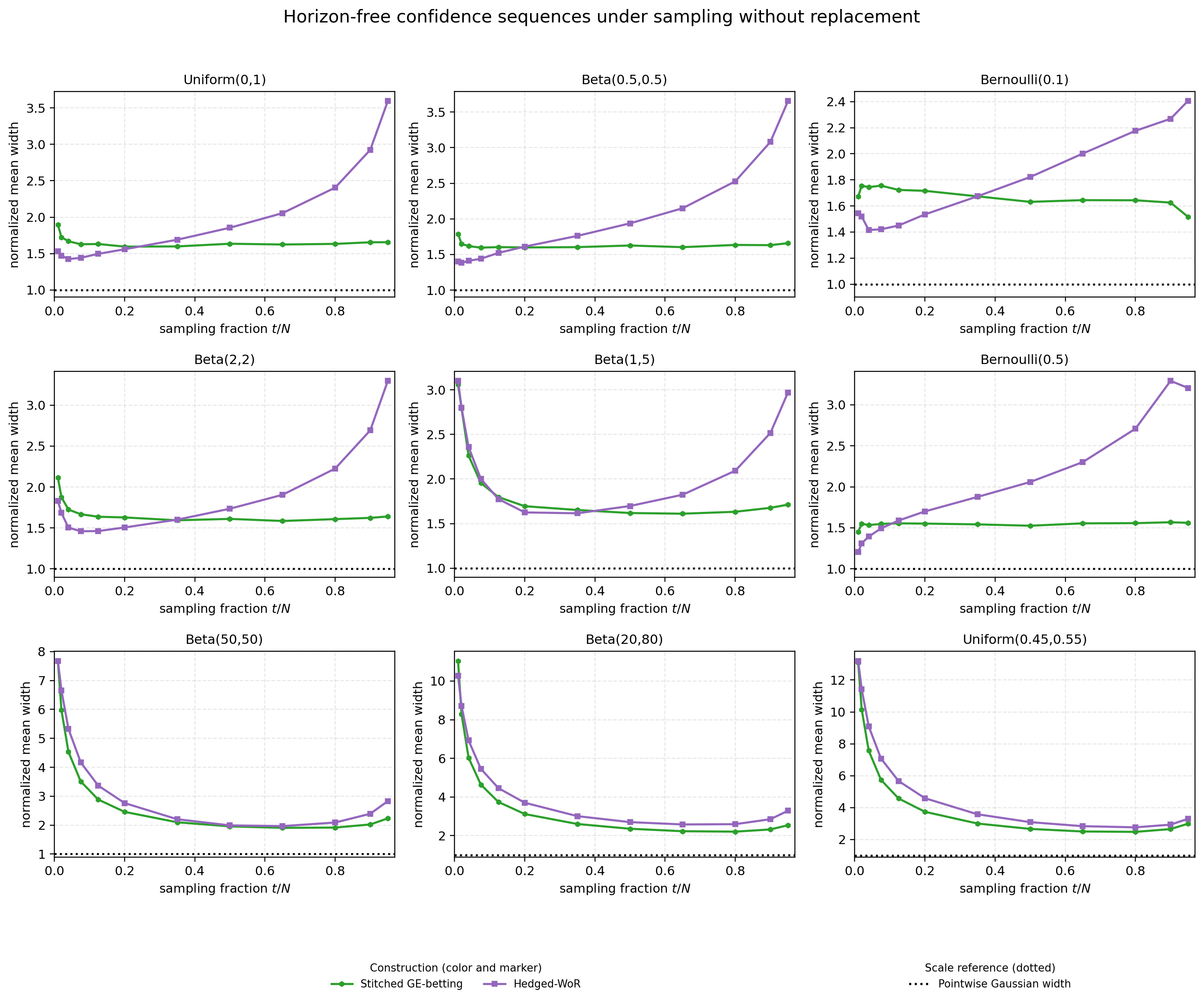}
 \caption{Normalized mean confidence-sequence width over 40 random reveal
 orders for the nine fixed populations corresponding to Figure~%
 \ref{fig:original-versus-star}.  Color and marker identify the construction,
 and the dotted line is
 the pointwise Gaussian width used only as a scale reference.  The stitched
 GE-betting procedure uses \(c=1\), whereas Hedged-WoR uses \(c=1/2\).}
 \label{fig:horizon-free-cs}
\end{figure}

Figure~\ref{fig:horizon-free-cs} shows the cost of distributing capital across
future checkpoints at the earliest sampling fractions.  Across all nine
populations, stitched GE-betting becomes narrower than
Hedged-WoR and remains narrower near the census.  Its normalized width stays
near \(1.5\)--\(1.7\) in the six standard-variance panels.  The predictable
variance regularization is more visible for the three low-variance
populations, but the stitched sequence still improves substantially on
Hedged-WoR there.

As a numerical check, simultaneous coverage over every
\(t=1,\ldots,N-1\) occurred in 354 of 360 stitched runs and 356 of 360
Hedged-WoR runs.  These counts are only diagnostic; Proposition~%
\ref{prop:horizon-free-validity} provides the finite-sample guarantee.

\section{Proofs}
\label{app:proofs}

This appendix proves the two Gaussian-efficiency theorems and the propositions
used to construct and calibrate the confidence intervals.

\subsection{Proofs of the main-text results}

\subsubsection{Proof of Theorem~\ref{thm:efficient-betting-efficiency}}
\label{proof:efficient-betting-efficiency}

\begin{proof}
The proof separates the two finite-sample claims from the iid efficiency
claim.  Step 1 uses only the conditional-mean model, whereas Steps 2--4 assume
iid sampling and identify the two endpoints on the \(n^{-1/2}\) scale.

\begin{itemize}
 \item \textit{Step 1: finite-sample validity and interval geometry.}
 Each one-sided wealth process is a test supermartingale under its
 corresponding null, so its terminal value is an e-value.  Markov's
 inequality and a union bound give coverage, while concavity of the GE-betting
 amount and the shared variance estimator make the inversion an interval.
 \item \textit{Step 2: strong preterminal tracking.}
 Before the final \(k_n\) observations, expand relative wealth on the
 normal-quantile scale.  The second-order cancellation makes this transform
 track the standardized partial sum almost surely, uniformly over a slowly
 expanding grid of local candidates.
 \item \textit{Step 3: summable terminal polarization.}
 Control the final \(k_n\) updates directly.  A positive preterminal
 separation reaches the rejection threshold exactly, while a negative one
 remains accepted, with error probabilities summable over \(n\).
 \item \textit{Step 4: from grid decisions to the interval endpoints.}
 Monotonicity converts the grid decisions into almost-sure bounds for both
 endpoints; subtracting those bounds gives the Gaussian width.
\end{itemize}

We now prove each of the four steps.

\paragraph{Step 1: finite-sample validity and interval geometry.}
At \(m=\mu\), the two centered increments in
\eqref{eq:efficient-one-sided-updates} have conditional mean zero.  Their
betting fractions are predictable, the caps keep every wealth nonnegative,
and truncating an overshoot at \(2/\delta\) can only decrease wealth.
Consequently, \(K_n^+(\mu)\) and \(K_n^-(\mu)\) are nonnegative
supermartingales starting from one.  Markov's inequality gives
\[
 \Pb\{K_n^\pm(\mu)\geq2/\delta\}
 \leq\frac{\delta}{2}\E K_n^\pm(\mu)
 \leq\frac{\delta}{2}.
\]
The true mean is excluded only if at least one of these events occurs, so a
union bound proves
\(\Pb\{\mu\in\mathcal I_n\}\geq1-\delta\).
Proposition~\ref{prop:terminal-calibration-validity} gives the same coverage
bound under uniformly randomized Markov calibration.

For interval geometry, relative wealth is \(p=(\delta/2)K\), and the
GE-betting amount as a fraction of the rejection threshold is
\[
 \sigma_{\rm eff}(p)=p\psi(p)=\phi\{\Phi^{-1}(p)\}.
\]
Writing \(q=\Phi^{-1}(p)\), we have
\[
 \sigma_{\rm eff}'(p)=-q,\qquad
 \sigma_{\rm eff}''(p)=-\frac1{\phi(q)}<0.
\]
Thus \(\sigma_{\rm eff}\) is nonnegative and concave, with
\(\sigma_{\rm eff}(0)=0\).  At every round, the scale in
\eqref{eq:efficient-one-sided-fractions} is positive, predictable, and common
to both processes and every candidate mean.
Proposition~\ref{prop:shared-scale-betting-interval} therefore implies that
\(\mathcal I_n(u_+,u_-)\) is an interval on every sample path for every fixed
\(u_+,u_-\in(0,1]\), including the deterministic choice \(u_+=u_-=1\).

\paragraph{Step 2: strong preterminal tracking.}
For the remainder of the proof, assume the observations are iid with
variance \(\sigma^2>0\).  Put \(z=z_{1-\delta/2}\) and
\[
 Z_n=\frac1{\sigma\sqrt n}\sum_{i=1}^n(X_i-\mu).
\]
For \(h\in\mathbb R\), introduce the local candidates
\[
 m_{n,\ell}(h)=\mu-\frac{h\sigma}{\sqrt n},
 \qquad
 m_{n,u}(h)=\mu+\frac{h\sigma}{\sqrt n}.
\]
For the upper- and lower-tail processes, respectively, write
\[
\begin{aligned}
 Y_{i,n}^+(h)&=X_i-m_{n,\ell}(h),
 &p_{i,n}^+(h)&=(\delta/2)K_i^+\{m_{n,\ell}(h)\},\\
 Y_{i,n}^-(h)&=m_{n,u}(h)-X_i,
 &p_{i,n}^-(h)&=(\delta/2)K_i^-\{m_{n,u}(h)\}.
\end{aligned}
\]
While a process is active, let
\(q_{i,n}^{\pm}(h)=\Phi^{-1}\{p_{i,n}^{\pm}(h)\}\), and put
\(r_{i,n}=n-i+1\).  Thus \(p=1\) means that the one-sided rejection
threshold has been reached, while \(q\to\infty\) and \(q\to-\infty\)
correspond to \(p\to1\) and \(p\to0\).

\textit{Controlling the estimator errors.} Write \(\xi_i=X_i-\mu\).  Hoeffding's inequality, a martingale Bernstein
inequality, and Borel--Cantelli give
\begin{equation}\label{eq:proof-variance-estimator-strong-rate}
 \sum_{i=1}^t\xi_i=O\{\sqrt{t\log(t+1)}\},
 \qquad
 \widehat v_t-\sigma^2
 =O\!\left\{\sqrt{\frac{\log(t+1)}{t+1}}\right\}
 \quad\text{almost surely}.
\end{equation}
For completeness, the second rate follows by putting
\(D_{i-1}=\widehat\mu_{i-1}-\mu\) and expanding
\[
 (X_i-\widehat\mu_{i-1})^2
 =\xi_i^2-2\xi_iD_{i-1}+D_{i-1}^2.
\]
The first relation in \eqref{eq:proof-variance-estimator-strong-rate} gives
\[
 D_{i-1}=O\!\left\{\sqrt{\frac{\log(i+1)}i}\right\},
 \qquad
 \sum_{i=1}^tD_{i-1}^2=O\{(\log(t+1))^2\}
 \quad\text{almost surely}.
\]
The centered sum of the \(\xi_i^2\)'s is
\(O\{\sqrt{t\log(t+1)}\}\) almost surely.  The cross term is a martingale
whose conditional variance is bounded by a constant times
\(\sum_{i\leq t}D_{i-1}^2\), so Bernstein's inequality and
Borel--Cantelli make it smaller than this first centered sum.  Dividing the
residual-square expansion by \(t+1\) proves the variance-estimator rate.
In particular, almost surely,
\begin{equation}\label{eq:proof-eventual-variance-bounds}
 \frac{\sigma^2}{2}\leq\widehat v_t\leq2
 \quad\text{for every sufficiently large \(t\)}.
\end{equation}

\textit{Taking a Taylor expansion.} Choose, only for this proof,
\[
 k_n=\left\lceil\sqrt n\{\log(n+1)\}^8\right\rceil,
 \qquad j_n=n-k_n.
\]
We next record the one-step identity that drives the argument.  Suppress
\(n,h\), and the sign, and write
\[
 r=r_{i,n},\qquad v=\widehat v_{i-1},\qquad
 q=q_{i-1,n},\qquad Y=Y_{i,n}.
\]
While the process is active and the nonnegativity cap is inactive, its
relative wealth satisfies
\[
 p_{i,n}-p_{i-1,n}
 =\phi(q)\frac{Y}{\sqrt{rv}}.
\]
Consequently, with \(a=Y/\sqrt{rv}\), a third-order Taylor expansion of
\(\Phi^{-1}\{\Phi(q)+\phi(q)a\}\) gives
\begin{equation}\label{eq:proof-normal-quantile-one-step}
 q_{i,n}=q+a+\frac12qa^2+R_{i,n},
 \qquad
 |R_{i,n}|\leq C(1+q^2)|a|^3,
\end{equation}
whenever \(|a|(1+|q|)\) is sufficiently small.  Multiply this identity by
\(\sqrt{r-1}\) and subtract \(\sqrt r\,q\).  The elementary expansions
\[
 \sqrt{1-r^{-1}}=1+O(r^{-1}),\quad
 \frac{\sqrt{r-1}}{2r}=\frac1{2\sqrt r}+O(r^{-3/2}),\quad
 \sqrt r-\sqrt{r-1}=\frac1{2\sqrt r}+O(r^{-3/2})
\]
yield
\begin{align}
 &\sqrt{r_{i,n}-1}\,q_{i,n}
   -\sqrt{r_{i,n}}\,q_{i-1,n}\nonumber\\
 &\quad=
 \frac{Y_{i,n}}{\sigma}
 +Y_{i,n}\bigl(\widehat v_{i-1}^{-1/2}-\sigma^{-1}\bigr)
 +\frac{q_{i-1,n}}{2\sqrt{r_{i,n}}}
    \left\{\frac{Y_{i,n}^2}{\widehat v_{i-1}}-1\right\}
 +e_{i,n},
 \label{eq:appendix-scaled-quantile-centered}
\end{align}
where
\begin{align}
 |e_{i,n}|
 \leq C\bigg\{&
  \frac{|Y_{i,n}|}{r_{i,n}\sqrt{\widehat v_{i-1}}}
  +\frac{|q_{i-1,n}|}{r_{i,n}^{3/2}}
       \left(1+\frac{Y_{i,n}^2}{\widehat v_{i-1}}\right)+\frac{(1+q_{i-1,n}^2)|Y_{i,n}|^3}
          {r_{i,n}\widehat v_{i-1}^{3/2}}
 \bigg\}.
 \label{eq:appendix-scaled-quantile-error-bound}
\end{align}
The centered quadratic term in
\eqref{eq:appendix-scaled-quantile-centered} is essential.  The curvature of
\(\Phi^{-1}\) produces \(qY^2/(2rv)\), while replacing \(\sqrt r\) by
\(\sqrt{r-1}\) produces \(-q/(2r)\).  Since
\(\E(Y^2\mid\mathcal F_{i-1})\) is asymptotically \(\sigma^2\) and
\(v\) estimates \(\sigma^2\), these two terms cancel to first order.  What
remains is a centered fluctuation whose cumulative contribution is
negligible.

\textit{Controlling the Taylor expansion error term within a grid.} Fix a mesh size and let \(G_n\) be a deterministic grid in
\([-\log(n+1),\log(n+1)]\).  Then \(|G_n|=O\{\log(n+1)\}\).  We will prove,
jointly over \(h\in G_n\) and both one-sided processes,
\begin{align}
\label{eq:proof-strong-preterminal-separation}
 \max_{h\in G_n}
 \left|
  \sqrt{\frac{k_n}{n}}q_{j_n,n}^+(h)-(Z_n+h-z)
 \right|&\longrightarrow0,\\
 \max_{h\in G_n}
 \left|
  \sqrt{\frac{k_n}{n}}q_{j_n,n}^-(h)-(-Z_n+h-z)
 \right|&\longrightarrow0
 \qquad\text{almost surely}. \nonumber
\end{align}

We verify the upper-tail display; reflection gives the second.  For each
\(h\in G_n\), let \(\tau_{n,h}^+\) be the first round \(i\leq j_n\) at
which
\[
 |q_{i-1,n}^+(h)|\sqrt{r_{i,n}/n}>C_0\log(n+1),
\]
or at which the unconstrained fraction is clipped or the proposed relative
wealth reaches zero or one; set \(\tau_{n,h}^+=\infty\) if no such round
exists.  All expansions and martingale estimates below are first applied to
the process stopped immediately before \(\tau_{n,h}^+\).  Define
\(\tau_{n,h}^-\) analogously for the lower-tail process.  Put
\(A_i=\widehat v_{i-1}^{-1/2}-\sigma^{-1}\).  On the probability-one event in
\eqref{eq:proof-variance-estimator-strong-rate},
\[
 A_i=O\!\left\{\sqrt{\frac{\log(i+1)}i}\right\},
 \qquad
 \sum_{i=1}^nA_i^2=O\{(\log(n+1))^2\}.
\]
The linear studentization error in
\eqref{eq:appendix-scaled-quantile-centered} is
\[
 \sum_{i\leq s}Y_{i,n}^+(h)A_i
 =\sum_{i\leq s}\xi_iA_i
   +\frac{h\sigma}{\sqrt n}\sum_{i\leq s}A_i.
\]
The first sum is a martingale transform with conditional variance
\(O\{(\log(n+1))^2\}\).  Uniformly over \(G_n\), the second sum, after
division by \(\sqrt n\), is at most
\[
 \frac{C|h|}{n}\sum_{i=1}^{j_n}|A_i|
 =O\!\left\{\frac{(\log(n+1))^{3/2}}{\sqrt n}\right\}=o(1).
\]
A maximal Bernstein inequality and a union bound over \(G_n\) therefore make
the linear error \(o(\sqrt n)\) almost surely, uniformly over its upper
summation limit.

For the quadratic term, put
\(B_i=\xi_i^2-\sigma^2\) and \(d_n=h\sigma/\sqrt n\).  The identity
\begin{equation}\label{eq:proof-quadratic-decomposition}
 \frac{Y_{i,n}^+(h)^2}{\widehat v_{i-1}}-1
 =\frac{B_i}{\widehat v_{i-1}}
  +\left(\frac{\sigma^2}{\widehat v_{i-1}}-1\right)
  +\frac{2d_n\xi_i+d_n^2}{\widehat v_{i-1}}
\end{equation}
separates a martingale difference, variance-estimation error, and local
shift.  Under the stopping rule, the martingale transform in the quadratic
term of \eqref{eq:appendix-scaled-quantile-centered}, after division by
\(\sqrt n\), has conditional variance
\(O\{(\log(n+1))^2/k_n\}\).  The variance-estimation part is bounded by
\[
 C\log(n+1)\sum_{i=1}^{j_n}
 \frac{|\widehat v_{i-1}-\sigma^2|}{r_{i,n}}
 =O\!\left\{\frac{(\log(n+1))^{5/2}}{\sqrt n}\right\}=o(1)
\]
almost surely, where the sum is split at \(n/2\) and
\eqref{eq:proof-variance-estimator-strong-rate} is used on each part.  The
two local-shift terms in \eqref{eq:proof-quadratic-decomposition} are
smaller.  Bernstein's inequality is again summable after the union bound over
\(G_n\), so the stopped quadratic sum is \(o(\sqrt n)\) almost surely.

Finally, the stopping rule and
\eqref{eq:appendix-scaled-quantile-error-bound} give
\[
 \frac1{\sqrt n}\max_{h\in G_n}
 \sum_{i=1}^{j_n}|e_{i,n}(h)|
 =O\!\left\{
  \frac{\log(n+1)}{\sqrt n}
  +\frac{\log(n+1)}{k_n}
  +\frac{\{\log(n+1)\}^2\sqrt n}{k_n}
 \right\}=o(1).
\]
Summing \eqref{eq:appendix-scaled-quantile-centered} now telescopes the
left-hand side:
\[
\begin{aligned}
 &\sum_{i=1}^{j_n}
 \left\{
  \sqrt{r_{i,n}-1}\,q_{i,n}^+(h)
  -\sqrt{r_{i,n}}\,q_{i-1,n}^+(h)
 \right\}\\
 &\qquad
 =\sqrt{k_n}\,q_{j_n,n}^+(h)-\sqrt n\,q_{0,n}
 =\frac1\sigma\sum_{i=1}^{j_n}Y_{i,n}^+(h)+o(\sqrt n),
\end{aligned}
\]
almost surely and uniformly over the grid.  Since
\(q_{0,n}=\Phi^{-1}(\delta/2)=-z\), division by \(\sqrt n\) gives
\[
 \sqrt{\frac{k_n}{n}}q_{j_n,n}^+(h)
 =-z+\frac1{\sigma\sqrt n}
     \sum_{i=1}^{j_n}Y_{i,n}^+(h)+o(1)
\]
almost surely, uniformly over the grid, as long as the process has not
stopped.

\textit{Proving that the process has not stopped.}  Let \(E_n\) be the event
that all of the finite-\(n\), stopped maximal inequalities and
variance-estimator inequalities used above hold simultaneously over
\(h\in G_n\) and both signs.  Enlarging the constants in the displayed
bounds if necessary, Bernstein's and Hoeffding's inequalities and the union
bound over \(G_n\) give
\[
 \Pb(E_n^c)\leq Cn^{-2},
 \qquad\text{and hence}\qquad
 \sum_n\Pb(E_n^c)<\infty.
\]
On \(E_n\), the same estimates hold maximally up to the corresponding
first-exit time.  Uniformly over \(h\in G_n\) and
\(i\leq j_n\wedge\tau_{n,h}^+\), they give
  \[
   \sqrt{\frac{r_{i,n}}{n}}\,|q_{i-1,n}^+(h)|
   \leq |z|+O\{\sqrt{\log(n+1)}\}+\log(n+1)+o\{\log(n+1)\}.
  \]
  Here the \(O(\sqrt{\log(n+1)})\) term comes from the centered partial sum,
  whereas the local drift is at most \(|h|\leq\log(n+1)\).  The right-hand side
  is therefore smaller than \(C_0\log(n+1)\) for all sufficiently large \(n\)
  when \(C_0\) is chosen large enough, so the normal-quantile process cannot cross the auxiliary bound
  \(C_0\log(n+1)\sqrt{n/r_{i,n}}\) before \(j_n\).

  On this region,
  \[
   |q|\leq C_0\log(n+1)\sqrt{\frac{n}{r}},
   \qquad r\geq k_n.
  \]
  The Mills bound
  \begin{equation}
   \frac{c}{1+|q|}
   \leq
   \frac{\min\{\Phi(q),1-\Phi(q)\}}{\phi(q)}
   \leq
   \frac{C}{1+|q|},
   \qquad q\in\mathbb R,
   \label{eq:proof-normal-mills-bounds}
  \end{equation}
  leads to the inverse Mills bound \(\psi\{\Phi(q)\}\leq C(1+|q|)\). The inverse Mills bound, together with the
  eventual lower bound on \(v=\widehat v_{i-1}\), consequently gives
  \[
   \frac{\psi\{\Phi(q)\}}{\sqrt{rv}}
   \leq C\left\{\frac1{\sqrt{k_n}}
         +\frac{\sqrt n\log(n+1)}{k_n}\right\}=o(1).
  \]
  Thus the unconstrained betting fraction is eventually smaller than the
  nonnegativity cap: because the local candidates converge uniformly to
  \(\mu\in(0,1)\), the caps remain bounded away from zero, whereas, up to the
  respective first-exit times, the unconstrained fractions satisfy
  \[
   \sup_{\substack{h\in G_n,\ s\in\{+,-\}\\
                    i\leq j_n\wedge\tau_{n,h}^s}}
   \frac{\psi\{\Phi(q_{i-1,n}^s(h))\}}
        {\sqrt{r_{i,n}\widehat v_{i-1}}}=o(1),
   \qquad
   \inf_{h\in G_n}
   \min\!\left\{\frac{c}{m_{n,\ell}(h)},
                \frac{c}{1-m_{n,u}(h)}\right\}
   \longrightarrow
   \min\!\left\{\frac{c}{\mu},\frac{c}{1-\mu}\right\}>0.
  \]
  Thus the minimum defining the betting fraction eventually selects the
  unconstrained fraction rather than the nonnegativity cap.  This argument
  holds for every fixed \(c\in(0,1]\).

  Moreover, since
  \(a=Y/\sqrt{rv}\) and \(Y\) is bounded, the same calculation yields
  \[
   |a|(1+|q|)
   \leq C\left\{\frac1{\sqrt{k_n}}
         +\frac{\sqrt n\log(n+1)}{k_n}\right\}=o(1).
  \]
  Hence every proposed update with
  \(i\leq j_n\wedge\tau_{n,h}^+\) eventually lies in the neighborhood on
  which the Taylor expansion is valid, and none of the stopping conditions is
  triggered.  Since
  \(|a|(1+|q|)=o(1)\) uniformly, the Mills bound~\eqref{eq:proof-normal-mills-bounds} implies that
  \[
   \frac{|\phi(q)a|}
        {\min\{\Phi(q),1-\Phi(q)\}}
   \leq C|a|(1+|q|)=o(1).
  \]
  Thus the increment \(\phi(q)a\) is eventually smaller in absolute value than
  the distance from \(p=\Phi(q)\) to either boundary.  Consequently,
  \[
   0<\Phi(q)+\phi(q)a<1,
  \]
  so the updated relative wealth reaches neither zero nor the rejection
  threshold before \(j_n\).  These strict bounds rule out every possible
  first exit in the definition of \(\tau_{n,h}^+\); the reflected argument
  does the same for \(\tau_{n,h}^-\).  Consequently,
  \[
   E_n\subseteq
   \left\{\min_{h\in G_n,\,s\in\{+,-\}}\tau_{n,h}^s>j_n\right\},
  \]
  and therefore
  \[
   \sum_n\Pb\!\left\{
    \min_{h\in G_n,\,s\in\{+,-\}}\tau_{n,h}^s\leq j_n
   \right\}<\infty.
  \]
  Thus the stopped and original processes agree through \(j_n\) for all
  sufficiently large \(n\), almost surely.

\textit{Approximating the final observations .} It remains only to replace the sum through \(j_n\) by the full standardized
sum.  Hoeffding's inequality gives, for every \(\rho>0\),
\[
 \Pb\!\left\{
  \left|\sum_{i=j_n+1}^n\xi_i\right|>\rho\sqrt n
 \right\}
 \leq2\exp\{-c_\rho n/k_n\}.
\]
These probabilities are summable, and
\(\max_{h\in G_n}k_n|h|/n\to0\).  Since
\[
 \frac1{\sigma\sqrt n}\sum_{i=1}^{j_n}Y_{i,n}^+(h)
 =Z_n+h-\frac1{\sigma\sqrt n}\sum_{i=j_n+1}^n\xi_i-\frac{k_nh}{n},
\]
Borel--Cantelli proves the first line of
\eqref{eq:proof-strong-preterminal-separation}.  Replacing \(\xi_i\) by
\(-\xi_i\) proves the second.

\paragraph{Step 3: summable terminal polarization.}
Define the terminal rejection indicators
\[
 D_n^+(h)=\mathbf1\{p_{n,n}^+(h)=1\},\qquad
 D_n^-(h)=\mathbf1\{p_{n,n}^-(h)=1\}.
\]

The goal of this step is to show that the sign of the preterminal
  normal-quantile state determines the terminal decision with summably small
  error probability.  Starting from \(q_{j_n,n}^+(h)\geq Q\), failure to reject
  can occur only if the process loses a substantial part of its positive
  separation or if no favorable final observation forces an exact threshold
  hit, whereas starting from \(q_{j_n,n}^+(h)\leq-Q\), a false rejection
  requires the process to lose its negative separation.

We control the last \(k_n\) observations directly, rather than extending
the Taylor expansion to rounds in which a single standardized observation
need not be small.  Fix \(\eta>0\) and put
\begin{equation}\label{eq:proof-terminal-separation-level}
 Q_n:=\frac{\eta}{2}\sqrt{\frac{n}{k_n}}.
\end{equation}
There exist \(C,c>0\) such that, for all sufficiently large \(n\), all
\(|h|\leq\log(n+1)\), and all \(Q\geq Q_n\), whenever
\eqref{eq:proof-eventual-variance-bounds} holds,
\begin{align}
\label{eq:proof-strong-terminal-polarization}
 \mathbf1\{q_{j_n,n}^+(h)\geq Q\}
 \Pb\{D_n^+(h)=0\mid\mathcal F_{j_n}\}
 &\leq C\exp\!\left\{-\frac{cQ^2}{\log k_n}\right\}
      +C\exp(-cQ^2),\\
 \mathbf1\{q_{j_n,n}^+(h)\leq-Q\}
 \Pb\{D_n^+(h)=1\mid\mathcal F_{j_n}\}
 &\leq C\exp\!\left\{-\frac{cQ^2}{\log k_n}\right\}. \nonumber
\end{align}

To justify the conditional calculation, first replace only the final-block
variance estimates by their truncations to \([\sigma^2/2,2]\).  The resulting
updates have deterministic upper and lower variance bounds.  Since
\(j_n\to\infty\), \eqref{eq:proof-eventual-variance-bounds} implies that the
truncated and original updates agree for all sufficiently large \(n\), almost
surely.

\textit{Proving the $C\exp\!\left\{-{cQ^2}/{\log k_n}\right\}$ term in~\eqref{eq:proof-strong-terminal-polarization}.}  The \(q\)-scale
Taylor expansion from Step 2 is not used in the final block.  Instead, put
\[
 H_+(q)=\log\{1-\Phi(q)\},\qquad
 \lambda_+(q)=\frac{\phi(q)}{1-\Phi(q)},
\]
and, for the lower tail of the \(p\)-scale,
\[
 H_-(q)=\log\Phi(q),\qquad
 \lambda_-(q)=\frac{\phi(q)}{\Phi(q)}.
\]
These transforms record the logarithmic distances of
\(p=\Phi(q)\) from its two boundaries: \(H_+(q)=\log(1-p)\) is useful when
positive \(q\) moves away from the rejection boundary at one, whereas
\(H_-(q)=\log p\) is useful when negative \(q\) moves away from the lower
boundary toward rejection.  Taking logarithms turns the additive update of
\(p\) into a relative change in the appropriate tail probability, which is
why the two inverse Mills ratios \(\lambda_+\) and \(\lambda_-\) appear below.
Write the capped betting fraction as a predictable multiple
\(\theta_i\in[0,1]\) of the unconstrained fraction and put
\(a_i=Y_{i,n}^+(h)/\sqrt{r_{i,n}\widehat v_{i-1}}\).  On paths that do
not hit the upper threshold, the first identity below is exact; for the
second, truncation at that threshold can only decrease \(H_-\).  Thus, with
boundary values interpreted by continuity,
\begin{align}
 H_+(q_i)-H_+(q_{i-1})
 &=\log\frac{1-p_i}{1-p_{i-1}}=\log\!\left\{
   1-\theta_i
   \frac{\phi(q_{i-1})}{1-\Phi(q_{i-1})}a_i
  \right\}\nonumber\\
 &=\log\{1-\theta_i\lambda_+(q_{i-1})a_i\}
 \leq-\theta_i\lambda_+(q_{i-1})a_i,
 \label{eq:proof-upper-log-tail-update}\\
 H_-(q_i)-H_-(q_{i-1})
 &=\log\frac{p_i}{p_{i-1}}\nonumber\\
 &\leq\log\!\left\{
   \frac{p_{i-1}+\theta_i\phi(q_{i-1})a_i}{p_{i-1}}
  \right\}=\log\!\left\{
   1+\theta_i
   \frac{\phi(q_{i-1})}{\Phi(q_{i-1})}a_i
  \right\}\nonumber\\
 &=\log\{1+\theta_i\lambda_-(q_{i-1})a_i\}
 \leq\theta_i\lambda_-(q_{i-1})a_i.
 \label{eq:proof-lower-log-tail-update}
\end{align}
Thus these bounds remain valid when \(|a_i|(1+|q_{i-1}|)\) is not small;
in particular, no Taylor expansion on the \(q\)-scale is being used.

The Mills bounds in \eqref{eq:proof-normal-mills-bounds} imply, 
\begin{align}
 H_+(v)
 &=\log\{1-\Phi(v)\}\nonumber\\
 &=\log\phi(v)-\log(1+v)+O(1)\nonumber\\
 &=-\frac{v^2}{2}-\log(1+v)
   -\frac12\log(2\pi)+O(1),\qquad v\to\infty,\nonumber\\
 H_-(-v)
 &=\log\Phi(-v)
  =\log\{1-\Phi(v)\}
  =H_+(v),\nonumber
\end{align}
and so, as
\(u\to\infty\),
\begin{align}
 H_+(u/2)-H_+(u)
 &=H_-(-u/2)-H_-(-u)\nonumber\\
 &=\left\{-\frac{u^2}{8}-\log(1+u/2)\right\}
   -\left\{-\frac{u^2}{2}-\log(1+u)\right\}+O(1)\nonumber\\
 &=\frac{3u^2}{8}
   +\log\!\left(\frac{1+u}{1+u/2}\right)+O(1)
  =\frac{3u^2}{8}+O(1)\geq cu^2.
 \label{eq:proof-log-tail-gap}
\end{align}
Consequently, a loss of half the magnitude of \(q\) requires an adverse
log-tail fluctuation of order \(u^2\).  Because the inverse Mills ratios grow with \(|q|\), their contribution cannot
  be bounded uniformly by a multiple of \(Q\) unless the path is localized.
  We therefore divide the possible magnitudes of \(q\) into dyadic ranges and
  prove the downcrossing bound separately on each range.

For \(u=2^sQ\), \(s\geq0\), let \(\tau_u\) be the first final-block round
whose pre-update state satisfies \(u\leq|q_{i-1}|\leq2u\), and let
\(\rho_u\) be the first subsequent exit from \(u/2\leq|q_i|\leq2u\), with
the infimum of the empty set equal to infinity.  Define the stopped-band
downcrossing event
\[
 \mathcal D_u
 =\{\tau_u\leq n,\ \rho_u\leq n,\ |q_{\rho_u}|<u/2\}.
\]
Throughout this stopped segment,
\(\lambda_+(q)\vee\lambda_-(q)\leq Cu\): indeed,
\eqref{eq:proof-normal-mills-bounds} gives
\(\lambda_+(q)\vee\lambda_-(q)\leq C(1+|q|)\), and here
\(|q|\leq2u\) with \(u\geq Q\to\infty\).  After variance truncation,
\(|a_i|\leq\sqrt{2}/\sigma\), so
\(|\Delta H_\pm|\leq Cu\) for the adverse increments on the stopped band.
Thus band-edge overshoot costs only
\(O(u)=o(u^2)\); replacing the gap \(cu^2\) by \(cu^2-O(u)\) leaves the
same exponential bound below.  Writing
\(\bar a_i=\E(a_i\mid\mathcal F_{i-1})\), the martingale parts on the
right-hand sides of \eqref{eq:proof-upper-log-tail-update}--%
\eqref{eq:proof-lower-log-tail-update} are sums of
\[
 \mp\theta_i\lambda_\pm(q_{i-1})(a_i-\bar a_i).
\]
Because \(X_i\in[0,1]\), their conditional Hoeffding variance proxy on this
excursion is at most
\[
 Cu^2\sum_{r\leq k_n}\frac1r\leq Cu^2\log k_n
\]
and their predictable drift has absolute value at most
\[
 Cu\frac{|h|}{\sqrt n}\sum_{r\leq k_n}\frac1{\sqrt r}
 \leq Cu|h|\sqrt{k_n/n}=o(u^2)
\]
uniformly for \(|h|\leq\log(n+1)\) and the levels \(u\geq Q_n\) used below.
Thus, after subtracting the negligible drift, a downcrossing requires a
martingale fluctuation of order \(u^2\).  The conditional martingale
Hoeffding inequality and \eqref{eq:proof-log-tail-gap} bound its probability
by
\[
 C\exp\!\left\{-\frac{cu^4}{u^2\log k_n}\right\}
 \leq C\exp\!\left\{-\frac{cu^2}{\log k_n}\right\}.
\]
Any path that eventually loses half of its initial separation, even if it
first moves farther from zero, must incur \(\mathcal D_u\) on at least one
of these magnitude scales.  A union bound over \(u=2^sQ\) therefore gives
\[
 \sum_{s\geq0}C\exp\!\left\{
  -\frac{c4^sQ^2}{\log k_n}\right\}
 \leq C\exp\!\left\{-\frac{c'Q^2}{\log k_n}\right\},
\]
where the final inequality follows because
\(Q^2/\log k_n\to\infty\).  Consequently,
\begin{equation}
 \Pb\{q\text{ loses half its initial magnitude}\mid\mathcal F_{j_n}\}
 \leq C\exp\!\left\{-\frac{cQ^2}{\log k_n}\right\}.
 \label{eq:proof-log-tail-downcrossing}
\end{equation}

\textit{Proving the $C\exp(-cQ^2)$ term in~\eqref{eq:proof-strong-terminal-polarization}.} It remains to ensure an exact hit of the rejection threshold when the initial
\(q\) is positive.  Put \(L=\lfloor c_0Q^2\rfloor\), where \(c_0>0\) is
sufficiently small.  If \(q\) has not lost half its magnitude, then, whenever
\(r_{i,n}\leq L\),
\[
 \frac{q_{i-1,n}^+(h)}
 {\sqrt{r_{i,n}\widehat v_{i-1}}}
 \geq\frac1{2\sqrt{2c_0}}.
\]
Since \(\sigma^2>0\), there are \(\xi,\pi>0\) such that
\(\Pb\{X_i-\mu\geq2\xi\}\geq\pi\).  Uniformly over
\(|h|\leq\log(n+1)\), this event implies
\(X_i-m_{n,\ell}(h)\geq\xi\) for all large \(n\).

At the beginning of the update, \(p=\Phi(q)\), so its remaining distance to
  the rejection target is \(1-\Phi(q)\).  On a favorable observation
  \(Y\geq\xi\), the uncapped update adds
  \(\phi(q)Y/\sqrt{r\widehat v}\).  Since
  \[
   \frac{q}{\sqrt{r\widehat v}}
   \geq\frac{1}{2\sqrt{2c_0}},
  \]
  choosing \(c_0\leq\xi^2/8\) ensures that
  \(qY/\sqrt{r\widehat v}\geq1\).  Consequently, the normal Mills bound gives
  \[
   \frac{\phi(q)Y}{\sqrt{r\widehat v}}
   =
   \frac{\phi(q)}q\,
   \frac{qY}{\sqrt{r\widehat v}}
   \geq\frac{\phi(q)}q
   \geq1-\Phi(q).
  \]
  Thus the increment is at least the entire distance from \(p=\Phi(q)\) to
  one, and the threshold truncation sets the updated relative wealth equal to
  \(p_i=1\). If the nonnegativity cap is active, the same conclusion holds for
  sufficiently large \(Q\): then \(p=\Phi(q)\) is arbitrarily close to one,
  while \(Y\geq\xi\) makes the capped multiplier strictly larger than one.

The probability that none
of the last \(L\) observations is favorable is thus at most
\((1-\pi)^L\leq C\exp(-cQ^2)\).  This proves
\eqref{eq:proof-strong-terminal-polarization}; reflection gives the same
bounds for \(D_n^-(h)\).

\textit{Combining~\eqref{eq:proof-strong-preterminal-separation} and~\eqref{eq:proof-strong-terminal-polarization}.}
  We restrict attention
  to grid candidates separated by a fixed positive margin from the Gaussian
  boundary, because this margin is converted into the diverging
  raw \(q\)-separation needed for the terminal error probabilities to be
  summable uniformly over the grid.

The fixed \(\eta\) in \eqref{eq:proof-terminal-separation-level} measures
separation from the Gaussian boundary on the
  \(Z_n+h-z\) scale. On either margin
\[
 |Z_n+h-z|\geq\eta
 \quad\text{or}\quad
 |-Z_n+h-z|\geq\eta.
\]
The quantity \(\eta\) is a fixed distance on the Gaussian-statistic scale
  \(Z_n+h-z\), which remains of order one.  By
\eqref{eq:proof-strong-preterminal-separation}, the corresponding distance on
the raw \(q\)-scale is at least \(Q_n\), and
\[
 \frac{Q_n^2}{\log k_n}
 \asymp\frac{\sqrt n}{\{\log(n+1)\}^9}.
\]

The bounds in \eqref{eq:proof-strong-terminal-polarization} are summable in
\(n\), even after a union bound over \(G_n\): since
\(|G_n|=O\{\log(n+1)\}\), for sufficiently large \(n_0\),
\[
\begin{aligned}
 &\sum_{n\geq n_0}|G_n|\left[
  C\exp\!\left\{-\frac{cQ_n^2}{\log k_n}\right\}
  +C\exp(-cQ_n^2)\right]\\
 &\qquad\leq C'\sum_{n\geq n_0}\log(n+1)
  \exp\!\left\{-\frac{c_\eta\sqrt n}{\{\log(n+1)\}^9}\right\}
  <\infty.
\end{aligned}
\]
Taking expectations in the
conditional bounds and applying Borel--Cantelli gives, almost surely for all
large \(n\),
\begin{align}
\label{eq:proof-strong-grid-decisions}
 D_n^+(h)&=\mathbf1\{Z_n+h>z\}
 &&\text{whenever \(h\in G_n\) and \(|Z_n+h-z|\geq\eta\)},\\
 D_n^-(h)&=\mathbf1\{-Z_n+h>z\}
 &&\text{whenever \(h\in G_n\) and \(|-Z_n+h-z|\geq\eta\)}. \nonumber
\end{align}

\paragraph{Step 4: from grid decisions to the interval endpoints.}
Fix \(0<\varepsilon<z\), take the grid mesh at most
\(\varepsilon/4\), and set \(\eta=\varepsilon/2\).  The law of the iterated
logarithm gives
\[
 |Z_n|=O(\sqrt{\log\log n})=o(\log(n+1))
 \qquad\text{almost surely}.
\]
Thus all local indices used below eventually lie inside the grid.
Each is of the form \(\pm Z_n\pm c\) for a fixed \(c\), so it lies within
the grid range with room for both neighboring grid points:
\[
 \max_{s,t\in\{-1,1\}}|sZ_n+tc|+\frac{\varepsilon}{4}
 \leq |Z_n|+|c|+\frac{\varepsilon}{4}<\log(n+1)
 \quad\text{eventually almost surely}.
\]
The candidate
\[
 m=\bar X_n-\frac{c\sigma}{\sqrt n}
\]
equals \(m_{n,\ell}(h_n)\) for \(h_n=c-Z_n\), and
\(Z_n+h_n-z=c-z\).  Bracket \(h_n\) by its two neighboring grid points.
The upper-tail decision is nondecreasing in \(h\), so
\eqref{eq:proof-strong-grid-decisions} gives, for every fixed \(|c-z| \geq \eta\),
\begin{equation}\label{eq:proof-strong-random-centered-decisions}
 \mathbf1\!\left\{
  K_n^+\!\left(\bar X_n-\frac{c\sigma}{\sqrt n}\right)
  \geq\frac2\delta
 \right\}
 \longrightarrow\mathbf1\{c>z\}
 \qquad\text{almost surely}.
\end{equation}
The reflected argument similarly gives
\[
 \mathbf1\!\left\{
  K_n^-\!\left(\bar X_n+\frac{c\sigma}{\sqrt n}\right)
  \geq\frac2\delta
 \right\}
 \longrightarrow\mathbf1\{c>z\}
 \qquad\text{almost surely}.
\]

Use these relations with \(c=z-\varepsilon\) and
\(c=z+\varepsilon\).  At a lower-side candidate
\(m=\bar X_n-c\sigma/\sqrt n\), the local index for the lower-tail process is
\(Z_n-c\), so its Gaussian boundary expression is
\(-Z_n+(Z_n-c)-z=-c-z<0\); hence that process accepts.  The reflected
calculation shows that the upper-tail process accepts at each upper-side
candidate.  Proposition~\ref{prop:shared-scale-betting-interval} then gives,
eventually almost surely,
\begin{align*}
 \bar X_n-\frac{\sigma(z+\varepsilon)}{\sqrt n}
 &\leq\inf\mathcal I_n
 \leq\bar X_n-\frac{\sigma(z-\varepsilon)}{\sqrt n},\\
 \bar X_n+\frac{\sigma(z-\varepsilon)}{\sqrt n}
 &\leq\sup\mathcal I_n
 \leq\bar X_n+\frac{\sigma(z+\varepsilon)}{\sqrt n}.
\end{align*}
The two inner points are accepted by both one-sided tests, so
\(\mathcal I_n\) is eventually nonempty.  Subtracting the endpoint bounds
gives
\[
 z-\varepsilon
 \leq\frac{\sqrt n}{2\sigma}\operatorname{len}(\mathcal I_n)
 \leq z+\varepsilon
\]
eventually almost surely.  Intersecting the resulting probability-one events
over positive rational \(\varepsilon\downarrow0\) proves the theorem.
\end{proof}

\subsubsection{Proof of Proposition~\ref{prop:shared-scale-betting-interval}}
\label{proof:shared-scale-betting-interval}

\begin{proof}
We first take \(m\in(0,1)\).  At \(m=0\), the upper-tail cap is defined as
\(c/m=+\infty\); at \(m=1\), the reflected lower-tail cap is defined as
\(c/(1-m)=+\infty\).  Substituting these definitions into the update gives
the same order inequalities at the endpoints.
We prove that one upper-tail update is nondecreasing in its current normalized
wealth \(p\) and nonincreasing in the candidate \(m\).

First fix \(p\) and put \(a=f(p)/s\).  Before normalized wealth reaches one,
the uncapped branch is
\[
 p\{1+a(x-m)\},\qquad m\leq c/a,
\]
whereas the capped branch is
\[
 p(1-c+cx/m),\qquad m\geq c/a.
\]
Both are nonincreasing in \(m\), and they have the same value at
\(m=c/a\).  Taking the minimum with one preserves this order.

Now fix \(m\).  Concavity of \(\sigma_f\), together with \(\sigma_f(0)=0\),
implies that \(f(p)=\sigma_f(p)/p\) is nonincreasing.  Thus the capped and
uncapped formulas again agree where their roles switch, and the capped branch
is a nonnegative multiple of \(p\).  On the uncapped branch, write
\[
 R(p)=p+\sigma_f(p)d,\qquad d=\frac{x-m}{s}.
\]
Then \(R'(p)=1+\sigma_f'(p)d\).  The two tangent bounds for a nonnegative
concave function on \([0,1]\) are
\[
 -\frac{\sigma_f(p)}{1-p}
 \leq\sigma_f'(p)\leq\frac{\sigma_f(p)}p.
\]
If \(d>0\) and \(R'(p)<0\), the left inequality forces
\(d>(1-p)/\sigma_f(p)\), whence \(R(p)>1\); the minimum with one makes the
actual update constant there.  If \(d<0\) and \(R'(p)<0\), the right
inequality forces \(-d>p/\sigma_f(p)\), whence \(R(p)<0\), which cannot occur
on a nonnegative uncapped branch.  The update is continuous in \(p\), because
the capped and uncapped formulas agree at every switch; a continuous function
that is locally nondecreasing on an interval is globally nondecreasing.
Thus the update, including the minimum with one, is nondecreasing in \(p\).

Induction over observations makes the upper-tail terminal fraction
nonincreasing in \(m\).  Reflection makes the lower-tail fraction
nondecreasing.  Each one-sided acceptance set at a fixed threshold is an
interval, and their intersection is an interval as well.
\end{proof}

\subsubsection{Proof of Proposition~\ref{prop:terminal-calibration-validity}}
\label{proof:terminal-calibration-validity}

\begin{proof}
For deterministic calibration, Markov's inequality gives
\[
 \Pb\{K_{n,A}^{\pm}(m)\geq2/\delta\}
 \leq\frac\delta2\E K_{n,A}^{\pm}(m)\leq\frac\delta2
\]
under the corresponding one-sided null.

For randomized calibration, condition on the e-value at time \(n\).
Independence and uniformity of \(U_\pm\) give
\[
 \begin{aligned}
 &\Pb\!\left\{U_\pm\leq
   \min\{(\delta/2)K_{n,A}^{\pm},1\}
   \,\middle|\,K_{n,A}^{\pm}\right\}\\
 &\qquad=\min\{(\delta/2)K_{n,A}^{\pm},1\}
 \leq\frac\delta2K_{n,A}^{\pm}.
 \end{aligned}
\]
Taking expectations again bounds the one-sided rejection probability by
\(\delta/2\).  At the point null both one-sided statements apply, so a union
bound gives two-sided coverage at least \(1-\delta\).  Finally, replacing an
accepted set by its convex hull can only enlarge it and therefore cannot
reduce coverage.
\end{proof}

\subsubsection{Proof of Theorem~\ref{thm:wor-main}}
\label{proof:wor-main}

\begin{proof}
The proof has the same structure as the proof of
Theorem~\ref{thm:efficient-betting-efficiency}.  Step 1 establishes the two
finite-sample claims for an arbitrary fixed population.  Steps 2--4 consider
the triangular array in the theorem and identify the endpoints on the
finite-population standard-error scale.

\begin{itemize}
 \item \textit{Step 1: finite-sample validity and interval geometry.}
 At the true population mean, the remaining-population residuals are
 martingale differences.  The clipped wealth processes are therefore
 e-processes, and the shared bridge scale makes their inversion an interval.
 \item \textit{Step 2: strong preterminal bridge tracking.}
 Transform relative wealth by \(\Phi^{-1}\) and stop \(k_N\) rounds before the
 horizon.  A bridge-clock identity cancels the quadratic term in the
 transformed update, leaving the standardized Doob martingale of the terminal
 sample sum.
 \item \textit{Step 3: summable terminal polarization.}
 Treat the last \(k_N\) draws directly.  A positive preterminal separation
 reaches the rejection threshold and a negative separation remains accepted,
 with error probabilities summable over \(N\).
 \item \textit{Step 4: from grid decisions to the interval endpoints.}
 Monotonicity turns the local decisions into almost-sure endpoint bounds at
 \(\bar X_n\pm z_{1-\delta/2}\tau_{N,n}\), which yield the asserted width.
\end{itemize}

We now prove each of the four steps.

\paragraph{Step 1: finite-sample validity and interval geometry.}
Fix a population \(x_{1:N}\) and let its values be revealed in uniformly
random order.  At \(m=\mu_N\), the quantity \(m_i(m)\) is the average of the
\(N-i+1\) values remaining before draw \(i\).  Consequently,
\[
 \E\{Y_i(\mu_N)\mid\mathcal F_{i-1}\}=0.
\]
The upper- and lower-tail fractions are predictable, and their caps keep the
wealth multipliers nonnegative.  Absorption at \(0\) or \(2/\delta\) can only
decrease the uncapped martingale wealth.  Thus \(K_i^+(\mu_N)\) and
\(K_i^-(\mu_N)\) are nonnegative supermartingales starting from one.  Markov's
inequality and a union bound give
\[
 \Pb_{x_{1:N}}\!\left\{K_n^+(\mu_N)\geq\frac2\delta
                 \quad\text{or}\quad
                 K_n^-(\mu_N)\geq\frac2\delta\right\}
 \leq\delta.
\]
The true population mean belongs to \(\mathcal M_n\) on every reveal order,
so this proves \eqref{eq:wor-validity}.

For uniformly randomized Markov calibration, condition on either terminal
wealth \(K\) and its independent uniform randomizer \(U\).  Since
\[
 \Pb\!\left(K\geq\frac{2U}{\delta}\,\middle|\,K\right)
 =\min\left\{\frac\delta2K,1\right\}
 \leq\frac\delta2K,
\]
each one-sided rejection probability is again at most \(\delta/2\).  The same
union bound proves randomized coverage.

It remains to verify the shape of the inversion.  Relative wealth is
\(p=(\delta/2)K\), and the amount bet as a fraction of the rejection threshold
is
\[
 p\psi(p)=\phi\{\Phi^{-1}(p)\},
\]
which is nonnegative and concave on \([0,1]\).  The scale \(b_{i,n}\) is
positive, predictable, and shared across candidate means, while
\[
 m_i(m)=\frac{Nm-S_{i-1}}{N-i+1}
\]
is increasing in \(m\).  The one-step order argument in
Proposition~\ref{prop:shared-scale-betting-interval}, applied with
\(m_i(m)\) in place of \(m\), shows inductively that \(K_n^+(m)\) is
nonincreasing and \(K_n^-(m)\) is nondecreasing.  Intersecting their two
sublevel sets with \(\mathcal M_n\) therefore gives an interval.  This remains
true for any fixed randomized thresholds.

\paragraph{Step 2: strong preterminal bridge tracking.}
Now consider the triangular array in the theorem, write \(n=n_N\), and put
\[
 z=z_{1-\delta/2},\qquad
 s_N=n\tau_{N,n}
 =\sigma_N\sqrt{\frac{n(N-n)}{N-1}},\qquad
 Z_N=\frac{S_n-n\mu_N}{s_N}.
\]
For \(h\in\mathbb R\), define the lower- and upper-side local candidates
\[
 m_{N,\ell}(h)=\mu_N-h\tau_{N,n},
 \qquad
 m_{N,u}(h)=\mu_N+h\tau_{N,n}.
\]
Let \(p_{i,N}^+(h)=(\delta/2)K_i^+\{m_{N,\ell}(h)\}\) and
\(p_{i,N}^-(h)=(\delta/2)K_i^-\{m_{N,u}(h)\}\).  While a process is active,
write \(q_{i,N}^{\pm}(h)=\Phi^{-1}\{p_{i,N}^{\pm}(h)\}\).

\textit{Setting the bridge clock.}
The deterministic bridge clock before draw \(i\) is
\begin{equation}
 r_{i,N}=\frac{(n-i+1)(N-n)}{N-i},
 \qquad i=1,\ldots,n,
 \label{eq:wor-proof-bridge-clock}
\end{equation}
and we put \(r_{n+1,N}=0\).  Recall
\(\gamma_i=(N-n)/(N-i)\).  The scale in
\eqref{eq:wor-predictable-scale} can then be written exactly as
\begin{equation}
 b_{i,n}=\frac{\gamma_i}
 {\sqrt{r_{i,N}\widehat v_{i-1}}}.
 \label{eq:wor-proof-scale-factorization}
\end{equation}
Moreover, uniformly in \(i\leq n\),
\begin{equation}
 d_{i,N}:=r_{i,N}-r_{i+1,N},
 \qquad
 \frac{d_{i,N}}{\gamma_i^2}
 =\frac{(N-n-1)(N-i)}{(N-n)(N-i-1)}
 =1+O(N^{-1}).
 \label{eq:wor-proof-clock-increment}
\end{equation}
The last identity is the discrete finite-population analogue of the variance
clock in \eqref{eq:wor-bridge-clock}.

\textit{Controlling the estimator errors.}
We first record the almost-sure estimates used below.  Let
\(v_{i-1,N}^{\rm rem}\) denote the variance of the population remaining before
draw \(i\), centered at its remaining mean.  Hoeffding bounds for random
permutations, applied to both \(x\) and \(x^2\), followed by a union bound over
the prefixes of each row, imply that, almost surely,
\begin{align}
 |Z_N|&=O(\sqrt{\log(N+1)}),
 \label{eq:wor-proof-row-concentration}\\
 \max_{(\log(N+1))^4\leq t\leq n}
 \sqrt{\frac{t}{\log(N+1)}}
 \left|\widehat v_t-\sigma_N^2\right|
 &=O(1),
 \label{eq:wor-proof-variance-concentration}\\
 \max_{i\leq n}
 \left|v_{i-1,N}^{\rm rem}-\sigma_N^2\right|
 &=O\!\left(\sqrt{\frac{\log(N+1)}N}\right).
 \label{eq:wor-proof-remaining-variance}
\end{align}
To see the middle display, expand
\(X_i-\widehat\mu_{i-1}=(X_i-\mu_N)
-(\widehat\mu_{i-1}-\mu_N)\), sum the squares, and use the prefix bounds for
the two population moments.  The cross term is handled by summation by parts,
and the squared estimation errors contribute
\(O\{\log^2(N+1)\}\).  This gives
\eqref{eq:wor-proof-variance-concentration}.  The remaining population always
has at least \(N-n\asymp N\) values, so the same bounds applied to the
complement of a prefix give \eqref{eq:wor-proof-remaining-variance}.  The
failure probabilities can be taken of order \(N^{-3}\); hence the conclusions
hold for any coupling of the rows by Borel--Cantelli.  Since
\(\sigma_N^2\to\sigma^2>0\), these estimates also give, eventually almost
surely,
\begin{equation}
 \frac{\sigma^2}{2}\leq \widehat v_t\leq2,
 \qquad (\log(N+1))^4\leq t\leq n.
 \label{eq:wor-proof-eventual-variance-bounds}
\end{equation}

The rounds before this lower variance bound becomes available are negligible.
Indeed, with \(L_N=\lceil\{\log(N+1)\}^4\rceil\), regularization gives
\(\widehat v_{i-1}\geq1/(4i)\), while \(r_{i,N}\asymp N\) for
\(i\leq L_N\).  Hence each standardized increment in
\eqref{eq:wor-proof-scale-factorization} is at most
\(C\sqrt{i/N}\), and their sum through \(L_N\) is
\(O(L_N^{3/2}/\sqrt N)=o(1)\).  We may therefore begin the uniform Taylor and
martingale estimates below at \(L_N\); the omitted initial contribution is
absorbed into their remainder.

\textit{Taking a Taylor expansion.}
Choose, only for this proof,
\[
 k_N=\left\lceil\sqrt N\{\log(N+1)\}^8\right\rceil,
 \qquad j_N=n-k_N.
\]
We next identify the transformed one-step update.  The upper-tail argument is
enough, since reflection gives the lower-tail statement.  Put
\(Y_{i,N}(h)=X_i-m_i\{m_{N,\ell}(h)\}\).  Before clipping or absorption, the
relative-wealth update and \eqref{eq:wor-proof-scale-factorization} give
\[
 p_{i,N}^+(h)=p_{i-1,N}^+(h)
 +\phi\{q_{i-1,N}^+(h)\}
   \frac{\gamma_iY_{i,N}(h)}
        {\sqrt{r_{i,N}\widehat v_{i-1}}}.
\]
Taylor expansion of \(\Phi^{-1}\), followed by multiplication by
\(\sqrt{r_{i+1,N}}\), yields
\begin{align}
 &\sqrt{r_{i+1,N}}q_{i,N}^+(h)
  -\sqrt{r_{i,N}}q_{i-1,N}^+(h)\nonumber\\
 &\quad=
 \frac{\gamma_iY_{i,N}(h)}{\sigma_N}
 +\gamma_iY_{i,N}(h)
   \left(\widehat v_{i-1}^{-1/2}-\sigma_N^{-1}\right)\nonumber\\
 &\qquad+
 \frac{q_{i-1,N}^+(h)}{2\sqrt{r_{i,N}}}
 \left\{
   \frac{\gamma_i^2Y_{i,N}(h)^2}{\widehat v_{i-1}}
   -d_{i,N}
 \right\}+e_{i,N}(h).
 \label{eq:wor-proof-scaled-quantile-update}
\end{align}
Uniformly while
\(|q_{i-1,N}^+(h)|\sqrt{r_{i,N}/r_{1,N}}\leq C_0\log(N+1)\) and
\(r_{i,N}\geq k_N/2\), boundedness of the observations and
\eqref{eq:wor-proof-eventual-variance-bounds} give
\begin{equation}
 |e_{i,N}(h)|
 \leq C\left\{
  \frac1{r_{i,N}}
  +\frac{|q_{i-1,N}^+(h)|}{r_{i,N}^{3/2}}
  +\frac{1+|q_{i-1,N}^+(h)|^2}{r_{i,N}}
 \right\}.
 \label{eq:wor-proof-expansion-error}
\end{equation}
Terms arising from replacing \(\sqrt{r_{i+1,N}/r_{i,N}}\) by one are included
in this remainder.

The centered quadratic term in
\eqref{eq:wor-proof-scaled-quantile-update} is the key point.  Conditional on
\(\mathcal F_{i-1}\),
\(\E\{Y_{i,N}(0)^2\mid\mathcal F_{i-1}\}
=v_{i-1,N}^{\rm rem}\).  Equations
\eqref{eq:wor-proof-variance-concentration}--%
\eqref{eq:wor-proof-clock-increment} therefore give
\[
 \frac{\gamma_i^2v_{i-1,N}^{\rm rem}}
      {\widehat v_{i-1}}-d_{i,N}=o(1)
\]
uniformly away from the first \((\log(N+1))^4\) rounds.  Thus the curvature
term from \(\Phi^{-1}\) cancels the decrease in the bridge clock, leaving a
centered martingale fluctuation.  This is precisely the cancellation that
would be missed by using a Brownian-motion clock.

\textit{Controlling the Taylor expansion error within a grid.}
Let \(G_N\) be a deterministic grid in
\([-\log(N+1),\log(N+1)]\) with fixed mesh, so
\(|G_N|=O\{\log(N+1)\}\).  We claim that, almost surely,
\begin{align}
 \max_{h\in G_N}\left|
  \sqrt{\frac{r_{j_N+1,N}}{r_{1,N}}}q_{j_N,N}^+(h)
  -(Z_N+h-z)\right|&\longrightarrow0,
 \label{eq:wor-proof-preterminal-tracking}\\
 \max_{h\in G_N}\left|
  \sqrt{\frac{r_{j_N+1,N}}{r_{1,N}}}q_{j_N,N}^-(h)
  -(-Z_N+h-z)\right|&\longrightarrow0.\nonumber
\end{align}

Here are the details.  For each \(h\in G_N\), let \(\tau_{N,h}^+\) be the
first round \(i\leq j_N\) at which
\[
 |q_{i-1,N}^+(h)|
 \sqrt{r_{i,N}/r_{1,N}}>C_0\log(N+1),
\]
the unconstrained fraction is clipped, or relative wealth is absorbed at zero
or one.  Set \(\tau_{N,h}^+=\infty\) if no such round exists.  All calculations
below are initially made for the process stopped immediately before
\(\tau_{N,h}^+\); define \(\tau_{N,h}^-\) by reflection.
The linear variance-estimation error in
\eqref{eq:wor-proof-scaled-quantile-update} is a martingale transform plus the
local shift
\[
 m_i(\mu_N)-m_i\{m_{N,\ell}(h)\}
 =\frac{Nh\tau_{N,n}}{N-i+1}
 =O\!\left(\frac{|h|}{\sqrt N}\right).
\]
Bernstein's inequality and
\eqref{eq:wor-proof-variance-concentration} make the cumulative linear error
\(o(\sqrt{r_{1,N}})\), uniformly over \(G_N\).  For the quadratic term, write
\[
 Y_{i,N}(0)^2-v_{i-1,N}^{\rm rem}
\]
as a bounded martingale difference.  Its predictable transform in
\eqref{eq:wor-proof-scaled-quantile-update} has conditional variance
\(O\{\log^2(N+1)/k_N\}\) after division by \(r_{1,N}\).  The predictable
remainder is \(o(1)\) by
\eqref{eq:wor-proof-variance-concentration}--%
\eqref{eq:wor-proof-clock-increment}, and the local-shift terms are smaller.
Finally, summing \eqref{eq:wor-proof-expansion-error} and dividing by
\(\sqrt{r_{1,N}}\) gives
\[
 O\!\left\{
   \frac{\log(N+1)}{\sqrt N}
  +\frac{\{\log(N+1)\}^2\sqrt N}{k_N}
 \right\}=o(1).
\]
The Bernstein bounds remain summable after a union bound over \(G_N\).

Summing \eqref{eq:wor-proof-scaled-quantile-update} now telescopes.  We have
\(q_{0,N}=\Phi^{-1}(\delta/2)=-z\) and
\[
 \sum_{i=1}^{j_N}\gamma_iY_{i,N}(h)
 =H_{j_N}^{(n)}\{m_{N,\ell}(h)\}.
\]
Therefore,
\begin{equation}
 \sqrt{\frac{r_{j_N+1,N}}{r_{1,N}}}q_{j_N,N}^+(h)
 =-z+
 \frac{H_{j_N}^{(n)}\{m_{N,\ell}(h)\}}
      {\sigma_N\sqrt{r_{1,N}}}+o(1)
 \label{eq:wor-proof-telescoped-tracking}
\end{equation}
almost surely and uniformly over the grid.

\textit{Proving that the process has not stopped.}
It remains to remove the stopping.
Let \(E_N\) collect the finite-\(N\), stopped maximal concentration and
variance events used above for all \(h\in G_N\) and both signs.  Enlarging
their constants if necessary, the bounds above and a union bound over \(G_N\)
give \(\Pb(E_N^c)\leq CN^{-2}\), and hence
\(\sum_N\Pb(E_N^c)<\infty\).  On
\(E_N\), uniformly up to the corresponding first-exit time,
\[
 \sqrt{\frac{r_{i,N}}{r_{1,N}}}\,|q_{i-1,N}^{\pm}(h)|
 \leq |z|+O\{\sqrt{\log(N+1)}\}+\log(N+1)
       +o\{\log(N+1)\},
\]
which is strictly below \(C_0\log(N+1)\) when \(C_0\) is sufficiently
large.  Moreover, \(r_{i,N}\gtrsim k_N\), the inverse Mills bound, and
boundedness of the observations give
\[
 \sup_{\substack{h\in G_N,\ s\in\{+,-\}\\
                  i\leq j_N\wedge\tau_{N,h}^s}}
 \left[
  \frac{\psi\{\Phi(q_{i-1,N}^s(h))\}\gamma_i}
       {\sqrt{r_{i,N}\widehat v_{i-1}}}
  \vee
  \frac{\gamma_i\{1+|q_{i-1,N}^s(h)|\}}
       {\sqrt{r_{i,N}\widehat v_{i-1}}}
 \right]
 \leq C\left\{
  \frac1{\sqrt{k_N}}
 +\frac{\sqrt N\log(N+1)}{k_N}
 \right\}=o(1)
\]
up to first exit.  The first term inside the supremum is the unconstrained
betting fraction, while the second bounds
\(|a|\{1+|q_{i-1,N}^{s}(h)|\}\), where \(a\) is the corresponding
signed standardized observation
\(\gamma_iY/\sqrt{r_{i,N}\widehat v_{i-1}}\), since the observations are
bounded.  Positive limiting variance forces the remaining-population means
to stay away from zero and one, and the local-candidate perturbation is
uniformly \(o(1)\); hence the solvency caps remain bounded away from zero and
are eventually inactive.

Moreover, \eqref{eq:proof-normal-mills-bounds} and the second bound above
give
\[
 \frac{|\phi(q)a|}
      {\min\{\Phi(q),1-\Phi(q)\}}
 \leq C|a|(1+|q|)=o(1).
\]
Thus the proposed increment is smaller than the distance from
\(p=\Phi(q)\) to either boundary, so the proposed relative wealth remains
strictly between zero and one.  None of the three exit conditions can
therefore be the first one to occur on \(E_N\), and
\[
 E_N\subseteq
 \left\{\min_{h\in G_N,\,s\in\{+,-\}}\tau_{N,h}^s>j_N\right\}.
\]
Consequently,
\[
 \sum_N\Pb\!\left\{
  \min_{h\in G_N,\,s\in\{+,-\}}\tau_{N,h}^s\leq j_N
 \right\}<\infty.
\]
It follows that the stopped and original bridge processes agree through
\(j_N\) for all sufficiently large \(N\), almost surely.

\textit{Approximating the final observations.}
It remains to replace the preterminal Doob martingale by its terminal value.
The final \(k_N\) martingale increments are bounded, and their conditional
variance is \(O(k_N)\).  Hence, for every \(\varepsilon>0\),
\[
 \Pb\!\left(
  \max_{h\in G_N}
  \left|H_n^{(n)}\{m_{N,\ell}(h)\}
       -H_{j_N}^{(n)}\{m_{N,\ell}(h)\}\right|
  >\varepsilon\sqrt N\right)
 \leq C\log(N+1)\exp\{-c_\varepsilon N/k_N\}.
\]
This is summable.  Since
\[
 H_n^{(n)}\{m_{N,\ell}(h)\}
 =S_n-n\mu_N+nh\tau_{N,n}=s_N(Z_N+h)
\]
and \(s_N=\sigma_N\sqrt{r_{1,N}}\), Borel--Cantelli proves the first line of
\eqref{eq:wor-proof-preterminal-tracking}.  Reflection proves the second.

\paragraph{Step 3: summable terminal polarization.}
Let \(D_N^+(h)\) and \(D_N^-(h)\) indicate that the corresponding wealth has
reached \(2/\delta\) by time \(n\).  Fix \(\eta>0\).  Because
\(r_{1,N}\asymp N\) and \(r_{j_N+1,N}\asymp k_N\), put
\[
 Q_N:=\frac\eta2
 \sqrt{\frac{r_{1,N}}{r_{j_N+1,N}}}
 \asymp\sqrt{\frac N{k_N}}.
\]
There exist \(C,c>0\) such that, for all sufficiently large \(N\), all
\(|h|\leq\log(N+1)\), and all \(Q\geq Q_N\), conditionally on
\(\mathcal F_{j_N}\),
\begin{align}
 \mathbf1\{q_{j_N,N}^+(h)\geq Q\}
 \Pb\{D_N^+(h)=0\mid\mathcal F_{j_N}\}
 &\leq C\exp\{-cQ^2/\log k_N\}+C\exp(-cQ^2),
 \label{eq:wor-proof-terminal-polarization}\\
 \mathbf1\{q_{j_N,N}^+(h)\leq-Q\}
 \Pb\{D_N^+(h)=1\mid\mathcal F_{j_N}\}
 &\leq C\exp\{-cQ^2/\log k_N\}.\nonumber
\end{align}
To justify this conditional calculation, first truncate only the final-block
variance estimates to \([\sigma^2/2,2]\).  The resulting updates have
deterministic variance bounds, while
\eqref{eq:wor-proof-eventual-variance-bounds} implies that the truncated and
original updates agree for all sufficiently large \(N\), almost surely.
The variance of each remaining population is likewise bounded above and away
from zero at time \(j_N\) by
\eqref{eq:wor-proof-remaining-variance}.  This is an
\(\mathcal F_{j_N}\)-measurable condition, and removing at most
\(k_N=o(N)\) bounded values from a remaining population of order \(N\)
preserves the same bounds along every continuation of the final block.

\textit{Proving the
\(C\exp\{-cQ^2/\log k_N\}\) term
in~\eqref{eq:wor-proof-terminal-polarization}.}
As in the iid proof, the \(q\)-scale Taylor expansion is not used in the
final block.  Instead, use
\[
 H_+(q)=\log\{1-\Phi(q)\},\qquad
 H_-(q)=\log\Phi(q),
\]
with inverse Mills ratios
\[
 \lambda_+(q)=\frac{\phi(q)}{1-\Phi(q)},\qquad
 \lambda_-(q)=\frac{\phi(q)}{\Phi(q)}.
\]
These transforms measure relative changes in the distances of
\(p=\Phi(q)\) from its two boundaries, so they remain useful when one
update is too large for a Taylor expansion on the \(q\)-scale.

Put
\[
 a_{i,N}(h)=
 \frac{\gamma_iY_{i,N}(h)}
      {\sqrt{r_{i,N}\widehat v_{i-1}}},
\]
and write the capped fraction as a predictable multiple
\(\theta_{i,N}\in[0,1]\) of the unconstrained fraction.  Suppress
\(N\) and \(h\) in \(q_{i,N}^+(h)\) in the next display.  The exact
relative-wealth update gives, with boundary values interpreted by continuity,
\begin{align}
 H_+(q_i)-H_+(q_{i-1})
 &=\log\{1-\theta_{i,N}\lambda_+(q_{i-1})a_{i,N}(h)\}
 \leq-\theta_{i,N}\lambda_+(q_{i-1})a_{i,N}(h),
 \label{eq:wor-proof-upper-log-tail-update}\\
 H_-(q_i)-H_-(q_{i-1})
 &\leq\log\{1+\theta_{i,N}\lambda_-(q_{i-1})a_{i,N}(h)\}
 \leq\theta_{i,N}\lambda_-(q_{i-1})a_{i,N}(h).
 \label{eq:wor-proof-lower-log-tail-update}
\end{align}
The first inequality is used on paths that do not hit the upper threshold;
truncation at that threshold can only decrease the left-hand side of the
second.  Neither inequality requires
\(|a_{i,N}(h)|(1+|q_{i-1}|)\) to be small.

By \eqref{eq:proof-log-tail-gap}, moving from magnitude \(u\) to \(u/2\)
requires an adverse log-tail fluctuation of order \(u^2\).  Because \(|q|\)
may first increase before returning toward zero, consider separately the
scales \(u=2^sQ\), \(s\geq0\).  Let \(\tau_{u,N}\) be the first final-block
round whose pre-update state satisfies \(u\leq|q_{i-1}|\leq2u\), and let
\(\rho_{u,N}\) be the first subsequent exit from
\(u/2\leq|q_i|\leq2u\), with the infimum of the empty set equal to infinity.
Define the stopped-band downcrossing event
\[
 \mathcal D_{u,N}
 =\{\tau_{u,N}\leq n,\ \rho_{u,N}\leq n,\
     |q_{\rho_{u,N}}|<u/2\}.
\]
On this stopped segment,
\eqref{eq:proof-normal-mills-bounds} gives
\(\lambda_+(q)\vee\lambda_-(q)\leq C(1+|q|)\leq Cu\).
After variance truncation, \(|a_{i,N}(h)|\leq\sqrt{2}/\sigma\), so
\(|\Delta H_\pm|\leq Cu\) for the adverse increments on the stopped band.
Thus band-edge overshoot costs only \(O(u)=o(u^2)\), and the required
log-tail gap remains of order \(u^2\).

Write
\(\bar a_{i,N}(h)=\E\{a_{i,N}(h)\mid\mathcal F_{i-1}\}\).
The centered parts of
\eqref{eq:wor-proof-upper-log-tail-update}--%
\eqref{eq:wor-proof-lower-log-tail-update} are martingale transforms of
bounded remaining-population draws.  Their conditional Hoeffding variance
proxy is, by \eqref{eq:wor-proof-clock-increment}, at most
\[
 Cu^2\sum_{i=j_N+1}^n\frac{\gamma_i^2}{r_{i,N}}
 \leq Cu^2\sum_{i=j_N+1}^n\frac{Cd_{i,N}}{r_{i,N}}
 \leq Cu^2\log k_N.
\]
Moreover, the local-shift identity from Step 2 and
\(\sum_{i=j_N+1}^n\gamma_i/\sqrt{r_{i,N}}\leq C\sqrt{k_N}\)
give the predictable-drift bound
\[
 \sum_{i=j_N+1}^n
  \theta_{i,N}\{\lambda_+(q_{i-1})\vee\lambda_-(q_{i-1})\}
  |\bar a_{i,N}(h)|
 \leq Cu|h|\sqrt{k_N/N}=o(u^2)
\]
uniformly for \(|h|\leq\log(N+1)\) and \(u\geq Q_N\).  A downcrossing therefore
requires a centered fluctuation of order \(u^2\).  The conditional
martingale Hoeffding inequality bounds its probability by
\[
 C\exp\!\left\{-\frac{cu^4}{u^2\log k_N}\right\}
 \leq C\exp\!\left\{-\frac{cu^2}{\log k_N}\right\}.
\]
Any path that loses half its initial separation must incur
\(\mathcal D_{u,N}\) at one of the scales \(u=2^sQ\).  Hence
\begin{equation}
 \sum_{s\geq0}C\exp\!\left\{
  -\frac{c4^sQ^2}{\log k_N}\right\}
 \leq C\exp\!\left\{-\frac{c'Q^2}{\log k_N}\right\}.
 \label{eq:wor-proof-log-tail-downcrossing}
\end{equation}

\textit{Proving the
\(C\exp(-cQ^2)\) term
in~\eqref{eq:wor-proof-terminal-polarization}.}
It remains to ensure an exact hit of the rejection threshold when the initial
\(q\) is positive.  Put \(L=\lfloor c_0Q^2\rfloor\), with \(c_0>0\)
sufficiently small.  If \(q\) has not lost half its magnitude, then during
the final \(L\) draws, \(q\geq Q/2\), \(r_{i,N}\leq L\),
\(\gamma_i\) is bounded away from zero, and
\(\widehat v_{i-1}\leq2\).  Hence
\[
 \frac{q_{i-1,N}^+(h)\gamma_i}
      {\sqrt{r_{i,N}\widehat v_{i-1}}}
 \geq\frac{c}{\sqrt{c_0}}.
\]
Because the remaining-population variance is bounded below, there are fixed
\(\xi,\pi>0\) such that, at every such round, at least a proportion \(\pi\)
of the remaining values satisfy
\(X_i-m_i(\mu_N)\geq2\xi\).  The local-candidate shift is uniformly
\(o(1)\), so these values also satisfy \(Y_{i,N}(h)\geq\xi\) for all
large \(N\).

Choose \(c_0\) so that a favorable draw has
\[
 \frac{q_{i-1,N}^+(h)\gamma_iY_{i,N}(h)}
      {\sqrt{r_{i,N}\widehat v_{i-1}}}\geq1.
\]
The normal Mills bound then gives
\[
 \frac{\phi(q)\gamma_iY_{i,N}(h)}
      {\sqrt{r_{i,N}\widehat v_{i-1}}}
 \geq\frac{\phi(q)}q
 \geq1-\Phi(q),
\]
so the update covers the entire remaining distance from \(p=\Phi(q)\) to
one and reaches \(p_i=1\).  If the solvency cap is active, the same
conclusion holds for large \(Q\), as in the iid argument.  Conditionally
at each of the last \(L\) rounds, a favorable draw has probability at least
\(\pi\); therefore,
\[
 \Pb\{\text{no favorable draw in the last }L\text{ rounds}
       \mid\mathcal F_{n-L}\}
 \leq(1-\pi)^L\leq C\exp(-cQ^2).
\]
The same bound holds after conditioning on \(\mathcal F_{j_N}\) by the
tower property.  Together with
\eqref{eq:wor-proof-log-tail-downcrossing}, this proves
\eqref{eq:wor-proof-terminal-polarization}; reflection gives the same
conclusion for the lower-tail wealth.

\textit{Combining~\eqref{eq:wor-proof-preterminal-tracking}
and~\eqref{eq:wor-proof-terminal-polarization}.}
Taking \(Q=Q_N\) in \eqref{eq:wor-proof-terminal-polarization}, we have
\[
 \frac{Q_N^2}{\log k_N}
 \asymp\frac{\sqrt N}{\{\log(N+1)\}^9}.
\]
The error bounds are summable even after a union bound over \(G_N\).  Combining
them with \eqref{eq:wor-proof-preterminal-tracking} and applying
Borel--Cantelli gives, almost surely for all sufficiently large \(N\),
\begin{align}
 D_N^+(h)&=\mathbf1\{Z_N+h>z\}
 &&\text{if \(h\in G_N\) and \(|Z_N+h-z|\geq\eta\)},
 \label{eq:wor-proof-grid-decisions}\\
 D_N^-(h)&=\mathbf1\{-Z_N+h>z\}
 &&\text{if \(h\in G_N\) and \(|-Z_N+h-z|\geq\eta\)}.\nonumber
\end{align}

\paragraph{Step 4: from grid decisions to the interval endpoints.}
Fix \(0<\varepsilon<z\), choose the grid mesh at most
\(\varepsilon/4\), and take \(\eta=\varepsilon/2\).  By
\eqref{eq:wor-proof-row-concentration},
\(|Z_N|=o\{\log(N+1)\}\) almost surely, so the random local indices below
eventually lie inside the grid.  For fixed \(a>0\),
\[
 \bar X_n-a\tau_{N,n}
 =m_{N,\ell}(a-Z_N),
 \qquad Z_N+(a-Z_N)-z=a-z.
\]
Bracketing \(a-Z_N\) by neighboring grid points and using the monotonicity
proved in Step 1 turns \eqref{eq:wor-proof-grid-decisions} into
\[
 \mathbf1\!\left\{
 K_n^+(\bar X_n-a\tau_{N,n})\geq\frac2\delta
 \right\}\longrightarrow\mathbf1\{a>z\}
 \qquad\text{almost surely}
\]
for every fixed \(|a-z|\geq\eta\).  Reflection gives the analogous statement for
\(K_n^-(\bar X_n+a\tau_{N,n})\).

The positive limiting variance implies that \(\mu_N\) stays in a compact
subset of \((0,1)\), while \(n/N\to\rho\in(0,1)\).  The four local candidates
used below therefore belong to the feasible interval \(\mathcal M_n\)
eventually almost surely.  At a lower-side candidate, the lower-tail process
accepts because its limiting boundary expression is \(-a-z<0\); the upper-tail
process likewise accepts at an upper-side candidate.  Applying the preceding
decisions with \(a=z-\varepsilon\) and \(a=z+\varepsilon\) yields
\begin{align*}
 \bar X_n-(z+\varepsilon)\tau_{N,n}
 &\leq\inf\mathcal I_{N,n}^{\rm br}
 \leq\bar X_n-(z-\varepsilon)\tau_{N,n},\\
 \bar X_n+(z-\varepsilon)\tau_{N,n}
 &\leq\sup\mathcal I_{N,n}^{\rm br}
 \leq\bar X_n+(z+\varepsilon)\tau_{N,n}
\end{align*}
eventually almost surely.  The two inner candidates are accepted by both
one-sided tests, so the interval is eventually nonempty.  Subtracting the
endpoint bounds gives
\[
 z-\varepsilon
 \leq\frac{\operatorname{len}(\mathcal I_{N,n}^{\rm br})}
              {2\tau_{N,n}}
 \leq z+\varepsilon
\]
eventually almost surely.  Intersecting over positive rational
\(\varepsilon\downarrow0\) proves \eqref{eq:wor-efficiency}.
\end{proof}

\subsubsection{Proof of Proposition~\ref{prop:ge-gaussian-ld-efficiency}}
\label{proof:ge-gaussian-ld-efficiency}

\begin{proof}
The proof separates the finite-sample guarantees from the two iid efficiency
claims. Steps 1--2 require no independence. Step 3 proves the Gaussian
endpoint expansions when \(\sigma^2>0\), whereas Steps 4--5 establish large
deviation efficiency without a positive-variance assumption.

\begin{itemize}
 \item \textit{Step 1: finite-sample validity and interval geometry.}
 The weighted wealths are test supermartingales, so Markov's inequality
 and a union bound give coverage. Monotonicity in the candidate mean
 makes their inversion an interval.
 \item \textit{Step 2: the information profile and a portfolio bound.}
 Relate the KL-inf profile to the optimal constant-fraction log wealth.
 A bound for the mixture wealths then places the accepted interval inside
 a sublevel set of the empirical information profile.
 \item \textit{Step 3: Gaussian efficiency.}
 Monotonicity in the significance level preserves the GE endpoint
 expansions at the allocated level. A local quadratic profile expansion
 and the GE terminal-polarization bounds show that the additional mixture
 wealth does not change those endpoints on the Gaussian scale.
 \item \textit{Step 4: a lower bound for every honest interval.}
 Change of measure, Borel--Cantelli, and the strong law show that every
 fixed mean strictly inside the information contour must eventually be
 included, almost surely, by any honest interval procedure.
 \item \textit{Step 5: large deviation efficiency.}
 The strong law excludes fixed means outside the information contour.
 Combining this outer bound with Step 4 and using countable bracketing
 gives almost-sure convergence of both endpoints.
\end{itemize}

We now prove each of the five steps.

\paragraph{Step 1: finite-sample validity and interval geometry.}
Under a candidate mean \(m\), each fixed-\(u\) product in
\eqref{eq:ge-ld-accounts} is a nonnegative martingale under
\eqref{eq:conditional-mean-model}; integrating preserves this property.
At \(m=0\), the upper mixture wealth is \(1/(t+1)\) if all \(t\) observations
are zero and \(+\infty\) otherwise, while \(M_t^-(0)=1\) under the
null; reflection gives the values at \(m=1\). Thus these processes remain
supermartingales under the boundary nulls, which force all observations
to equal the corresponding endpoint almost surely. The GE wealth processes are test supermartingales by
Theorem~\ref{thm:efficient-betting-efficiency}. Their deterministic convex
combination starts from one, so Markov's inequality and a union bound over
the two directions give coverage \(1-\delta\). Replacing an empty set by a
singleton cannot reduce coverage.

Proposition~\ref{prop:shared-scale-betting-interval} gives monotonicity of
the upper-tail GE wealth in \(m\), and the same monotonicity follows directly
for \(M_t^+(m)\) from its integrand. Both are nonincreasing; the reflected
wealth processes are nondecreasing. Their sums therefore have interval inversion,
as does its closure or the specified singleton replacement.

\paragraph{Step 2: the information profile and a portfolio bound.}
The bounded-mean variational identity \citep{honda2010} is
\begin{equation}\label{eq:ge-ld-duality}
 d(P,m)=\sup_{\lambda\in[-1/(1-m),1/m]}
       \E_P\log\{1+\lambda(X-m)\},\qquad 0<m<1.
\end{equation}
For completeness, write \(w_\lambda(x)=1+\lambda(x-m)\). For every \(Q\) of
mean \(m\), \(\E_Qw_\lambda=1\); the variational inequality for relative entropy
gives \(\E_P\log w_\lambda\leq\operatorname{KL}(P\|Q)\). At an interior
maximizer \(\lambda_*\) of the concave objective, its first-order condition
shows that \(Q(dx)=P(dx)/w_{\lambda_*}(x)\) is a probability measure of mean
\(m\), giving equality. If \(\lambda_*=1/m\), the one-sided condition gives
\(r=\E_P(m/X)\leq1\), and
\(Q(dx)=(m/x)P(dx)+(1-r)\delta_0(dx)\) gives equality instead. Reflection
handles the other endpoint by adding missing mass at one. A boundary
maximizer cannot have \(w_{\lambda_*}=0\) on a set of positive \(P\)-mass,
since its objective would then be \(-\infty\), whereas \(\lambda=0\) gives
zero. This proves \eqref{eq:ge-ld-duality}.

With \(\widehat P_n=n^{-1}\sum_i\delta_{X_i}\), write
\(L_n(m)=n d(\widehat P_n,m)\). For arbitrary \(e_i\geq0\), expansion in the
Bernstein basis gives
\[
 \prod_{i=1}^n(1-u+ue_i)
 =\sum_{k=0}^n\binom nk u^k(1-u)^{n-k}A_k,
 \qquad
 A_k=\frac{1}{\binom nk}\sum_{|S|=k}\prod_{i\in S}e_i.
\]
The right side is a convex combination of the \(A_k\), and its integral is
\((n+1)^{-1}\sum_kA_k\). Consequently,
\begin{equation}\label{eq:ge-ld-portfolio-bound}
 \int_0^1\prod_{i=1}^n(1-u+ue_i)du
 \geq\frac{1}{n+1}\sup_{u\in[0,1]}\prod_{i=1}^n(1-u+ue_i).
\end{equation}
Taking \(e_i=X_i/m\) and \((1-X_i)/(1-m)\) covers the two signs of
\(\lambda\) in \eqref{eq:ge-ld-duality}. Thus raw acceptance by the hedged
wealths implies
\begin{equation}\label{eq:ge-ld-profile-containment}
 L_n(m)<\log\frac{2(n+1)}{\delta w_{\rm K}}.
\end{equation}
The corresponding non-strict sublevel set is a closed interval containing
\(\bar X_n\): convexity of relative entropy in its second argument makes
\(d(\widehat P_n,m)\) convex in \(m\), and its value at \(\bar X_n\) is zero.
Closedness follows from lower semicontinuity of relative entropy and
weak compactness of the laws on \([0,1]\). The sublevel set therefore also
contains the closure of the raw acceptance interval and its singleton
replacement. Likewise, before replacement, nonnegativity
of the mixture wealths gives containment in the closure of the GE
interval at level \(a_n=\delta w_{\rm G}\).

\paragraph{Step 3: Gaussian efficiency.}
Fix \(\delta\) and write \(z=z_{1-\delta/2}\). The shared-scale normalized
GE recursion is independent of its level \(a\) apart from its initial value
\(p_0=a/2\). The proof of
Proposition~\ref{prop:shared-scale-betting-interval} shows that this
recursion is nondecreasing in current normalized wealth. Hence
\(p_{t,n}^\pm(m;a)=(a/2)K_t^\pm(m;a)\) is nondecreasing in \(a\),
and the GE acceptance interval is nonincreasing in \(a\).
Since \(a_n\to\delta\), squeezing between fixed levels on either side of
\(\delta\), using the endpoint conclusion in Step~4 of
Appendix~\ref{proof:efficient-betting-efficiency}, gives the same Gaussian
endpoint expansions at level \(a_n\).

We also need that averaging cannot move either endpoint inward on this
scale. Uniformly for each fixed \(R<\infty\),
\begin{equation}\label{eq:ge-ld-local-profile}
 \sup_{|t|\leq R}\left|
 L_n\!\left(\bar X_n+\frac{t}{\sqrt n}\right)
       -\frac{t^2}{2\sigma^2}\right|\longrightarrow0
 \quad\text{almost surely}.
\end{equation}
To see this directly, the empirical second moment about
\(\bar X_n+t/\sqrt n\) converges uniformly to \(\sigma^2>0\). The derivative
of the concave objective in \eqref{eq:ge-ld-duality}, evaluated at
\(\lambda=\pm C/\sqrt n\) for sufficiently large fixed \(C\), has opposite
signs uniformly over \(|t|\leq R\). Its maximizer is therefore
\(O(n^{-1/2})\). Taylor expansion of the bounded logarithmic factors then
gives \eqref{eq:ge-ld-local-profile}. Each integrand in
\eqref{eq:ge-ld-accounts} is at most \(\exp L_n(m)\), so both mixture
wealths are \(O(1)\) uniformly on this neighborhood, almost surely.
Because \(w_{\rm K}=O(1/n)\) at fixed \(\delta\), their weighted contributions
vanish uniformly there.

For any \(0<\varepsilon<z\), put
\(m_n^\pm=\bar X_n\pm\sigma(z-\varepsilon)/\sqrt n\).
The negative-separation argument in Step~3 of
Appendix~\ref{proof:efficient-betting-efficiency} gives more than
acceptance: \eqref{eq:proof-log-tail-downcrossing} keeps the terminal
normal-quantile state below half its diverging negative preterminal
separation. Its summable bounds, the same grid argument, and monotonicity
used in Step~4 therefore give
\[
 p_{n,n}^+(m_n^\pm;\delta)\longrightarrow0,
 \qquad p_{n,n}^-(m_n^\pm;\delta)\longrightarrow0
 \quad\text{almost surely}.
\]
Since \(a_n\leq\delta\), the same holds at \(a_n\). The normalized hedged
wealth is
\[
 \frac\delta2 H_{n,n}^\pm(m)
 =p_{n,n}^\pm(m;a_n)+\frac{\delta w_{\rm K}}2M_n^\pm(m).
\]
Thus both \(m_n^-\) and \(m_n^+\) are eventually accepted, almost surely.
The interval is then nonempty and its outer endpoints are bounded by the
GE interval at level \(a_n\). Letting \(\varepsilon\downarrow0\) proves the
two Gaussian expansions.

\paragraph{Step 4: a lower bound for every honest interval.}
Suppose \(n^{-1}\log(1/\delta_n)\to\rho>0\), and \(C_n\) has coverage
\(1-\delta_n\) for every iid law on \([0,1]\). Fix \(m\) with \(d(P,m)<\rho\).
Choose a law \(Q\) of mean \(m\) with \(D=\operatorname{KL}(P\|Q)<\rho\),
and choose \(\epsilon>0\) with \(D+\epsilon<\rho\). For
\(A_n=\{m\notin C_n\}\) and
\(R_n=\sum_{i=1}^n\log(dP/dQ)(X_i)\), change of measure gives
\[
 \sum_{n=1}^\infty P^{\otimes n}\!\left(
 A_n\cap\{R_n\leq n(D+\epsilon)\}\right)
 \leq\sum_{n=1}^\infty e^{n(D+\epsilon)}\delta_n<\infty,
\]
where summability follows from \(D+\epsilon<\rho\).
Borel--Cantelli therefore implies that these truncated exclusion events
occur only finitely often, almost surely. Finite relative entropy makes
the log likelihood ratio integrable under \(P\), so the strong law gives
\(R_n/n\to D\) almost surely. Thus \(R_n\leq n(D+\epsilon)\) eventually,
and hence \(m\in C_n\) eventually, almost surely. This proves the asserted
universal lower bound. Applying it to countable sequences of fixed points
approaching the endpoints of \(J_\rho(P)\) from within bounds the endpoints
of an honest interval from inside, almost surely. Independent external
randomization is covered by using the same randomization law under \(P\)
and \(Q\); no independence across horizons is needed.

\paragraph{Step 5: large deviation efficiency.}
For \(0<\mu<1\), the function \(d(P,m)\) is finite and continuous for
\(0<m<1\), convex in \(m\), zero only at \(\mu\), and divergent as \(m\) approaches
zero or one. Indeed, finiteness follows by mixing \(P\) with mass at zero or
one, and data processing under the kernel
\(Y\mid X\sim\operatorname{Bernoulli}(X)\) gives
\(d(P,m)\geq\operatorname{kl}(\mu,m)\). Convexity then makes it strictly
monotone on each side of \(\mu\). Thus
\(0<l_\rho(P)<\mu<u_\rho(P)<1\), with profile value \(\rho\) at both
endpoints.

Fix \(m\in(0,1)\) outside \(J_\rho(P)\). By
\eqref{eq:ge-ld-duality}, an interior fraction \(\lambda\) has expected log
return \(r>\rho\). One can choose an interior fraction by contracting any
fraction with return above \(\rho\) slightly toward zero. Its log return
is bounded, so the strong law yields
\[
 \liminf_n\frac{L_n(m)}n
 \geq\E_P\log\{1+\lambda(X-m)\}=r>\rho
 \quad\text{almost surely}.
\]
In this regime \(w_{\rm K}\to\rho/(1+\rho)>0\), and the threshold in
\eqref{eq:ge-ld-profile-containment} divided by \(n\) tends to \(\rho\).
Consequently the profile interval eventually excludes each such fixed
\(m\). It contains \(\bar X_n\to\mu\), so fixed points arbitrarily close
outside \(l_\rho(P)\) and \(u_\rho(P)\) bound its endpoints from outside.
The hedged interval is contained in that profile interval, while its own
honesty supplies the matching inner bounds just proved. Apply both bounds
to countable sequences of fixed points approaching \(l_\rho(P)\) and
\(u_\rho(P)\) from each side. On the intersection of the resulting
probability-one events, for every sufficiently small member of a sequence
\(\eta_k\downarrow0\), eventually
\begin{align*}
 l_\rho(P)-\eta_k
 &\leq\inf\mathcal I_n^{\rm H}(\delta_n)\leq l_\rho(P)+\eta_k,\\
 u_\rho(P)-\eta_k
 &\leq\sup\mathcal I_n^{\rm H}(\delta_n)\leq u_\rho(P)+\eta_k.
\end{align*}
Thus both endpoints converge almost surely. The argument uses no positive
variance assumption, so it also applies when \(P\) is a point mass at an
interior mean.
\end{proof}

\subsubsection{Proof of Proposition~\ref{prop:ge-pvalues}}
\label{proof:ge-pvalues}

\begin{proof}
For either sign, the inner minimum in \eqref{eq:ge-pvalue-maps} is a
nonnegative concave function of \(p\), vanishing at zero and one.
With a nonnegative centered increment, the clipped update is concave,
bounded above by one, and equals one at \(p=1\), so is nondecreasing.
With a negative increment, the update is convex and nonnegative by the
cap, vanishes at zero, and is therefore nondecreasing as well.
Thus each map is continuous and nondecreasing, fixing zero and one.
For every \(q\in[0,1]\), its upper level set is exactly
\([\inf\{p:F_i^\pm(p;m)\geq q\},1]\); for \(q=0\) the infimum is zero.
Applying this equivalence backward proves the stated pathwise identity,
including equality at the rejection boundary.  For each fixed \(a\),
predictability and clipping give \(\E V_n^{a,\pm}\leq a\), whence
\[
 \Pb\{\pi_n^\pm(m)\leq a\}
 =\Pb\{V_n^{a,\pm}(m)=1\}\leq\E V_n^{a,\pm}(m)\leq a.
\]
A union bound over the two tails proves two-sided validity.
\end{proof}

\end{document}